\documentclass[12pt]{article}
\usepackage{amsthm}
\usepackage{amsmath,amssymb,amsfonts,bm,amsbsy,bbm}
\usepackage{epsfig}
\usepackage{graphicx,rotating,lscape}
\usepackage{verbatim}
\usepackage{natbib}
\usepackage{multirow,color}
\usepackage{setspace}
\usepackage{bbm}
\usepackage{enumitem}
\usepackage{caption}
\usepackage{subcaption}

\renewcommand{\theequation}{\arabic{equation}}
\newtheorem{Lem}{Lemma}
\newtheorem{Th}{Theorem}

\newtheorem{Pro}{Proposition}

\newtheorem{Rem}{Remark}

\def\X{{\bf X}}
\def\x{{\bf x}}
\def\U{{\bf U}}
\def\u{{\bf u}}

\def\S{{\bf S}}

\def\a{{\bf a}}
\def\b{{\bf b}}
\def\e{{\bf e}}
\def\c{{\bf c}}

\def\p{{\bf p}}

\def\bb{{\boldsymbol\beta}}

\def\bxi{{\boldsymbol\xi}}
\def\bphi{\boldsymbol\phi}

\def\brho{\boldsymbol\rho}

\def\0{{\bf 0}}
\def\trans{^{\rm T}}
\def\pr{\hbox{pr}}
\def\wh{\widehat}
\def\wt{\widetilde}

\def\n{\nonumber}

\def\var{\hbox{var}}
\def\cov{\hbox{cov}}

\def\trace{\hbox{trace}}

\def\diag{\mbox{diag}}

\def\sumi{\sum_{i=1}^n}

\def\fyx{f_{Y\mid \X}}

\def\fx{f_{\X}}

\def\boxit#1{\vbox{\hrule\hbox{\vrule\kern6pt  \vbox{\kern6pt#1\kern6pt}\kern6pt\vrule}\hrule}}
\def\bse{\begin{eqnarray*}}
\def\ese{\end{eqnarray*}}
\def\be{\begin{eqnarray}}
\def\ee{\end{eqnarray}}
\def\bsq{\begin{equation*}}
\def\esq{\end{equation*}}
\def\bq{\begin{equation}}
\def\eq{\end{equation}}

\def\var{\hbox{var}}
\def\cov{\hbox{cov}}

\def\trace{\hbox{trace}}
\def\wh{\widehat}
\def\wt{\widetilde}
\def\eff{_{\rm eff}}

\def\mR{\mathbb{R}}
\def\n{\nonumber}

\def\cov{\mbox{cov}}

\def\diag{\mbox{diag}}

\def\sumi{\sum_{i=1}^n}

\def\trans{^{\rm T}}

\def\ba{{\boldsymbol\alpha}}

\def\bg{{\boldsymbol\gamma}}

\def\A{{\bf A}}
\def\a{{\bf a}}
\def\B{{\bf B}}
\def\C{{\bf C}}
\def\c{{\bf c}}

\def\g{{\bf g}}

\def\r{{\bf r}}

\def\h{{\bf h}}
\def\b{{\bf b}}
\def\u{{\bf u}}
\def\I{{\bf I}}
\def\M{{\bf M}}

\def\U{{\bf U}}
\def\u{{\bf u}}
\def\v{{\bf v}}
\def\W{{\bf W}}

\def\X{{\bf X}}

\def\O{{\bf O}}
\def\o{{\bf o}}
\def\R{{\bf R}}
\def\x{{\bf x}}

\def\bSig{{\bf \Sigma}}
\def\bLam{{\bf \Lambda}}

\def\diag{\hbox{diag}}

\def\trace{\hbox{trace}}

\def\diag{\hbox{diag}}
\def\log{\hbox{log}}

\def\squarebox#1{\hbox to #1{\hfill\vbox to #1{\vfill}}}
\def\btheta{{\boldsymbol \theta}}

\def\bta{{\boldsymbol \eta}}

\def\0{{\bf 0}}

\def\mx{\mathcal X}

\def\var{\hbox{var}}
\def\cov{\hbox{cov}}

\def\trace{\hbox{trace}}

\def\pr{\hbox{pr}}
\def\wh{\widehat}
\def\wt{\widetilde}
\def\diag{\hbox{diag}}
\def\log{\hbox{log}}

\def\diag{\hbox{diag}}

\definecolor{green}{rgb}{0, 0.42, 0.25}

\newcommand{\blind}{1}
\begin{document}

\def\spacingset#1{\renewcommand{\baselinestretch}%
{#1}\small\normalsize} \spacingset{1}


\if1\blind
{
  \title{\bf Semiparametric Approach to Estimation of Marginal and Quantile Effects}
  \author{Seong-ho Lee$^1$, Yanyuan Ma$^1$, and Elvezio Ronchetti$^2$\\
    $^1$Department of Statistics, Pennsylvania State University, USA \\
    $^2$Research Center for Statistics and GSEM, University of Geneva, Switzerland}
  \date{}
  \maketitle
} \fi

\if0\blind
{
  \title{\LARGE\bf Semiparametric Approach to Estimation of Marginal and Quantile Effects}
  \date{}
  \maketitle
} \fi

\bigskip
\begin{abstract}
We consider a semiparametric generalized linear model and
study estimation of both marginal and quantile effects in
this model. We propose an approximate maximum likelihood
estimator, and rigorously establish the consistency, the asymptotic
normality, and
the semiparametric efficiency of our method in both the marginal effect
and the quantile effect estimation. Simulation studies are conducted
to illustrate the finite sample performance, and we apply the new tool
to analyze a Swiss non-labor income data and discover a new interesting
predictor.
\end{abstract}

\noindent%
{\it Keywords:}  generalized linear model, 
marginal effect, quantile effect, misspecification, robustness, semiparametric efficiency.
\vfill

\newpage
\spacingset{1.9} 

%
%

\section{Introduction}

Generalized linear models (GLMs) \citep{McCullaghNelder1989} are arguably the most
frequently used models in statistical research and application. With response
variable $Y\in\mR$ and covariates $\X\in\mR^p$, GLMs
  have the familiar form 
\be\label{eq:EF}
f_{Y\mid\X}(y,\x,\bb,\phi) =
\exp\left\{\frac{y\bb\trans\x-b(\bb\trans\x)}{a(\phi)}+c(y,\phi)\right\}, 
\ee
where $\bb\in\mR^p$ is an unknown regression coefficient vector, and $\phi$ is
typically an unknown scalar parameter, if it is present. Here, 
the functions $a(\cdot), b(\cdot)$, and $c(\cdot)$ have pre-specified
forms, and various choices of these functions lead to different GLMs
extensively studied in the statistical literature.

While a large body of the classic literature is mainly concerned with the
estimation and the inference of the model parameters $\bb$ and $\phi$ in
(\ref{eq:EF}), an equally important target of research is to
estimate marginal effect, say $\bxi$, defined as 
\be\label{eq:bxidef}
\bxi \equiv E\left\{\frac{\partial E(Y\mid\X)}{\partial\X}\right\}. 
\ee
The marginal effect is a quantity of important meaning in several fields, 
in particular econometrics, and is used to 
capture the average rate of change of the regression mean with respect
to the covariates \citep{greene2000econometric}. 

A complementary, less known but important quantity is
quantile effect, defined as 
\be\label{eq:btadef}
\bta_\tau \equiv E\left\{\frac{\partial Q_\tau(Y\mid\X)}{\partial \X}\right\}
\ee
for a continuous response variable $Y$,
where $\tau\in(0,1)$ is the quantile level and
$Q_\tau(Y\mid\x)$ is the $\tau$th conditional quantile of $Y$ given
$\x$. Here, we 
  restrict our attention to the continuous response case only
because the definition of  conditional quantile in the discrete case
is
  not universally agreed upon in the literature \citep{parzen2004}.
The quantile effect measures
the average rate of change in the conditional quantiles, 
and
can be equivalently written as
$E\{\bb Q_\tau'(Y\mid\bb\trans\X)\}$,
where $Q_\tau'(Y\mid\bb\trans\x)$ is the derivative of
$Q_\tau(Y\mid\bb\trans\x)$ with respect to $\bb\trans\x$.
Obviously, our consideration of  quantile effect 
is a natural consequence of considering
quantile regression, which is a key statistical tool widely used 
as an alternative to the traditional mean regression models.
Instead of modeling the conditional mean,
the quantile regression estimates the relationship
between the covariates and the response through the conditional quantile
$Q_\tau(Y\mid\bb\trans\x)$, and considering its corresponding
rate of change directly leads to the quantile effect defined in
(\ref{eq:btadef}). Comparing (\ref{eq:bxidef}) and
(\ref{eq:btadef}), it is clear that $\bxi$ and $\bta_\tau$ are in fact
marginal mean effect and marginal quantile effect respectively, while
we shorten the names to marginal effect and quantile
effect. 

Marginal effect has direct interpretation as
  measures of risk. 
For instance, in logistic regression, especially in the medical
literature, it is customary
to misinterpret odds ratios as measures of risk. Instead,  marginal
and  quantile 
effects represent changes in the probability of the occurrence of a binary event
with respect to given risk factors and have direct interpretation 
 on the 
probability scale. This point has been stressed in the medical
literature, e.g., by 
\cite{NortonDowdMaciejewski2019}. Marginal effect also captures how
strongly the mean treatment outcome relies on different covariates on
average. This further leads to the understanding on how average treatment
effect is affected by covariates. Treatment effect is directly linked
to causal inference which is currently one of the trending topics in
Statistics, as further confirmed by the 2021 Nobel Prize in Economics
based on the seminal paper by \cite{imbensangrist1994}.
In fact,
estimation of marginal effect has been studied by many earlier works
in the statistics literature, where the name ``average mean
derivatives estimation'' was adopted. For example, 
\cite{hardle1989investigating} systematically studied  the estimation
problem, 
and \cite{newey1993efficiency} studied semiparametric efficiency
properties 
of weighted average mean derivative estimators.
Likewise, quantile effect, formerly termed as
average quantile derivatives, 
has also been proposed and studied as an alternative to the marginal
effect by \cite{chaudhuri1997average}.
Following these earlier works,
different approaches to estimating marginal effect under various
settings have been proposed, see,
e.g.,  \cite{hristache2001direct}, \cite{cattaneo2010robust}, and \cite{cattaneo2013generalized}.
The estimation of marginal effect is also quite widespread in medical
economics concerning health outcomes.
For instance, a two-year study performed at the University of
Chicago focuses on marginal effect on expenditures and length of stays
in hospitals,  rather than
on parameter estimation of  GLM; see \cite{br2005}. For the same study,
\cite{ManningBasuMullahy2005} estimate marginal effect
on inpatient 
expenditures  
using a parametric family of distributions, a three parameter generalized Gamma
distribution, which is an extension of the GLMs. 
Marginal effect is also frequently the target of study in
economics. Parametric estimation
of marginal effect in microeconometrics
concerning the female labor force participation for 872 women from
Switzerland can be found in \cite{gerfin1996parametric} and \cite{kleiber2008applied}.
We will use this example as a benchmark to compare our nonparametric estimation
in Section \ref{sec:real}.
In a GLM, the marginal effect happens to be
identical to the regression coefficient $\bb$ multiplied by 
average conditional variance, i.e. 
$\bxi  =\bb E\left\{\var(Y\mid\bb\trans\X)\right\}$.
Similarly for the quantile effect, we have
$\bta_\tau = \bb E\left\{Q_\tau'(Y\mid\bb\trans\X)\right\}$.
This allows 
for their straightforward estimation. 
Standard statistical packages provide direct estimation of marginal effect under  GLMs,
including the {\bf margins} command in Stata and the
package {\bf mfx} in R \citep{Fernihough2019}.

However, despite its central role in the statistical literature and its wide
application,  GLM has its restrictions. Indeed,
 GLM prespecifies the three functions $a(\cdot), b(\cdot)$, and
$c(\cdot)$ hence is purely parametric, therefore the model is susceptible to
model misspecification. In addition, to facilitate
computation, the choices of the three functions are often out of
convenience, which further increases the chance of mis-modeling.
To take into account the possibly large impact of deviation from the
distributional assumptions on the
parameter estimates and on the corresponding inference based on  GLMs, 
a broad stream of the literature in the past decades has focused on robust methods
for the estimation of the parameters;
see e.g. 
\cite{stefanskietal86}, \cite{kunschetal89}, and \cite{cantoniandronchetti01} for a general class of optimally bounded influence estimators and tests for the parameters 
of  GLMs. The corresponding formulas for the estimation of
marginal and  quantile effects remain the same
as the parametric case, but with the robust estimates of the parameters.

Here we go one step further in relaxing
the parametric form of  GLMs and consider a
semiparametric generalized linear model (sGLM), where we assume
\be\label{eq:sglm}
f_{Y\mid\X}(y,\x,\bb,c)
=\frac{\exp\left\{y\bb\trans\x+c(y)\right\}}
{\int\exp\left\{y\bb\trans\x+c(y)\right\}d\mu(y)},
\ee
and leave $c(\cdot)$ unspecified. 
Here and in the following text, $\mu(\cdot)$ denotes the
  Lebesgue measure for a continuous 
  variable $y$ and the counting measure for a discrete $y$.
For identifiability, we fix $c(0)=0$ and for convenience, 
we assume each component of $\x$ to be centered so 
$E(\X)=\0$. Our focus is to estimate  marginal and  quantile effects 
of (\ref{eq:sglm}), i.e. $\bxi$ defined in (\ref{eq:bxidef}) and
$\bta_\tau$ in (\ref{eq:btadef}).

Model (\ref{eq:sglm}) is not entirely new, and has been proposed and
studied in 
\cite{luo2012proportional} and \cite{huangrathouz2012}
as a semiparametric proportional likelihood
ratio model. 
\cite{ning2017likelihood} studied  high dimensional issue in 
(\ref{eq:sglm}), while \cite{lin2021} studied the estimation of the
parameter $\bb$ in  (\ref{eq:sglm}).
Different from these existing
works, our main focus is not in the model parameter $\bb$, but in the
marginal effect $\bxi$ and the quantile effect $\bta_\tau$.


In this context, we develop an approximate
likelihood 
procedure to estimate both the model parameters $\bb$ and $c(\cdot)$
in \eqref{eq:sglm}, which allows us to
subsequently estimate the marginal effect $\bxi$ and the quantile
effect $\bta_\tau$. We show that despite of
the apparent difference of the estimation variability results from the
theoretical efficiency bounds, our
estimators are in fact semiparametricly efficient for both $\bxi$ and
$\bta_\tau$. 
 We point out that our method differs from \cite{lin2021} 
  primarily in how to estimate the nonparametric component $c(y)$,
  which in turn has significant influence on estimation of the model parameters
  and  functionals. Due to the
  special structure of  model  (\ref{eq:sglm}), kernel estimation
  for $c(y)$
  is difficult to implement without giving up the model information
and resorting to model-free estimation
such as the Nadaraya-Watson estimator, 
  hence the estimation of $c(y)$ in \cite{lin2021} is not ideal. 
  This leads to efficiency loss of the $\bb$ estimation in finite samples,
  and greatly affects the efficiency of the $\bxi$ and $\bta_\tau$ estimation.
  The efficiency gain from our method is demonstrated by simulation
  experiments 
  in Section \ref{sec:sim}.

Specifically, our main contributions to the literature
are the following.
First,  we define and consider estimating
 marginal effect and quantile effect in a semiparametric model, which
have been studied by many earlier works in Statistics under various
settings and 
are receiving more attention due to the recent
interest in  mean treatment effect or mean outcome estimation. 
Second, we
discover that in this specific problem, kernel estimation
may not be a good approach to estimate the function $c(y)$ while
B-spline based estimation fits naturally and has a
clear advantage.
 This 
is not usually the case in many semiparametric estimation problems.
Third, we provide a clear and transparent estimation procedure, which is
computationally simple and enjoys the advantage of convex
optimization. 
Fourth, we establish the optimality of the resulting marginal and quantile
effect estimators as functionals of the model parameters. This
optimality is not a straightforward result to 
obtain. It requires keen statistical insight to link two seemingly unrelevant
results and the right mathematical tools to achieve it.

The rest of the article is organized as the following. In Section
\ref{sec:main}, we provide the main methodological
results of our work and describe the estimation procedures for
both $\bxi$ and $\bta_\tau$. 
In Section \ref{sec:asym},
we establish the asymptotic properties of
our estimators, compare them to the theoretical efficiency bounds, and
show that they in fact reach the efficiency bounds. 
Simulation studies are conducted in Section \ref{sec:sim} to
illustrate the finite sample performance of our method in comparison
with existing methods. An interesting dataset concerning
Swiss non-labor income is analyzed in Section \ref{sec:real},
where our nonparametric analysis reveals new insights concerning model
choice, marginal effect,  and significant covariates that would remain hidden
if using a strict parametric approach.
We conclude our paper with a discussion in Section \ref{sec:dis}. All
the technical details are relegated to the supplementary material.

\section{Methodology}\label{sec:main}

\subsection{Efficiency Bound of Marginal Effect Estimation}\label{sec:xibound}

Before engaging ourselves in the estimation  and  inference of the
marginal effect $\bxi$ associated with the model (\ref{eq:sglm}), we
would like first to understand the limit of our endeavor by
establishing the efficiency bound of the $\bxi$ estimation. 
Let $v(\bb\trans\x)\equiv
E[\{Y-E(Y\mid\bb\trans\x)\}^2\mid\bb\trans\x]$. The derivation
detailed in Section \ref{sec:derive1} leads to 
the efficient influence function as
\bse
\bphi\eff(y,\x)
= \bb v(\bb\trans\x)-\bb E\{v(\bb\trans\X)\}
+\bb y^2+\M\x y+\a(y) 
-E\{\bb Y^2+\M\x Y+\a(Y)\mid\x\},
\ese
where $\a(y)$ satisfies 
\bse
&&E[E\{\a(Y) \mid \X\}\mid y]-\a(y) \n\\
&=& 2\bb E[ yE(Y\mid\X)-
    E\{YE(Y\mid\X) \mid \X\}\mid y]+\M E[\X\{ y - E( Y \mid \X)\}\mid y].
\ese
 and 
\bse
\M&=&\left(E\{v(\bb\trans\X)\}\I
-E\left[2\bb\X\trans Y v(\bb\trans\X) +\a(Y)
\{Y-E(Y\mid\bb\trans\X)\}\X\trans 
\right]\right) \n\\
&&\times\left[E\{\X\X\trans v(\bb\trans\X)\}\right]^{-1},
\ese
Obviously, the variance of the efficient influence function,
i.e. $\var\{\bphi\eff(Y,\X)\}$, is the efficiency bound in
estimating $\bxi$.

\subsection{Efficiency Bound of Quantile Effect Estimation}\label{sec:btabound}
In Section \ref{sec:derive2}, we further derive the efficient influence function of $\bta_\tau$
under \eqref{eq:sglm}. 
Now for notational brevity, let 
$q(\nu)\equiv Q_\tau(Y\mid\nu)$, $q'(\nu)\equiv Q'_\tau(Y\mid\nu)$,
$q''(\nu)\equiv Q''_\tau(Y\mid\nu)$,
$\epsilon(\nu)\equiv\tau-I\{Y<q(\nu)\}$, 
and $\epsilon'(\nu)\equiv-\delta\{q(\nu)-Y\}q'(\nu)$.
The efficient influence function turns out to be
\bse
\bphi\eff(y,\x)
=\bb q'(\bb\trans\x)-\bb E\{q'(\bb\trans\X)\} +\M_1\x y+\a(y) - E\{\M_1\x Y+\a(Y)|\x\},
\ese
where $\a(y)$ satisfies 
\bse
E[E\{\a(Y) \mid \X\}\mid y]-\a(y)
=-\bb E\{r(y,\bb\trans\X)|y\} + \M_1 E[\X\{y-E(Y|\X)\}|y],
\ese
in which
\bse
r(Y,\nu)\equiv\frac{\epsilon(\nu) Y+\epsilon'(\nu)-\epsilon(\nu)[q(\nu)+q'(\nu)\nu+q'(\nu) c'\{q(\nu)\}]}{f\{q(\nu),\nu\}},
\ese
and
\bse
\M_1
&=&(E\{q'(\bb\trans\X)\}\I+\bb E\{\X\trans q''(\bb\trans\X)\}-E[\a(Y)\X\trans\{Y-E(Y|\X)\}])\\
&&\times[E\{\X\X\trans v(\bb\trans\X)\}]^{-1}.
\ese
Similar to the case for
estimating $\bxi$, the variance of the efficient influence function is
the efficiency bound for estimating $\bta_\tau$.

\subsection{Estimation Procedure}\label{sec:est}

We now propose an estimation procedure under the model \eqref{eq:sglm}.
We first consider the case where the response $Y$ is distributed
continuously.
In such case, we approximate $c(\cdot)$
of the conditional density $f_{Y\mid\X}(y,\x,\bb,c)$ by a B-spline
curve $\B(\cdot)\trans\bg$, 
where $\B(\cdot)\equiv\{B_1(\cdot),\dots,B_m(\cdot)\}\trans$ is a
B-spline basis vector 
and $\bg\equiv(\gamma_1,\dots,\gamma_m)\trans$ is an unknown
coefficient of the bases. 
Since we assume $c(0)=0$ for identifiability,
we ignore the first B-spline basis $B_1(\cdot)$ which corresponds to
the intercept of the curve, 
by letting $\gamma_1=0$.
Therefore, we replace $c(y)$ in \eqref{eq:sglm} by $\B(y)\trans\bg$
with $\gamma_1=0$, 
and estimate $\bb$ and $\bg$ through maximizing the approximate
loglikelihood 
\be\label{eq:logl}
l(\bb,\bg) \equiv \sumi \left[ y_i\bb\trans\x_i+\B(y_i)\trans\bg
-\log\int \exp \left\{ y\bb\trans\x_i+\B(y)\trans\bg \right\} d\mu(y) \right].
\ee
Note that the loglikelihood $l(\bb,\bg)$ is a concave
  function, hence the optimizer is unique and can be readily obtained
  using off-the-shelf convex optimizers. In our implementation, we
  used the built-in R function {\bf optim} for optimization,
  yet found the computation fast with
  satisfactory performance. For example, 
  in our real data analysis where the dataset contains $871$ observations
  and $7$ predictors,
  it takes less than 10 seconds
  to obtain the optimizer within the relative tolerance of $10^{-6}$.

Once we obtain $\wh\bb$ and $\wh c(\cdot)\equiv\B(\cdot)\trans\wh\bg$,
we can easily estimate the marginal effect $\bxi$ using \eqref{eq:bxidef} through 
\bse
\wh\bxi
&\equiv& \wh\bb\wh E\left\{\wh\var(Y\mid\bb\trans\X)\right\} \\
&=& \wh\bb \, 
n^{-1}\sumi\left[\int y^2 f_{Y\mid\X}(y,\x_i,\wh\bb,\wh c)d\mu(y)
-\left\{\int y f_{Y\mid\X}(y,\x_i,\wh\bb,\wh c)d\mu(y)\right\}^2\right].
\ese

Similarly, we estimate the quantile effect $\bta_\tau$ using \eqref{eq:btadef} through
\bse 
\wh\bta_\tau
&\equiv& \wh\bb\wh E\left\{\wh Q'_\tau(Y\mid\bb\trans\X)\right\} \\
&=& \wh\bb \,
n^{-1}\sumi
\frac{\tau\int_0^1 y f_{Y\mid\X}(y,\x_i,\wh\bb,\wh c)d\mu(y)
- \int_0^{q_i} y f_{Y\mid\X}(y,\x_i,\wh\bb,\wh c)d\mu(y)}
{f_{Y\mid\X}(q_i,\x_i,\wh\bb,\wh c)}
\bigg|_{q_i=\wh Q_\tau(Y\mid\bb\trans\x_i)},
\ese
where $\wh Q_\tau(Y\mid \bb\trans\x_i),i=1,\dots,n$,
the estimated conditional $\tau$th quantiles of $Y$ given $\x_i$,
are obtained by solving for $q_i$ from
\bse
\int_0^{q_i} f_{Y\mid\X}(y,\x_i,\wh\bb,\wh c)d\mu(y) = \tau.
\ese

When the response $Y$ is a categorical variable taking values  $\{0,...,m\}$,
then the goal is to estimate $\bb$ and $c(0),\dots,c(m)$. In this
case, (\ref{eq:sglm}) is a purely parametric model and we can proceed
with maximum likelihood estimation (MLE). An alternative way of
viewing this is that we replace $c(\cdot)$ by $\B(\cdot)\trans\bg$
where $\B(\cdot)\equiv\{I(\cdot=0),\dots,I(\cdot=m)\}\trans$ with $\gamma_1=0$,
and maximize the approximate loglikelihood given in \eqref{eq:logl}
with respect to $\bb$ and $\bg$.
Having obtained $\wh\bb$ and $\wh c(\cdot)$, we use 
\bse
\pr(Y=y|\x;\bb,c)
=\frac{\exp\{y\bb\trans\x+c(y)\}}{\sum_{y=0}^{m}\exp\{y\bb\trans\x+c(y)\}},
\ese
$y=0,\dots,m$ to estimate the marginal effect $\bxi$ by
\bse
\wh\bxi \equiv \wh\bb \,
n^{-1}\sumi\left[\sum_{y=0}^{m} y^2\, \pr(Y=y|\x_i;\wh\bb,\wh c)
-\left\{\sum_{y=0}^{m} y\, \pr(Y=y|\x_i;\wh\bb,\wh c)\right\}^2\right].
\ese
Because there is no generally accepted unique way of defining quantiles
for discrete data, we do not further study the quantile effect
estimation in this case.

\section{Theoretical Properties}\label{sec:asym}

We now establish the theoretical properties of our proposed estimators
for both the marginal effect and the quantile effect in
(\ref{eq:sglm}). 

\subsection{Continuous Response}\label{sec:asymcont}

First we analyze the properties of our estimators
under the continuous response case. We simplify
$d\mu(y)$ as $dy$ below.
To set notation, let
\bse
\fyx^*(y,\x,\bb,\bg)
&\equiv&
\frac{ \exp\{ y\bb\trans\x + \B(y)\trans\bg \} }
{ \int \exp\{ y\bb\trans\x + \B(y)\trans\bg \}dy }, \\
E^*\left\{ \g(Y) | \x,\bb,\bg \right\}
&\equiv&
\int \g(y) \fyx^*(y,\x,\bb,\bg) dy, \\
\var^*\left( Y | \x,\bb,\bg \right)
&\equiv&
E^*( Y^2 | \x,\bb,\bg ) - \left\{ E^*( Y | \x,\bb,\bg ) \right\}^2, \\
\cov^*\left\{ \g(Y), \h(Y) | \x,\bb,\bg \right\}
&\equiv&
E^*\left\{ \g(Y)\h(Y)\trans | \x,\bb,\bg \right\} \\
&&- E^*\left\{ \g(Y) | \x,\bb,\bg \right\}
E^*\left\{\h(Y)\trans | \x,\bb,\bg \right\},
\ese
where we use $E^*$ instead of $E$ to emphasize that the expectation is
computed under the approximate model, which has the same form as
(\ref{eq:sglm}), but with $c(\cdot)$ replaced by $\B\trans(\cdot)\bg$.

For real numbers $a_n$ and $b_n$,
$a_n\asymp b_n$ denotes $a_n=O(b_n)$ and $b_n=O(a_n)$ simultaneously. 
Similarly, for random variables $A_n$ and $B_n$,
$A_n\asymp_p B_n$ denotes $A_n=O_p(B_n)$ and $B_n=O_p(A_n)$ simultaneously.
For a vector $\a=(a_1,\dots,a_d)\trans\in\mR^d$,
we denote the $l_p$-norm of $\a$ as
$\|\a\|_p\equiv(|a_1|^p+\dots+|a_d|^p)^{1/p}$, $1\leq p\leq\infty$.
For a matrix $\A\in\mR^{r\times c}$,
we denote the induced $l_p$-norm of $\A$ as
$\|\A\|_p\equiv\sup_{\u\in\mR^c,\|\u\|_p=1}\|\A\u\|_p$, $1\leq p\leq\infty$.
For a function $g(\cdot)$ in the $L^2$ space, 
we denote its $L_p$-norm as
$\lVert g(\cdot) \rVert_p \equiv \{ \int_0^1 |g(y)|^p dy \}^{1/p}$.
We denote the set of the $q$th order smooth functions as
$C^q([0,1])\equiv\{g:g^{(q)}\in C([0,1])\}$.

To facilitate the theoretical derivation, 
we view the estimation procedure described in Section
\ref{sec:est} alternatively as a profile procedure. 
Specifically, treating $\bb$ as a fixed parameter,
estimate $c(\cdot)$ by the spline curve $\wh c(\cdot,\bb)\equiv\B(\cdot)\trans\wh\bg(\bb)$
via maximizing the approximate loglikelihood
\bse
l(\bb,\bg) =
\sumi \left( y_i\bb\trans\x_i+\B(y_i)\trans\bg
-\log \left[ \int_0^1 \exp \{ y\bb\trans\x_i+\B(y)\trans\bg \} dy \right] \right)
\ese
with respect to $\bg$. Then estimate $\bb$ by maximizing
\bse
\wh l(\bb) \equiv
\sumi \left( y_i\bb\trans\x_i+\wh c(y_i,\bb)
-\log \left[ \int_0^1 \exp \{ y\bb\trans\x_i+\wh c(y,\bb) \} dy \right] \right).
\ese
We point out that profiling and
performing maximization jointly with respect to $\bb$ and $\bg$ yield
the same result. 

In the following, we first aim at proving the convergence property of
$\wh c(y,\bb)$ when $\bb=\bb_0$, 
where we let $\bb_0$ denote the regression coefficient of the true
conditional density.  

We assume the following regularity conditions.

\begin{enumerate}[label=(C\arabic*),ref=(C\arabic*),start=1]
\item\label{con:bdd}
The function $c(\cdot) \in C^q([0,1])$ where $q\geq 1$ and $c(0)=0$.
The true conditional density of $Y$ given $\X$, $f_{Y|\X}(y|\x)$,
has a compact support $[0,1]$, is positive and bounded on its support.
The marginal density of $\X$, $f_{\X}(\x)$, 
has compact support $\mx$
and is bounded on its support.
\item\label{con:order}
The spline order $r\geq q$.
\item\label{con:knots}
Define the knots
$t_{-r+1}=\dots=t_0=0 < t_1 < \dots < t_N < 1 = t_{N+1} = \dots = t_{N+r}$
where N is the number of interior knots
and [0,1] is divided into $(N+1)$ subintervals.
Let $m=N+r$.
N satisfies $N \xrightarrow{}\infty$,
$N^{-q}n^{1/2} \xrightarrow{} 0$,
and $N^{-1}n(\log n)^{-1} \xrightarrow{} \infty$ as $n \xrightarrow{} \infty$.
\item\label{con:space}
Let $h_p$ be the distance between the $p$th and $(p+1)$th interior knots,
$h = \max_{r\leq p\leq N+r}h_p$ and $h' = \min_{r\leq p\leq N+r}h_p$.
There exists a constant $C_h$ such that $0<C_h<\infty$ and $h/h'<C_h$.
Therefore, $h\asymp N^{-1}$ and $h'\asymp N^{-1}$.
\item\label{con:deboor}
\citep{db1978} Under Conditions (C1)-(C4),
there exists a spline coefficient $\bg_0$
with the first component $\gamma_{01}=0$ such that
\bse
\sup_{y\in [0,1]}\left| \B(y)\trans\bg_0 - c(y) \right| = O(h^q).
\ese
\item\label{con:bb}
$\bSig_{22}\equiv E[\var\{\B(Y)|\X\}]$ is invertible
and $\|\bSig_{22}\|_2^{-1}\bSig_{22}$ has all eigenvalues bounded above $C_l$, where $C_l>0$ is a constant.
\end{enumerate}

Condition \ref{con:bdd} imposes the smoothness of the functional
component $c(\cdot)$, 
and the boundedness of the densities involved in the model, which are
standard requirements. 
The compact support requirement of the densities is also a standard
requirement in the B-spline literature. 
Condition \ref{con:order} requires that the order of the B-spline basis is
sufficiently large 
for the B-spline curve to converge to the true function fast enough.
We further assume that there are an appropriate number of interior
knots by Condition \ref{con:knots},
and that the knots are uniformly distributed in the asymptotic sense by
Condition \ref{con:space}. 
Lastly, we point out that Condition \ref{con:deboor} does not further
impose any additional requirement. It is a direct
 result given Conditions \ref{con:bdd}-\ref{con:space}, 
and is only stated as a condition for convenience.

\begin{Pro} \label{pro:bg}
Under Conditions \ref{con:bdd}-\ref{con:bb},
$\|\wh\bg(\bb_0)-\bg_0\|_2=O_p\{(nh)^{-1/2}\}$ 
and
\bse
\wh\bg(\bb_0)-\bg_0
&=& \bSig_{22}^{-1}n^{-1}\sumi[\B(y_i)-E\{\B(Y)|\x_i\}] + \r_1,
\ese
where $\|\r_1\|_2=o_p\{(nh)^{-1/2}\}$. Furthermore,
\bse
\sup_{y\in[0,1]}|\wh c(y,\bb_0)-c(y)| = O_p\{(nh)^{-1/2}+h^q\}.
\ese
\end{Pro}

Proposition \ref{pro:bg} states that given the true regression
coefficient $\bb_0$, 
$\wh\bg$ converges to the B-spline basis coefficient $\bg_0$
at the nonparametric convergence rate $(nh)^{-1/2}$,
hence the B-spline curve approximates $c(\cdot)$.
We further obtain that
the estimator $\wh c(\cdot,\bb_0)$ converges to the true function
$c(\cdot)$ at the $(nh)^{-1/2}$ rate as well. 
Based on Proposition \ref{pro:bg},
below we further establish the asymptotic property 
of the estimator $\wh\bb$. Note $\bb$ is not our
research interest, hence this property is stated as a by-product.
We first impose one additional regularity condition, which is a
standard requirement.

\begin{enumerate}[label=(C\arabic*),ref=(C\arabic*),start=7]
\item\label{con:bSig}
The expectation of the conditional covariance of $\{\X\trans Y,
\B\trans(Y)\}\trans$ given $\X$, i.e. 
\bse
\bSig
\equiv \begin{bmatrix}
\bSig_{11} & \bSig_{12} \\ \bSig_{21} & \bSig_{22}
\end{bmatrix}
\equiv E\begin{bmatrix}
\X\X\trans \var(Y|\X) & \X\cov\{Y,\B(Y)|\X\} \\
\cov\{\B(Y),Y|\X\}\X\trans & \var\{\B(Y)|\X\}
\end{bmatrix},
\ese
is invertible.
\end{enumerate}

It is easy to see that
the conditional covariance of $\{\X\trans Y, \B\trans(Y)\}\trans$
given $\X$ is positive semidefinite. 
Thus, $\bSig$ is also positive semidefinite.
Condition \ref{con:bSig} further requires that $\bSig$ is positive
definite, which is very mild. 
Note that this guarantees
$\bSig_{11}$ and the Schur complement of $\bSig_{22}$,
$\bSig^*\equiv \bSig_{11} - \bSig_{12}\bSig_{22}^{-1}\bSig_{21}$,
are both positive definite.

\begin{Pro} \label{pro:bb}
Under Conditions \ref{con:bdd}-\ref{con:bSig},
$\|\wh\bb-\bb_0\|_2=O_p(n^{-1/2})$ and
\bse
\wh\bb-\bb_0
&=& \bSig^{*-1}n^{-1}\sumi\x_i\{y_i-E(Y|\x_i)\} \\
&&-
\bSig^{*-1}\bSig_{12}\bSig_{22}^{-1}n^{-1}\sumi[\B(y_i)-E\{\B(Y)|\x_i\}]
+ \r_2,
\ese
where $\|\r_2\|_2=o_p(n^{-1/2})$.
Furthermore,
\bse
\bSig^{*1/2}\sqrt{n}(\wh\bb-\bb_0)\to N(\0_p,\I_p)
\ese
in distribution as $n\to\infty$.
\end{Pro}

Proposition \ref{pro:bb} establishes how the regression coefficient estimator
$\wh\bb$ is asymptotically distributed, 
and allows to perform inference
since the asymptotic variance of $\wh\bb$ is estimable based on
$\wh\bb$ and $f_{Y\mid\X}(y,\x,\wh\bb,\wh c)$. 
Below, we study how the estimator of the marginal effect $\bxi$
is asymptotically distributed in Theorem \ref{th:xi}.
We point out that the asymptotic variance of $\wh\bxi$ is estimable as well
so that we can perform inference on the marginal effect.

\begin{Th}\label{th:xi}
Let $\bSig_\bxi\equiv
\A\bSig^{-1}\A\trans+\bb_0\bb_0\trans\var\{\var(Y|\X)\}$, 
where $\A\equiv[\A_1, \A_2]$ and
\bse
\A_1 &\equiv& E\{\var(Y|\X)\}\I + \bb_0E[\{Y-E(Y|\X)\}^3\X\trans], \\
\A_2 &\equiv& \bb_0 E(\{Y-E(Y|\X)\}^2[\B(Y)-E\{\B(Y)|\X\}])\trans.
\ese
Under Conditions \ref{con:bdd}-\ref{con:bSig},
\bse
\bSig_\bxi^{-1/2}\sqrt{n}(\wh\bxi-\bxi_0)\to N(\0_p, \I_p)
\ese
in distribution as $n\to\infty$.
\end{Th}

To facilitate the analysis of the properties of the quantile effect
estimator $\wh\bta_\tau$, 
we assume two additional regularity conditions.

\begin{enumerate}[label=(C\arabic*),ref=(C\arabic*),start=8]
\item\label{con:q}
$c''(\cdot)$ is bounded on $[0,1]$.
\item\label{con:order2}
The spline order $r\geq2$.
\end{enumerate}

Note that for an arbitrary twice differentiable function $g(\cdot)$,
$\|g'(\cdot)\|_\infty\le2(\|g(\cdot)\|_\infty\|g''(\cdot)\|_\infty)^{1/2}$
by the Landau-Kolmogorov inequality.
Setting $g(\cdot)=\B(\cdot)\trans\bg_0-c(\cdot)$,
Conditions \ref{con:deboor}, \ref{con:q}, and \ref{con:order2} guarantee
$\|\B'(\cdot)\trans\bg_0-c'(\cdot)\|_\infty=O(h^{q/2})$, i.e.
$\B'(\cdot)\trans\bg_0$ converges to $c'(\cdot)$ uniformly at the
rate $O(h^{q/2})$. 
Theorem \ref{th:bta} below provides
the asymptotic normality of the quantile effect estimator $\wh\bta_\tau$.

\begin{Th}\label{th:bta}
Let $\bSig_{\bta_\tau}\equiv\C\bSig^{-1}\C\trans+\bb_0\bb_0\trans\var\{q'(\X\trans\bb_0)\}$,
where $\C\equiv[\C_1,\C_2]$ and
\bse
\C_1&\equiv&E\{q'(\X\trans\bb_0)\}\I+\bb_0E\Bigg[\X\trans\Bigg\{\frac{ E([\tau-I\{Y\leq q(\X\trans\bb_0)\}]Y^2|\X)}{\fyx\{q(\X\trans\bb_0)|\X\}}\n\\
&&-2q'(\X\trans\bb_0)q(\X\trans\bb_0)-\{q'(\X\trans\bb_0)\}^2[\X\trans\bb_0+c'\{q(\X\trans\bb_0)\}]\}],\\
\C_2&\equiv&\bb_0E\Bigg\{\frac{E([\tau-I\{Y\leq q(\X\trans\bb_0)\}]Y\B\trans(Y)|\X)}{\fyx\{q(\X\trans\bb_0)|\X\}}\\
&&-q'(\X\trans\bb_0)\B\trans\{q(\X\trans\bb_0)\}
-\frac{E([\tau-I\{Y\leq q(\X\trans\bb_0)\}]\B\trans(Y)|\X)}{\fyx\{q(\X\trans\bb_0)|\X\}}\\
&&\times[q(\X\trans\bb_0)+q'(\X\trans\bb_0)\X\trans\bb_0+q'(\X\trans\bb_0)c'\{q(\X\trans\bb_0)\}]\}.
\ese
Under Conditions \ref{con:bdd}-\ref{con:order2},
\bse
\bSig_{\bta_\tau}^{-1/2}\sqrt{n}(\wh\bta_\tau-\bta_{\tau0})\to N(\0_p, \I_p)
\ese
in distribution as $n\to\infty$.
\end{Th}

The asymptotic properties we established, 
especially the estimation variances of $\wh\bxi$ and $\wh\bta_\tau$,
have very different forms from the efficiency bounds we derived in
Sections \ref{sec:xibound} and \ref{sec:btabound}. Nevertheless, 
closer inspection, together with some basic but less frequently
adopted linear algebra tools reveal that these two sets of results
have much closer connections, and  
the efficiency bounds in estimating both $\bxi$ and $\bta_\tau$ are
actually reached by  our B-spline based approximate maximum likelihood
estimators. As a by-product, 
we also state the efficiency property of $\wh\bb$ as a proposition,
even though our interest is not in $\bb$.

\begin{Pro}\label{pro:effbb}
Under Conditions \ref{con:bdd}-\ref{con:bSig}, 
the approximate maximum likelihood estimator $\wh\bb$ is efficient.  
\end{Pro}

\begin{Th}\label{th:effxi}
Under Conditions \ref{con:bdd}-\ref{con:bSig}, 
the estimator $\wh\bxi$ based on 
the approximate maximum likelihood estimator $\wh\bb$ and $\wh
c(\cdot)$ is efficient.   
\end{Th}

\begin{Th}\label{th:effbta}
Under Conditions \ref{con:bdd}-\ref{con:order2}, 
the estimator $\wh\bta_\tau$ based on 
the approximate maximum likelihood estimator $\wh\bb$ and $\wh
c(\cdot)$ is efficient.
\end{Th}

\subsection{Discrete Response}\label{sec:asymcat}

We now analyze the properties of our estimators
under the discrete response case.
For notational simplicity, we denote
$\btheta\equiv(\bb\trans,\bg\trans)\trans$, 
$\p(\x,\btheta)\equiv\{\pr(Y=1\mid\x;\bb,\c),\dots,\pr(Y=m\mid\x;\bb,\c)\}\trans$, 
and $\p_k(\x,\btheta)\equiv\{1^k\pr(Y=1\mid\x;\bb,\c),\dots,m^k
\pr(Y=m\mid\x;\bb,\c)\}\trans$. 
Also in this section, we use $\B(y)$ to denote a vector of indicator functions,
i.e., $\B(y)\equiv\{I(y=1),\dots,I(y=m)\}\trans$. Note that we allow
$m$ to grow with the sample size $n$. 
Here we present theoretical results when $m$ grows to infinity.
Note that in the finite $m$ case, the analysis can be done easily
through incorporating the classical maximum likelihood approach.
We first list a set of regularity conditions.

\begin{enumerate}[label=(D\arabic*),ref=(D\arabic*),start=1]
\item\label{con:catbdd}
The true conditional mass function of $Y$ given $\X$, $\pr(Y=y\mid\x;\btheta_0)$,
has a support set $\{0,\dots,m\}$.
$E\{\pr(Y=0\mid\X;\btheta_0)\}\leq1-\delta$ for some constant $0<\delta<1$
and $E(Y^4\mid\x)$ is bounded.
The marginal density of $\X$, $f_{\X}(\x)$, 
has compact support $\mx$ and is bounded on its support.
\item\label{con:catlipschitz}
There exist constants $L_k$ such that
$|E(Y^k\mid\x,\btheta^*)-E(Y^k\mid\x,\btheta_0)|\leq L_k\|\btheta^*-\btheta_0\|_2$
for $k=1,2,3$.
\item\label{con:catm}
$m\to\infty$, $n^{-1}m^3\to0$, and $E\{\|\p(\X,\btheta_0)\|_2^2\}\to0$ as $n\to\infty$.
\item\label{con:catvar22}
$\bSig_{22}\equiv E[\var\{\B(Y)\mid\X\}]$ is invertible.
$\|\bSig_{22}\|_2^{-1}\bSig_{22}$ has all eigenvalues bounded above a
constant $C_l>0$.
\item\label{con:catvar}
The expectation of the conditional covariance of $\{\X\trans Y,
\B\trans(Y)\}\trans$ given $\X$, i.e.
\bse
\bSig
\equiv \begin{bmatrix} \bSig_{11} & \bSig_{12} \\ \bSig_{21} & \bSig_{22} \end{bmatrix}
\equiv E\begin{bmatrix}
\X\X\trans \var(Y\mid\X) & \X\cov\{Y,\B(Y)\mid\X\} \\
\cov\{\B(Y),Y\mid\X\}\X\trans & \var\{\B(Y)\mid\X\} 
\end{bmatrix},
\ese
is invertible.
\end{enumerate}

Conditions \ref{con:catbdd} and \ref{con:catlipschitz} require
 boundedness and  Lipschitz continuity on the conditional moments of $Y$,
which are standard requirements.
Condition \ref{con:catm} requires $m$ to tend to infinity
at the rate slower than $n^{1/3}$,
and the mass function not to concentrate on a finite subset of the support.
Further, since $\bSig_{22}$ and $\bSig$ are positive semidefinite by their definitions,
the invertibility imposed by Conditions \ref{con:catvar22} and \ref{con:catvar} is very mild.
Lastly, we point out some results on $\bSig_{22}$ in Remark
\ref{rem:var22} below.

\begin{Rem}\label{rem:var22}
Note that the sum of the eigenvalues of $\bSig_{22}$ is of constant order by
Conditions \ref{con:catbdd} and \ref{con:catm}, because
$\trace(\bSig_{22})=E\left\{1-\pr(Y=0\mid\X;\btheta_0)-\|\p(\X,\theta_0)\|_2^2\right\}$.
Thus we get $\|\bSig_{22}\|_2\asymp m^{-1}$ and
$\|\bSig_{22}^{-1}\|_2\asymp m$ 
because the eigenvalues are of the same order by Condition \ref{con:catvar22}.
\end{Rem}

We now state the convergence rate and the asymptotic properties
of the estimators of the model parameters $\bb$ and $\bg$.
Our analysis shows that the estimator of the regression coefficient $\wh\bb$
achieves the parametric convergence rate
under the regularity conditions stated above.
Further, we establish the asymptotic distribution of $\wh\bb$
as a by-product, 
by which one can perform inference on the regression coefficient.
We formally state the results in Proposition \ref{pro:catbbbg} below.

\begin{Pro} \label{pro:catbbbg}
Under Conditions \ref{con:catbdd}-\ref{con:catvar},
$\|\wh\bb-\bb_0\|_2=O_p(n^{-1/2})$,
$\|\wh\bg-\bg_0\|_2=O_p(n^{-1/2}m^{1/2})$, and 
\bse
\begin{bmatrix} \wh\bb-\bb_0 \\ \wh\bg-\bg_0 \end{bmatrix}
=
\bSig^{-1} n^{-1}\sumi
\begin{bmatrix} \x_i\{y_i-E(Y\mid\x_i)\} \\ \B(y_i)-E\{\B(Y)\mid\x_i\} \end{bmatrix}
+
\begin{bmatrix} \r_1 \\ \r_2 \end{bmatrix},
\ese
where $\|\r_1\|_2=o_p(n^{-1/2})$ and $\|\r_2\|_2=o_p(n^{-1/2}m^{1/2})$.
Furthermore, let $\bSig_\bb\equiv(\bSig_{11} - \bSig_{12}\bSig_{22}^{-1}\bSig_{21})^{-1}$, then
\bse
\bSig_\bb^{-1/2}\sqrt{n}(\wh\bb-\bb_0)\to N(\0_p,\I_p)
\ese
in distribution as $n\to\infty$.
\end{Pro}

Based on Proposition \ref{pro:catbbbg},
we further establish a theoretical result on the marginal effect estimator,
which states the convergence rate and the asymptotic distribution of $\wh\bxi$.
Theorem \ref{th:catxi} shows that $\wh\bxi$,
the functional of both the parametric and the nonparametric components,
achieves the parametric convergence rate.
We also provide the closed form of the asymptotic variance of $\wh\bxi$,
which can be used to infer the marginal effect of a population.

\begin{Th}\label{th:catxi}
Let $\bSig_\bxi\equiv
\A\bSig^{-1}\A\trans+\bb_0\bb_0\trans\var\{\var(Y\mid\X)\}$, 
where $\A\equiv[\A_1, \A_2]$ and
\bse
\A_1 &\equiv& E\{\var(Y\mid\X)\}\I + \bb_0E[\{Y-E(Y\mid\X)\}^3\X\trans], \\
\A_2 &\equiv& \bb_0 E(\{Y-E(Y\mid\X)\}^2[\B(Y)-E\{\B(Y)\mid\X\}])\trans.
\ese
Under Conditions \ref{con:catbdd}-\ref{con:catvar},
\bse
\bSig_\bxi^{-1/2}\sqrt{n}(\wh\bxi-\bxi_0)\to N(\0_p, \I_p)
\ese
in distribution as $n\to\infty$.
\end{Th}

We now show that our proposed estimators
$\wh\bb$ and $\wh\bxi$ for the discrete case also
achieve the efficiency bounds.
Although the asymptotic variances of $\wh\bb$ and $\wh\bxi$
established in Proposition \ref{pro:catbbbg} and 
in Theorem \ref{th:catxi} 
appear very different from the efficiency bounds derived
in Appendix \ref{proof:effbb} and Section \ref{sec:derive1},
our analysis shows that these two
seemingly different variance structures are actually identical.
Below, we formally state the efficiency of $\wh\bxi$ as Theorem
\ref{th:cateffxi}, 
and that of $\wh\bb$ as Proposition \ref{pro:cateffbb}.

\begin{Pro}\label{pro:cateffbb}
Under Conditions \ref{con:catbdd}-\ref{con:catvar}, 
the approximate maximum likelihood estimator $\wh\bb$ is efficient.  
\end{Pro}

\begin{Th}\label{th:cateffxi}
Under Conditions \ref{con:catbdd}-\ref{con:catvar}, 
the estimator $\wh\bxi$ based on the approximate maximum likelihood estimator
$\wh\bb$ and $\wh c(\cdot)$ is efficient.
\end{Th}

\section{Simulation Experiments}\label{sec:sim}

We conduct simulation studies to investigate the finite sample
performance of the proposed methods. We consider the case where the
response follows a GLM or a truncated GLM.
All results are based on 1000 replicates with sample size $n=1000$.
For comparison, we implemented three different estimators,
our proposed approximate maximum likelihood estimator (aMLE),
the pairwise marginal likelihood estimator (pMLE) by \cite{lin2021},
and the maximum likelihood estimator (MLE) under the non-truncated
regression model.
Note that MLE is the most efficient estimator if the response
$Y$ is not truncated, but wrongly specifies the distribution when the
response $Y$ is actually truncated. 

For our proposed aMLE for the continuous response,
we used the cubic B-spline basis
with the number of  interior knots equal to
the smallest integer larger than $0.7n^{1/5}$, i.e., $N=\lceil 0.7n^{1/5}\rceil$,
where the knots are the quantiles of $\{y_i:i=1,\dots,n\}$ of length $(N+2)$
whose levels are evenly spaced in $[0,1]$.
On the other hand, because \cite{lin2021} only studied the estimation
of $\bb$ and $c(\cdot)$, 
we implemented the estimators $\wh\bxi$ and $\wh\bta_\tau$ of pMLE
in the same manner as for aMLE,  based on the pMLE
estimated $\bb$ and  $c(\cdot)$.

We report the average of the absolute bias,
the sample standard error $\sigma_{\text{sim}}$,
the average of the asymptotic standard error $\wh\sigma_{\text{est}}$,
and the empirical coverage  of the estimated confidence interval at
95\% confidence level (CI).
The  $\wh\sigma_{\text{est}}$ and CI of pMLE are omitted
because \cite{lin2021} did not provide them.

\subsection{Normal Distribution}

We first examine the normal regression model. 
A three dimensional covariate vector $\X_i$ was
independently drawn from a multivariate normal distribution with mean $\0$
and covariance $\Sigma=(\sigma_{kl})$ where
$\sigma_{kl}=0.1^{|k-l|},k=1,\dots,3,l=1,\dots,3$. 
Then for the truncated case, we further generated a response $Y_i$
independently 
from a truncated normal distribution on $[a,b]$, which has the density
\bse
f_{Y}(y;\theta_i,\sigma,a,b) = \frac{\frac{1}{\sigma}\phi(\frac{y-\theta_i}{\sigma})}
{\Phi(\frac{b-\theta_i}{\sigma})-\Phi(\frac{a-\theta_i}{\sigma})},
\ese
where $\theta_i=\bb\trans\x_i, \sigma=1, a=-5, b=5$,
$\phi(\cdot)$ and $\Phi(\cdot)$ are the probability density function
(pdf) and the cumulative distribution function (cdf) of the 
standard normal distribution. 
We set $\bb=(1,2,3)\trans$ in the simulation procedure.
For the non-truncated case, we generated replicates with the same
parameters but without truncation. 
The non-truncated simulation design is identical to that of \cite{lin2021}.

To evaluate the performance under the truncated case, we illustrate
in Figure \ref{fig:normalcy} 
how the estimated $c(\cdot)$ performed by our method while fixing $c(-5)=0$.
As can be seen from the plot, $\wh c(\cdot)$ approximated 
$c(\cdot)$ with satisfactory bias and variance at sample size
1000. This reflects the theoretical properties described in Lemma
\ref{lem:whbg} of the supplementary material. 
In addition, in Table \ref{tab:normal}, we illustrate the estimation
properties for  $\bb$ and $\bxi$ via aMLE, pMLE and MLE.
The aMLE and pMLE methods were numerically almost identical in terms of
both bias and standard error, reflecting the theoretical properties
established in Proposition \ref{pro:effbb}  as well as those in Theorem 1
of \cite{lin2021}. In terms of inference of aMLE, the estimated
standard error was very close to the sample standard error, and
the coverage rate of the confidence intervals was
close to the nominal level $95\%$. These indicate that our asymptotic
properties are already useful at sample size $n=1000$ in this
model. In contrast, the inference results of MLE were very bad, with
the coverage rate almost 0. This is a direct consequence of the
estimation bias caused by model misspecification.
Table \ref{tab:normaleta} in the supplementary material
provides the estimation results of $\bta_\tau$
at the quantile levels $\tau=0.05, 0.25, 0.5, 0.75$ and  0.95
respectively. Again, the estimation and  inference properties of aMLE
were satisfactory, suggesting that the properties described in 
Theorem \ref{th:bta} is reflected in this model for $n=1000$.
Also, MLE performed poorly in estimating $\bta_\tau$ especially at the
low or high quantile level.

For the non-truncated case, the corresponding results are shown in
Figure \ref{fig:normal2cy} and Tables \ref{tab:normal2} and
\ref{tab:normal2eta} in the supplementary material.
We can see from Figure \ref{fig:normal2cy}
that $\wh c(\cdot)$ estimated the curve $c(\cdot)$ sufficiently well, 
even though the response is infinitely supported hence our compact
support assumption is violated. We also note the small biases of
the estimators $\wh\bb$, $\wh\bxi$, and $\wh\bta_\tau$ from 
Tables \ref{tab:normal2} and \ref{tab:normal2eta}.
Moreover, interestingly,
our estimators of $\bxi$ and $\bta_\tau$ performed as well as MLE in
this simulation setting, even when $\tau$ is near 0 or 1. In terms of
inference, the empirical coverage of the estimated confidence
interval by our method was close to the nominal level. These seem to
suggest an empirical robustness property of our aMLE against the
compact support assumption.
Lastly, we point out that our method aMLE was more efficient than pMLE
in this example in terms of estimating $\bb, \bxi$ as well as the quantile
effect $\bta_\tau$ at various quantile levels $\tau$.

\subsection{Gamma Distribution}

Next, we consider the situation when $Y$ given $\x$ has a gamma
distribution.
We first generated a covariate vector $\X_i=(X_{1i},X_{2i})\trans$,
where $X_{ki},k=1,2$ are independently and identically distributed
(iid) as a uniform random variable on $[0.5,1]$.
Then for the truncated case, a response $Y_i$ was generated independently
from a truncated gamma distribution with the pdf
\bse
f_{Y|\X}(y;\alpha,\theta_i,b) =
\frac{y^{\alpha-1}e^{-y/\theta_i}}{\Gamma(\alpha)\theta_i^\alpha}\left\{\int_0^b\frac{y^{\alpha-1}e^{-y/\theta_i}}{\Gamma(\alpha)\theta_i^\alpha}\right\}^{-1}, 
\ese
where $\alpha=5$, $\alpha\theta_i=1/\bb\trans\x_i$, $b=2$, and
$\bb=(0.5,1)\trans$.
For the non-truncated case, we simply carried out the same data
generation mechanism while applying the usual gamma distribution
with the same parameters.

For the truncated case,
Figure \ref{fig:gammacy} illustrates the performance of the estimator $\wh
c(\cdot)$ with $c(2)=0$. The estimation had very small bias on
most part of the support  except when $y$ is
close to 0 due to the boundary effect.
Indeed, the true curve $c(y)=(\alpha-1)\log(y)$ is unbounded near 0
and very few observations are available since the density converges to
0. Nevertheless, we can see from Tables \ref{tab:gamma} and
\ref{tab:gammaeta} 
that aMLE estimated $\bb$, $\bxi$, and $\bta_\tau$ well with small
bias, and its inference was also sufficiently precise with a good match
between the sample variance and the estimated variance, and the 95\%
confidence intervals had coverage close to the nominal level. In
contrast, MLE showed bias and did not perform well in general, and
deteriorated 
further when $\tau$ is near 1.
In our application, we also find that
pMLE was computationally unstable in  this setting and did not lead
to reasonable results.

For the case when the response is not truncated, 
Figure \ref{fig:gamma2cy} shows that  similar performance was observed
as in the truncated gamma regression case. In addition, 
Tables \ref{tab:gamma2} and \ref{tab:gamma2eta} in the supplementary
material suggest that 
aMLE  still provided good results in estimating the
corresponding parameters and can be used for reliable inference.
In this situation, aMLE still outperformed pMLE numerically
even though both should be inconsistent in theory. Because MLE assumes a
fully parametric model which happens to be correct, it had the best
performance among all three methods as we expected.

\subsection{Bernoulli Distribution}

Our first simulation for the discrete response was carried out
under a conditional Bernoulli
distribution with the intention to investigate the performance of aMLE
in the
discrete response case.
A three dimensional covariate vector $\X_i$ was generated
independently from a multivariate normal distribution with mean $\0$
and covariance $\Sigma=(\sigma_{kl})$ where $\sigma_{kl}=0.1^{|k-l|},k=1,\dots,3,l=1,\dots,3$.
Then a binary response $Y_i$ was drawn independently
from a Bernoulli distribution with a success rate $1/\{1+\exp(-\bb\trans\x_i)\}$.
We set $\bb=(-0.5,0.5,1)\trans$. 

The performances of the estimators $\wh\bb$ and $\wh\bxi$ by the three methods
are given in Table \ref{tab:bernoulli}.    
All methods performed similarly to each other,
because the problem is fully parametric as explained in Section \ref{sec:est}.

\subsection{Poisson and Negative Binomial Distributions}\label{sec:poisson}
We now conduct simulation studies for the case where
$Y$ is discrete with infinite support. Specifically, $Y$ has
Poisson and negative binomial distributions respectively.
We first generated a covariate vector $\X_i=(X_{1i},X_{2i})\trans$,
where $X_{ki},k=1,2$ are independently and identically distributed
(iid)  uniform random variables on $[0.5,1]$.
Then for the Poisson regression model, a response $Y_i$ was generated
from a Poisson distribution with the rate
$\theta_i=\exp(\bb\trans\x_i)$ where $\bb=(0,1)\trans$.
For the negative binomial case,
we generated $Y_i$ independently from a negative binomial distribution with the mass function
\bse
\pr(Y=y;r,\theta_i) = 
{y+r-1 \choose r-1} \theta_i^y (1-\theta_i)^r,
\ese
where $r=2$, $\theta_i=\exp(\bb\trans\x_i)$, and $\bb=(0,-1)\trans$.
In estimation, 
MLE assumes $\theta_i=\exp(\bb\trans\x_i)$ and $c(y)=-\log(y!)$
hence results in the most efficient estimator under the Poission regression
among the three methods,
but wrongly specifies the structure
when the conditional distribution is the negative binomial.

Tables \ref{tab:poisson} and \ref{tab:nbinomial} show the $\bb$ and $\bxi$ estimation results
under the Poisson and the negative binomial regression models respectively.
In both cases, we found that aMLE estimated $\bxi$ more efficiently than pMLE
even though our method aMLE showed similar performance to pMLE in estimating $\bb$.
Especially for the estimation of $\xi_2$ in Table \ref{tab:nbinomial},
we can see that aMLE was about twice as efficient as pMLE
in terms of both bias and standard error.
This supports that our estimation procedure is more suitable for estimating $\bxi$
in finite sample situations than the method proposed by \cite{lin2021}.
Also, considering that the empirical coverage of our approximate confidence interval
reached the nominal level of 95\%,
the results confirm the asymptotic properties of $\wh\bb$ and $\wh\bxi$
proposed in Proposition \ref{pro:catbbbg} and Theorem \ref{th:catxi}.
On the other hand, MLE in Table \ref{tab:poisson} indeed had the best
performance in terms of both bias and variance, benefiting from the
perfectly correct parametric model setting. 
However,  we can observe from Table \ref{tab:nbinomial} that
the estimators from MLE were severely biased when $c(y)$ was misspecified.

\section{Swiss Non-labor Income Data Analysis}\label{sec:real}

We now analyze a data set concerning non-labor income situation
in Switzerland.
The data set is publicly available from the R-package \texttt{AER}
\citep{kleiber2008applied}, 
and consists of information collected from 871 married
women randomly 
drawn from the representative health survey for Switzerland (SOMIPOPS)
in 1981 with income not abnormally low.

Our goal is to estimate the marginal effect and the quantile effect of
covariates on
non-labor income (such as husband's income). Thus, the non-labor 
income in the log scale is used as response $Y$. 
All the other variables in the dataset are considered as explanatory
variables, and they include
Participation (taking the value 1 if the individual participated
in the labor market, 0 otherwise), 
Age (age in years),
Age2 (squared age in years then divided by 10),
Education (years of formal education),
Foreign (taking the value 1 if the individual is a permanent foreign
resident, 0 otherwise),
Youngkids (number of young children),
Oldkids (number of older children).
The age category for children was decided by whether the age is under 7.
For  more detailed description, see \cite{gerfin1996parametric}.

We implemented our method aMLE. For comparison, we also implemented
 the normal and the gamma regression.
We fitted the models without any transformation on or
interaction between the covariates, while
the number of knots of the B-spline basis is chosen in the same way as
explained in Section \ref{sec:sim}. 

In Figure \ref{fig:cyreal}, we can see that
the fitted distribution from aMLE does not resemble that of the normal
regression, whose $c(y)=-y^2/2\sigma^2$ with $\sigma>0$. It also does
not resemble that of the gamma distribution whose $c(y)=(\alpha-1)\log
(y)$ with $\alpha>0$. In Table \ref{tab:criteriareal}, we compared
the information criteria AIC and  BIC of the three models, and found
that our aMLE on (\ref{eq:sglm}) has the lowest AIC and BIC compared
to the normal and 
the gamma regression model, hence is the most suitable modeling choice.
This suggests that our model in (\ref{eq:sglm}) is the most suitable
to use in analyzing the swiss non-labor income data.

Table \ref{tab:xireal} shows the inference result in estimating the marginal effect.
All except one explanatory variables have similar levels of
significance based on all three methods. The exception is ``Foreign'', 
which was selected as a significant variable at 5\%
level in the analysis based on aMLE, but was considered non-significant
in both the normal and the gamma regression.
In other words, only by using (\ref{eq:sglm}) in combination with aMLE,
we can conclude that being a 
permanent foreign resident has a negative effect on the non-labor income
in Switzerland.

In Figure \ref{fig:etareal}, we further illustrate the estimated $\wh\bta_\tau$
at the quantile level $\tau= 0.05$, 0.25, 0.5, 0.75 and 0.95.
We fixed the significance level at 5\%.
The p-values of $\wh\bta_\tau$ were largely similar to those of $\wh\bxi$.
However, as we can see from the plot,
the magnitude of the quantile effect increases when $\tau$ is near 0 or 1.
That is, our model suggests that
when the non-labor income level is relatively low or high,
the effect of the covariates on the non-labor income is more extreme.

\section{Discussion}\label{sec:dis}

We have proposed a B-spline based approximate maximum likelihood
estimation procedure to estimate both the marginal effect and the quantile
effect in a semiparametric generalized linear model. Compared to the
classical GLM, our model is more flexible hence less susceptible to
model misspecification.
The estimators we proposed are shown to reach the
semiparametric efficiency bounds, hence are optimal.

Following the spirit of GLMs,  we have worked with a linear summary
of covariates $\bb\trans\x$ in our work. It is easy to see that we
can replace $\bb\trans\x$ by a more general form $m(\x,\bb)$, where
$m$ is a known function which can be nonlinear. All our procedures can
be carried through while replacing $\x$ by $\partial m(\x,\bb)/\partial\bb$,
where we only need to ensure identifiability, sufficient smoothness,
and boundedness of the corresponding quantities to facilitate the
almost identical derivation to the linear case. 

An interesting but difficult further extension of our work worth
mentioning is the possibility of allowing non-compactly
supported distribution of $Y$. 
 Although truncation is
routinely done in practice to bypass the complexity, we found it very difficult
to rigorously extend our 
theoretical analysis to encompass this scenario. We suspect that
more
fundamental work is needed,  possibly in a model simpler than the
semiparametric GLM proposed in (\ref{eq:sglm}).

\bibliographystyle{agsm}
\bibliography{marg}

\vfill
\begin{figure}[hp]
    \centering
    \caption{$c(\cdot)$ estimation results.
    Red: the true curve $c(\cdot)$; Black: the median curve of $\wh c(\cdot)$;
    Filled curves: the 2.5\% and 97.5\% quantiles of $\wh c(\cdot)$.}
    \begin{subfigure}{0.45\textwidth}
        \centering
        \caption{truncated normal}
        \includegraphics[width=\textwidth]{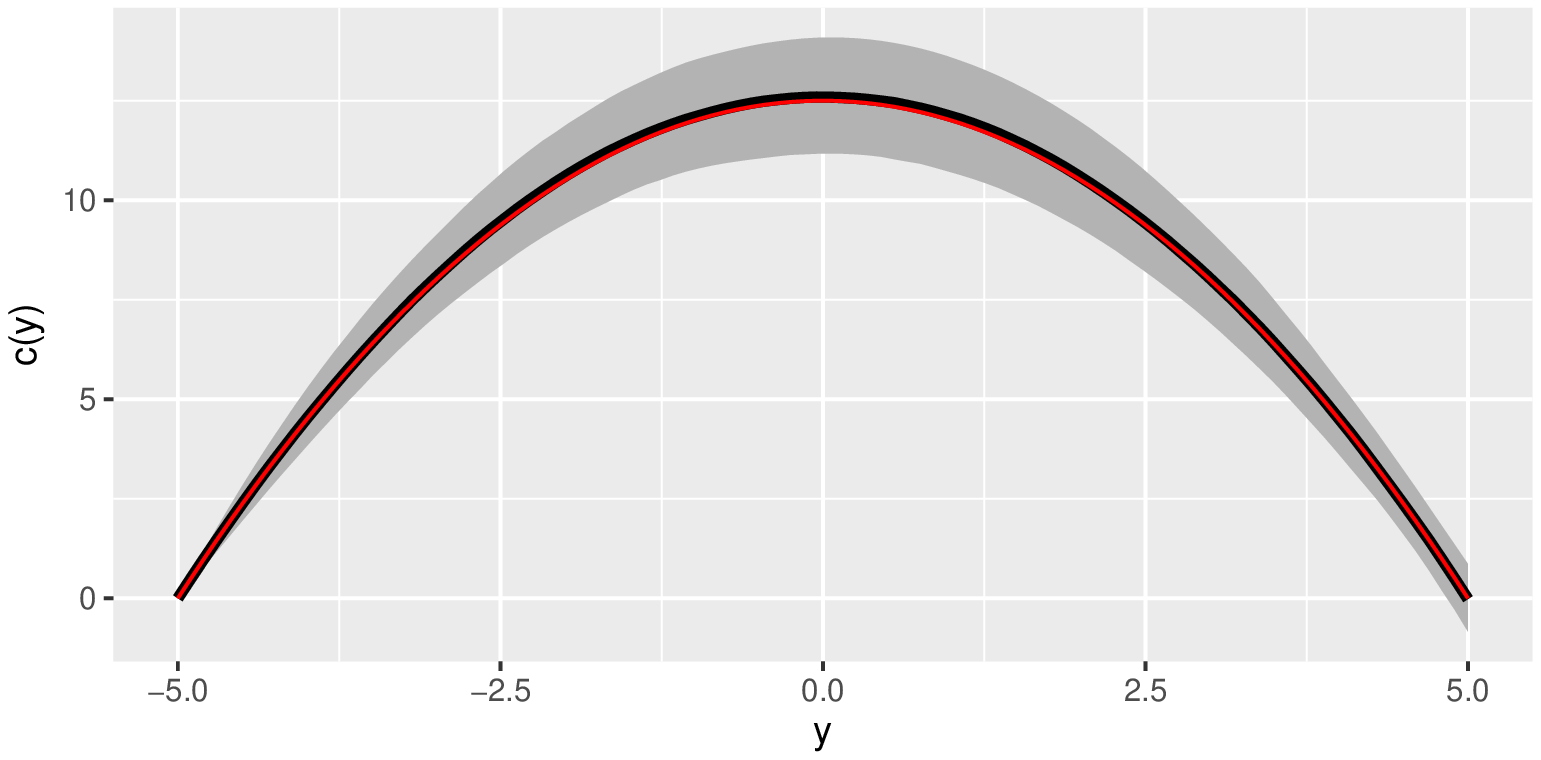}
        \label{fig:normalcy}
    \end{subfigure}
    \hfill
    \begin{subfigure}{0.45\textwidth}
        \centering
        \caption{normal}
        \includegraphics[width=\textwidth]{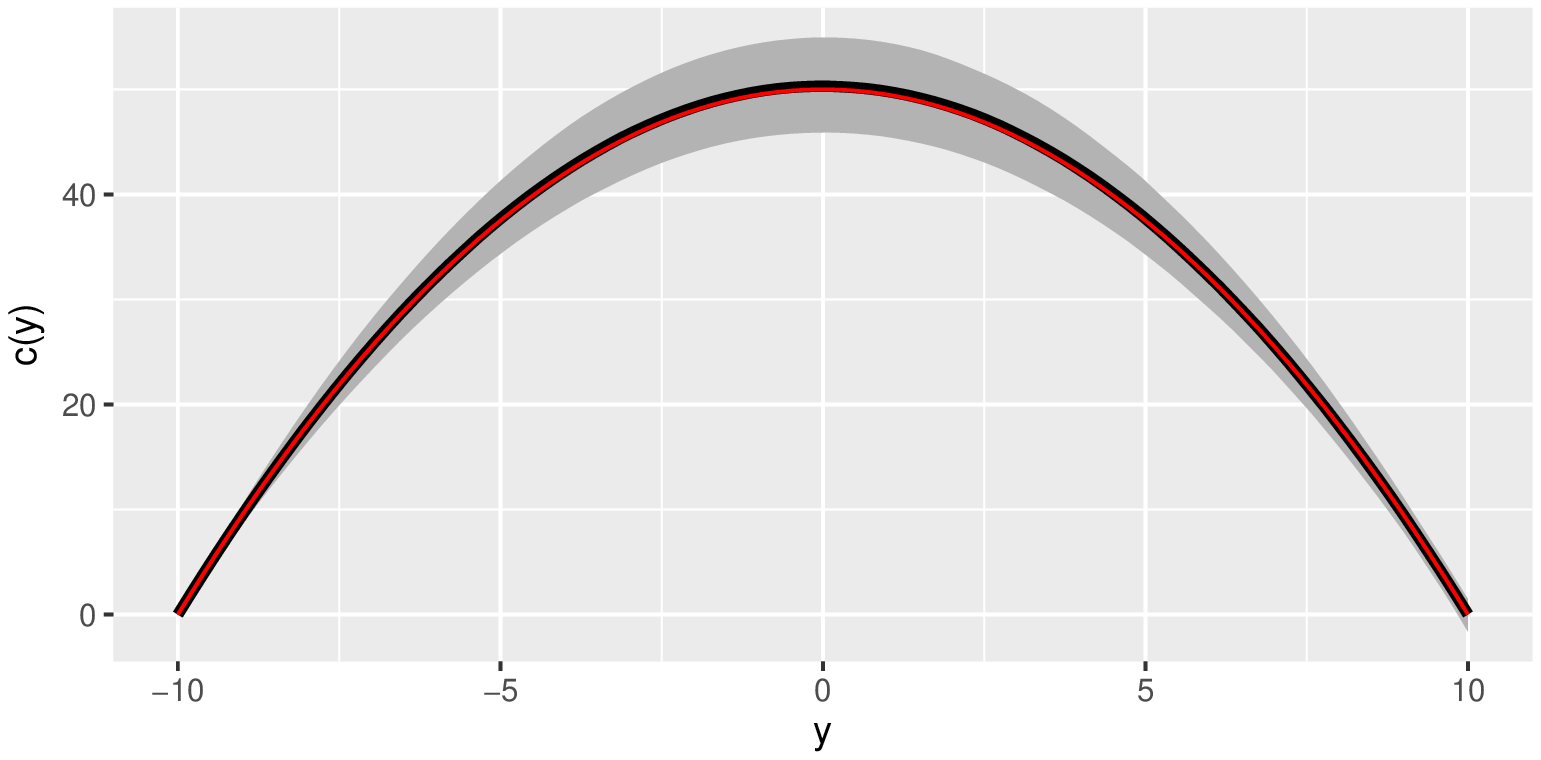}
        \label{fig:normal2cy}
    \end{subfigure}
    \begin{subfigure}{0.45\textwidth}
        \centering
        \caption{truncated gamma}
        \includegraphics[width=\textwidth]{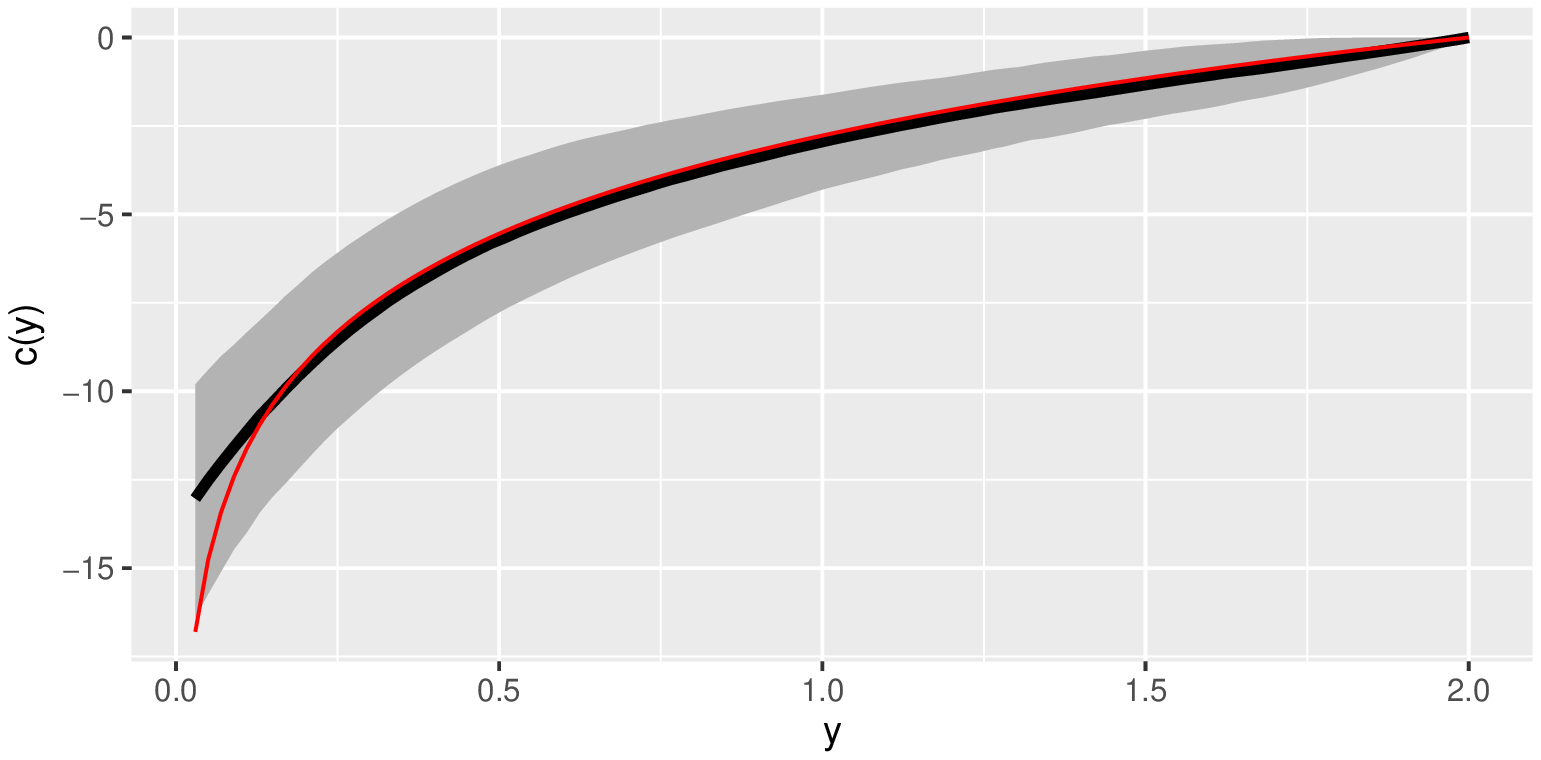}
        \label{fig:gammacy}
    \end{subfigure}
    \hfill
    \begin{subfigure}{0.45\textwidth}
        \centering
        \caption{gamma}
        \includegraphics[width=\textwidth]{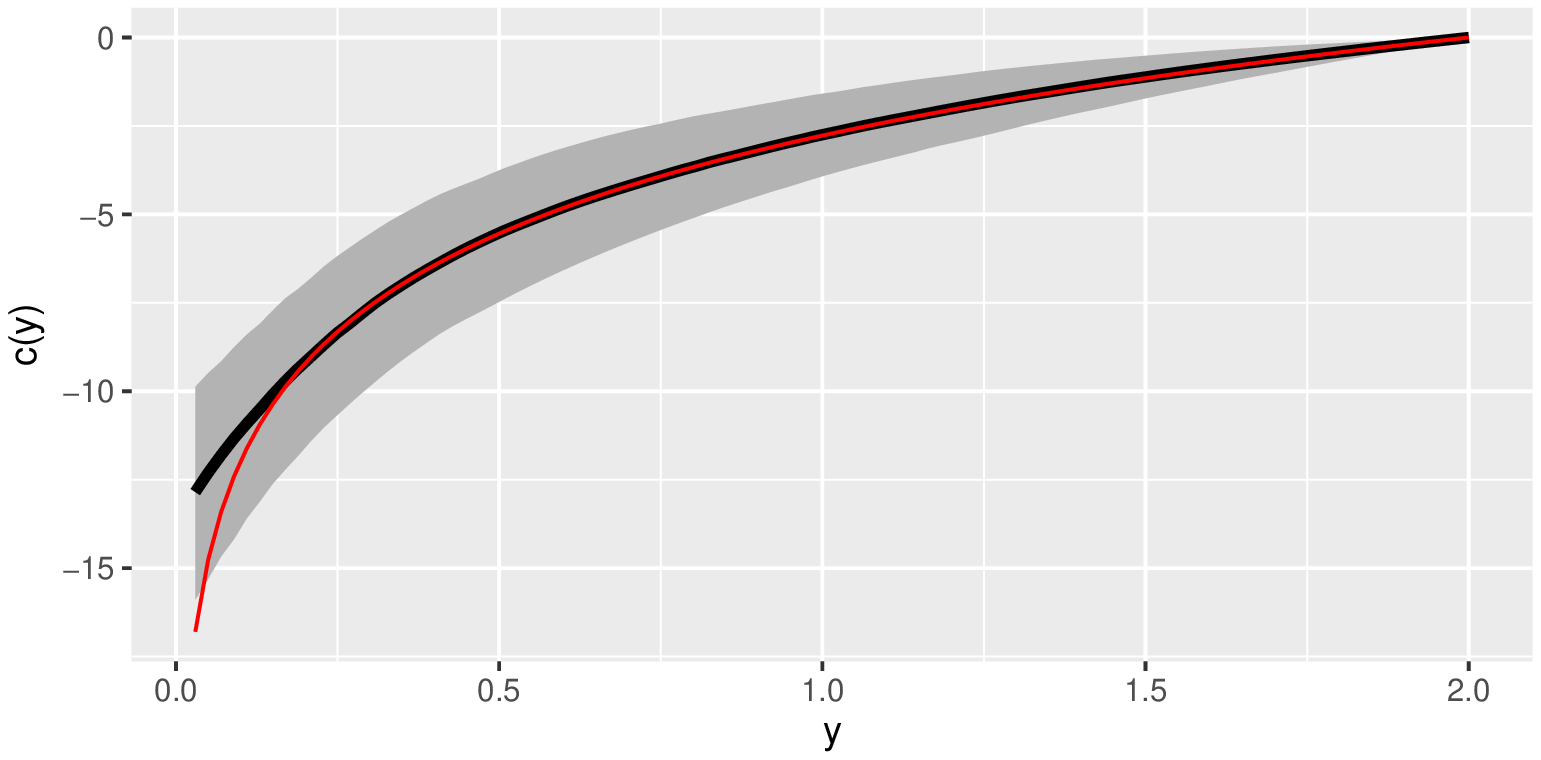}
        \label{fig:gamma2cy}
    \end{subfigure}
\end{figure}

\begin{figure}[hp]
    \centering
    \caption{$c(\cdot)$ estimation in the Swiss non-labor income data.
    Black: the curve $\wh c(\cdot)$;
    Filled curves: the estimated pointwise confidence band of $c(\cdot)$.}
    \label{fig:cyreal}
    \includegraphics[width=0.6\textwidth]{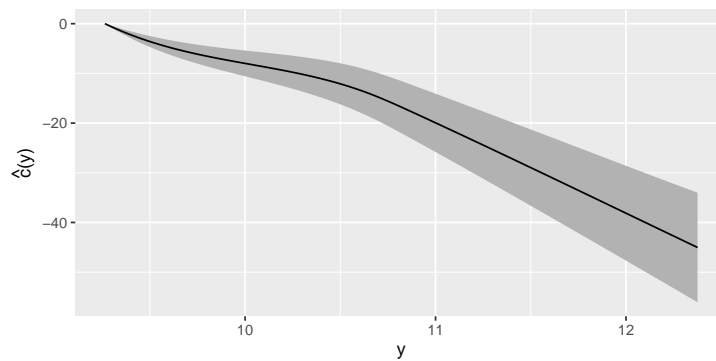}
\end{figure}

\begin{figure}[hp]
    \centering
    \caption{The result of $\bta_\tau$ estimation in the Swiss
      non-labor income data.} 
    \label{fig:etareal}
    \includegraphics[width=0.9\textwidth]{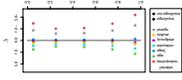}
\end{figure}

\begin{table}[]
    \centering
    \caption{$\bb$ and $\bxi$ estimation results under the truncated normal distribution.}
    \label{tab:normal}
    \small
\begin{tabular}{c|ccc|ccc|ccc|ccc}
          & \multicolumn{3}{c|}{$|$bias$|$} & \multicolumn{3}{c|}{$\sigma_{\text{sim}}$} & \multicolumn{3}{c|}{$\hat\sigma_{\text{est}}$} & \multicolumn{3}{c}{C.I.} \\
          & aMLE      & pMLE     & MLE      & aMLE         & pMLE         & MLE          & aMLE           & pMLE          & MLE           & aMLE   & pMLE   & MLE    \\ \hline
$\beta_1$ & .049      & .049     & .263     & .061         & .061         & .035         & .060           & -             & .035          & .947   & -      & .000   \\
$\beta_2$ & .083      & .084     & .530     & .103         & .103         & .041         & .101           & -             & .036          & .951   & -      & .000   \\
$\beta_3$ & .118      & .118     & .793     & .145         & .145         & .044         & .145           & -             & .035          & .949   & -      & .000   \\ \hline
$\xi_1$   & .022      & .023     & .028     & .028         & .029         & .035         & .029           & -             & .035          & .948   & -      & .951   \\
$\xi_2$   & .028      & .027     & .033     & .035         & .034         & .041         & .035           & -             & .036          & .941   & -      & .912   \\
$\xi_3$   & .032      & .031     & .036     & .039         & .038         & .044         & .043           & -             & .035          & .964   & -      & .881  
\end{tabular}
\end{table}

\begin{table}[]
    \centering
    \caption{$\bb$ and $\bxi$ estimation results under the normal distribution.}
    \label{tab:normal2}
    \small
\begin{tabular}{c|ccc|ccc|ccc|ccc}
          & \multicolumn{3}{c|}{$|$bias$|$} & \multicolumn{3}{c|}{$\sigma_{\text{sim}}$} & \multicolumn{3}{c|}{$\hat\sigma_{\text{est}}$} & \multicolumn{3}{c}{C.I.} \\
          & aMLE      & pMLE     & MLE      & aMLE         & pMLE         & MLE          & aMLE           & pMLE          & MLE           & aMLE   & pMLE   & MLE    \\ \hline
$\beta_1$ & .046      & .049     & .026     & .057         & .064         & .032         & .055           & -             & .032          & .938   & -      & .954   \\
$\beta_2$ & .081      & .088     & .025     & .099         & .114         & .032         & .096           & -             & .032          & .946   & -      & .957   \\
$\beta_3$ & .116      & .127     & .025     & .143         & .166         & .032         & .139           & -             & .032          & .940   & -      & .941   \\ \hline
$\xi_1$   & .026      & .029     & .026     & .032         & .038         & .032         & .032           & -             & .032          & .953   & -      & .954   \\
$\xi_2$   & .025      & .030     & .025     & .032         & .038         & .032         & .035           & -             & .032          & .961   & -      & .957   \\
$\xi_3$   & .026      & .030     & .025     & .032         & .038         & .032         & .038           & -             & .032          & .964   & -      & .941  
\end{tabular}
\end{table}

\begin{table}[]
    \centering
    \caption{$\bb$ and $\bxi$ estimation results under the truncated gamma distribution.}
    \label{tab:gamma}
    \small
\begin{tabular}{c|ccc|ccc|ccc|ccc}
          & \multicolumn{3}{c|}{$|$bias$|$} & \multicolumn{3}{c|}{$\sigma_{\text{sim}}$} & \multicolumn{3}{c|}{$\hat\sigma_{\text{est}}$} & \multicolumn{3}{c}{C.I.} \\
          & aMLE      & pMLE     & MLE      & aMLE         & pMLE         & MLE          & aMLE           & pMLE          & MLE           & aMLE   & pMLE   & MLE    \\ \hline
$\beta_1$ & .097      & .114     & .074     & .124         & .141         & .074         & .122           & -             & .076          & .944   & -      & .889   \\
$\beta_2$ & .101      & .117     & .062     & .128         & .145         & .077         & .128           & -             & .077          & .942   & -      & .945   \\ \hline
$\xi_1$   & .064      & .284     & .108     & .081         & .154         & .059         & .080           & -             & .061          & .942   & -      & .589   \\
$\xi_2$   & .064      & .523     & .117     & .080         & .126         & .061         & .079           & -             & .061          & .947   & -      & .531  
\end{tabular}
\end{table}

\begin{table}[hp]
    \centering
    \caption{$\bta_\tau$ estimation results under the truncated gamma distribution.}
    \label{tab:gammaeta}
    \vspace{0.5em}
    \begin{tabular}{c|cl|cccc}
    $\tau$ & & Method & $|$bias$|$ & $\sigma_{\rm sim}$ & $\wh\sigma_{\rm est}$ & C.I. \\ \hline
    \multirow{6}{*}{0.05} & $\eta_1$ & aMLE & 0.032 & 0.041 & 0.041 & 0.945 \\
                          &          & pMLE & 0.120 & 0.077 & -     & -     \\
                          &          & MLE  & 0.030 & 0.025 & 0.025 & 0.825 \\
                          & $\eta_2$ & aMLE & 0.038 & 0.048 & 0.048 & 0.945 \\
                          &          & pMLE & 0.217 & 0.090 & -     & -     \\
                          &          & MLE  & 0.023 & 0.026 & 0.027 & 0.943 \\ \hline
    \multirow{6}{*}{0.25} & $\eta_1$ & aMLE & 0.052 & 0.066 & 0.065 & 0.942 \\
                          &          & pMLE & 0.395 & 0.163 & -     & -     \\
                          &          & MLE  & 0.048 & 0.041 & 0.041 & 0.826 \\
                          & $\eta_2$ & aMLE & 0.056 & 0.071 & 0.070 & 0.938 \\
                          &          & pMLE & 0.750 & 0.127 & -     & -     \\
                          &          & MLE  & 0.036 & 0.042 & 0.044 & 0.947 \\ \hline
    \multirow{6}{*}{0.50} & $\eta_1$ & aMLE & 0.068 & 0.086 & 0.085 & 0.943 \\
                          &          & pMLE & 0.436 & 0.192 & -     & -     \\
                          &          & MLE  & 0.071 & 0.055 & 0.056 & 0.799 \\
                          & $\eta_2$ & aMLE & 0.071 & 0.089 & 0.089 & 0.939 \\
                          &          & pMLE & 0.824 & 0.136 & -     & -     \\
                          &          & MLE  & 0.055 & 0.058 & 0.059 & 0.919 \\ \hline
    \multirow{6}{*}{0.75} & $\eta_1$ & aMLE & 0.083 & 0.106 & 0.105 & 0.950 \\
                          &          & pMLE & 0.296 & 0.181 & -     & -     \\
                          &          & MLE  & 0.118 & 0.073 & 0.074 & 0.669 \\
                          & $\eta_2$ & aMLE & 0.087 & 0.109 & 0.107 & 0.945 \\
                          &          & pMLE & 0.529 & 0.144 & -     & -     \\
                          &          & MLE  & 0.112 & 0.077 & 0.079 & 0.761 \\ \hline
    \multirow{6}{*}{0.95} & $\eta_1$ & aMLE & 0.089 & 0.113 & 0.111 & 0.936 \\
                          &          & pMLE & 0.120 & 0.109 & -     & -     \\
                          &          & MLE  & 0.336 & 0.106 & 0.107 & 0.113 \\
                          & $\eta_2$ & aMLE & 0.089 & 0.112 & 0.113 & 0.943 \\
                          &          & pMLE & 0.220 & 0.141 & -     & -     \\
                          &          & MLE  & 0.496 & 0.111 & 0.115 & 0.008
    \end{tabular}
\end{table}

\begin{table}[]
    \centering
    \caption{$\bb$ and $\bxi$ estimation results under the gamma distribution.}
    \label{tab:gamma2}
    \small
\begin{tabular}{c|ccc|ccc|ccc|ccc}
          & \multicolumn{3}{c|}{$|$bias$|$} & \multicolumn{3}{c|}{$\sigma_{\text{sim}}$} & \multicolumn{3}{c|}{$\hat\sigma_{\text{est}}$} & \multicolumn{3}{c}{C.I.} \\
          & aMLE     & pMLE      & MLE      & aMLE         & pMLE         & MLE          & aMLE           & pMLE          & MLE           & aMLE   & pMLE   & MLE    \\ \hline
$\beta_1$ & .088     & .372      & .061     & .112         & .578         & .077         & .111           & -             & .076          & .948   & -      & .954   \\
$\beta_2$ & .093     & .372      & .064     & .119         & .541         & .080         & .119           & -             & .078          & .944   & -      & .943   \\ \hline
$\xi_1$   & .073     & 1.040     & .051     & .093         & .986         & .064         & .092           & -             & .065          & .941   & -      & .951   \\
$\xi_2$   & .079     & 1.551     & .055     & .099         & .782         & .068         & .098           & -             & .065          & .944   & -      & .942  
\end{tabular}
\end{table}

\begin{table}[]
    \centering
    \caption{$\bb$ and $\bxi$ estimation results under the Bernoulli distribution.}
    \label{tab:bernoulli}
    \small
\begin{tabular}{c|ccc|ccc|ccc|ccc}
          & \multicolumn{3}{c|}{$|$bias$|$} & \multicolumn{3}{c|}{$\sigma_{\text{sim}}$} & \multicolumn{3}{c|}{$\hat\sigma_{\text{est}}$} & \multicolumn{3}{c}{C.I.} \\
          & aMLE      & pMLE     & MLE      & aMLE         & pMLE         & MLE          & aMLE           & pMLE          & MLE           & aMLE   & pMLE   & MLE    \\ \hline
$\beta_1$ & .089      & .089     & .088     & .111         & .111         & .111         & .109           & -             & .109          & .950   & -      & .949   \\
$\beta_2$ & .088      & .088     & .088     & .112         & .112         & .111         & .110           & -             & .110          & .947   & -      & .946   \\
$\beta_3$ & .100      & .099     & .099     & .123         & .123         & .123         & .124           & -             & .124          & .956   & -      & .954   \\ \hline
$\xi_1$   & .015      & .015     & .015     & .019         & .019         & .019         & .019           & -             & .019          & .950   & -      & .949   \\
$\xi_2$   & .015      & .015     & .015     & .020         & .020         & .020         & .019           & -             & .019          & .945   & -      & .944   \\
$\xi_3$   & .014      & .014     & .014     & .017         & .017         & .017         & .017           & -             & .017          & .947   & -      & .948  
\end{tabular}
\end{table}

\begin{table}[]
    \centering
    \caption{$\bb$ and $\bxi$ estimation results under the Poisson distribution.}
    \label{tab:poisson}
    \small
\begin{tabular}{c|ccc|ccc|ccc|ccc}
          & \multicolumn{3}{c|}{$|$bias$|$} & \multicolumn{3}{c|}{$\sigma_{\text{sim}}$} & \multicolumn{3}{c|}{$\hat\sigma_{\text{est}}$} & \multicolumn{3}{c}{C.I.} \\
          & aMLE      & pMLE     & MLE      & aMLE         & pMLE         & MLE          & aMLE           & pMLE          & MLE           & aMLE   & pMLE   & MLE    \\ \hline
$\beta_1$ & .121      & .116     & .088     & .154         & .146         & .110         & .151           & -             & .109          & .939   & -      & .944   \\
$\beta_2$ & .130      & .130     & .087     & .162         & .157         & .109         & .158           & -             & .106          & .947   & -      & .938   \\ \hline
$\xi_1$   & .258      & .268     & .187     & .327         & .337         & .236         & .320           & -             & .233          & .938   & -      & .944   \\
$\xi_2$   & .267      & .296     & .195     & .334         & .365         & .243         & .326           & -             & .237          & .940   & -      & .934  
\end{tabular}
\end{table}

\begin{table}[]
    \centering
    \caption{$\bb$ and $\bxi$ estimation results under the negative binomial distribution.}
    \label{tab:nbinomial}
    \small
\begin{tabular}{c|ccc|ccc|ccc|ccc}
          & \multicolumn{3}{c|}{$|$bias$|$} & \multicolumn{3}{c|}{$\sigma_{\text{sim}}$} & \multicolumn{3}{c|}{$\hat\sigma_{\text{est}}$} & \multicolumn{3}{c}{C.I.} \\
          & aMLE     & pMLE     & MLE       & aMLE         & pMLE          & MLE         & aMLE           & pMLE          & MLE           & aMLE   & pMLE   & MLE    \\ \hline
$\beta_1$ & .088     & .075     & 1.388     & .112         & .095          & .171        & .113           & -             & .114          & .956   & -      & .000   \\
$\beta_2$ & .103     & .158     & .383      & .130         & .115          & .176        & .128           & -             & .120          & .951   & -      & .202   \\ \hline
$\xi_1$   & .336     & .395     & 2.535     & .426         & .502          & .340        & .430           & -             & .223          & .958   & -      & .000   \\
$\xi_2$   & .424     & .886     & 2.708     & .528         & 1.008         & .322        & .518           & -             & .217          & .947   & -      & .000  
\end{tabular}
\end{table}

\begin{table}[hp]
    \centering
    \caption{AIC and BIC in the Swiss non-labor income data.}
    \label{tab:criteriareal}
    \vspace{0.5em}
    \begin{tabular}{l|cc}
    Method & AIC & BIC \\ \hline
    aMLE & 538.800 & 600.806 \\
    Normal   & 667.201 & 710.128 \\
    Gamma    & 657.090 & 700.016
    \end{tabular}
\end{table}

\begin{table}[hp]
    \centering
    \caption{$\bxi$ estimation results in the Swiss non-labor income data.
    ``*'' indicates the significance of the corresponding predictor at
    $5\%$ significance level.} 
    \label{tab:xireal}
    \vspace{0.5em}
    \begin{tabular}{l|rl|rl|rl}
    Variable &
    $\wh\xi_{\rm{aMLE}}$ & p-value &
    $\wh\xi_{\rm{Normal}}$ & p-value &
    $\wh\xi_{\rm{Gamma}}$ & p-value \\ \hline
    Participation & -0.133 & $<$1e-3* & -0.130 & $<$1e-3* & -0.130 & $<$1e-3*  \\
    Age           &  0.070 & $<$1e-3* &  0.065 & $<$1e-3* &  0.065 & $<$1e-3*  \\
    Age2          & -0.008 & $<$1e-3* & -0.008 & $<$1e-3* & -0.008 & $<$1e-3*  \\
    Education     &  0.042 & $<$1e-3* &  0.042 & $<$1e-3* &  0.042 & $<$1e-3*  \\
    Youngkids     &  0.010 & 0.695    &  0.011 & 0.657    &  0.011 & 0.654     \\
    Oldkids       &  0.019 & 0.147    &  0.023 & 0.091    &  0.022 & 0.095     \\
    Foreign       & -0.079 & 0.017*   & -0.060 & 0.064    & -0.062 & 0.053     
    \end{tabular}
\end{table}

\clearpage
\pagenumbering{arabic}
{\centering
\section*{SUPPLEMENTARY MATERIAL}}
\setcounter{equation}{0}\renewcommand{\theequation}{S.\arabic{equation}}
\setcounter{subsection}{0}\renewcommand{\thesubsection}{S.\arabic{subsection}}
\setcounter{table}{0}\renewcommand{\thetable}{S.\arabic{table}}

\subsection{Derivation of the Efficiency Bound of Marginal Effect Estimation}\label{sec:derive1}

Consider an arbitrary parametric submodel
\be\label{eq:submodel}
f_{\X,Y}(y,\x,\btheta)
=
f_{\X}(\x,\ba)\frac{\exp\left\{y\bb\trans\x+c(y,\bg)\right\}}
{\int\exp\left\{y\bb\trans\x+c(y,\bg)\right\}d\mu(y)},
\ee
where $\btheta = (\ba\trans, \bb\trans, \bg\trans)\trans$. We can
verify that the score functions associated with an arbitrary $\btheta$ are 
$\S_\ba(\x)=\a(\x)$, where $\a(\x)$ can be any function that satisfies
$E\{\a(\X)\}=\0$, 
$\S_\bb(y,\x)=\x\{y-E(Y\mid\bb\trans\x)\}$, 
and $\S_\bg(y,\x)=\a(y)-E\{\a(Y)\mid\bb\trans\x\}$, where $\a(y)$ can
be any function.
We can verify that 
$\bxi=\bb E\{v(\bb\trans\X)\}$ and
\bse
\frac{\partial\bxi}{\partial\ba\trans}&=&\bb E\{v(\bb\trans\X)\a\trans(\X)\}, \\
\frac{\partial\bxi}{\partial\bb\trans}&=&E\{v(\bb\trans\X)\}\I+\bb 
E[\X\trans\{Y-E(Y\mid\bb\trans\X)\}^3],\\
\frac{\partial\bxi}{\partial\bg\trans}&=&\bb E
\left([\a(Y)-E\{\a(Y)\mid\bb\trans\X\}]\trans\{Y-E(Y\mid\bb\trans\X)\}^2\right).
\ese
Thus, following \cite{bickel1998} and \cite{tsiatis2006}, a
possible influence function is 
\bse
&&\bphi(y,\x)\\
&=&\bb v(\bb\trans\x)-\bb E\{v(\bb\trans\X)\}
+\bb \{y-E(Y\mid\bb\trans\x)\}^2\\
&&-\bb
E[\{Y-E(Y\mid\bb\trans\x)\}^2\mid\bb\trans\x]
+\B E\{v(\bb\trans\X)\}[\b(y,\x)-E\{\b(Y,\x)\mid\x\}],
\ese
where 
\be\label{eq:B}
\B=\left(E[\{Y-E(Y\mid\bb\trans\X)\}\b(Y,\X)\X\trans]\right)^{-1},
\ee
and $\b(y,\x)$ is such that 
\be\label{eq:b}
E\{\b(y,\X)\mid
y\}=E[E\{\b(Y,\X)\mid\X\}\mid y].
\ee
We can verify that
$\partial\bxi/\partial\btheta\trans=E(\bphi\S_\btheta\trans)$
where
$\S_\btheta=(\S_\ba\trans, \S_\bb\trans, \S_\bg\trans)\trans$. Now, 
the tangent space is $\cal T=\cal T_\ba\oplus(\cal T_\bb+\cal T_\bg)$
where
\bse
\cal T_\ba&=&[\a(\x): E\{\a(\X)\}=\0],\\
\cal T_\bb&=&[\M\x\{y-E(Y\mid\bb\trans\x)\}: \forall\M],\\
\cal T_\bg&=&[\a(y)-E\{\a(Y)\mid\bb\trans\x\}:\forall \a(y)].
\ese
Let
\be\label{eq:A1}
\M&=&
E\left(
\B E\{v(\bb\trans\X)\}\b(Y,\X)\X\trans\{Y-E(Y\mid\bb\trans\X)\}
-2\bb\X\trans Y v(\bb\trans\X)\right.\n\\
&&\left.-\a(Y)\{Y-E(Y\mid\bb\trans\X)\}\X\trans
\right)
\left[E\{\X\X\trans v(\bb\trans\X)\}\right]^{-1},
\ee
where $\a(y)$ is such that
\be\label{eq:a}
&&E[E\{\a(Y) \mid \X\}\mid y]-\a(y) \n\\
&=& 2\bb E[ yE(Y\mid\X)-
    E\{YE(Y\mid\X) \mid \X\}\mid y]+\M E[\X\{ y - E( Y \mid \X)\}\mid y].
\ee
Then the efficient influence function is
\bse
\bphi\eff(y,\x)&=&\bb v(\bb\trans\x)-\bb E\{v(\bb\trans\X)\}
+\bb \{y^2-E(Y^2\mid\bb\trans\x)\}\n\\
&&+\M\x\{y-E(Y\mid\bb\trans\x)\}
+\a(y)-E\{\a(Y)\mid\bb\trans\x\},
\ese
where $\a(y)$ satisfies (\ref{eq:a}), $\M$ is given in (\ref{eq:A1}), 
$\b(y,\x)$ satisfies (\ref{eq:b}) and $\B$ is given in (\ref{eq:B}).

Note that once we have $\a(y)$, we can let 
$\b(y,\x)=2\bb yE(Y\mid\bb\trans\x)+\M\x y+\a(y)$ and it satisfies
(\ref{eq:b}). Using this specific $\b(y,\x)$ function, 
we obtain that $\M$ and $\B$ need to satisfy
\bse
\B=\left(E[\{Y-E(Y\mid\bb\trans\X)\}\{
2\bb YE(Y\mid\bb\trans\X)+\M\X Y+\a(Y)
\}\X\trans]\right)^{-1}
\ese
and
\bse 
\M E\{\X\X\trans v(\bb\trans\X)\}
=E\{v(\bb\trans\X)\}\I
-E\left[2\bb\X\trans Y v(\bb\trans\X) +\a(Y)
\{Y-E(Y\mid\bb\trans\X)\}\X\trans 
\right].
\ese
This leads to 
\be\label{eq:M}
\M&=&\left(E\{v(\bb\trans\X)\}\I
-E\left[2\bb\X\trans Y v(\bb\trans\X) +\a(Y)
\{Y-E(Y\mid\bb\trans\X)\}\X\trans 
\right]\right) \n\\
&&\times\left[E\{\X\X\trans v(\bb\trans\X)\}\right]^{-1},
\ee
where $\a(y)$ satisfies (\ref{eq:a}).

Gathering the above derivations and results, we obtain 
the summary description of the efficient influence function as
\bse
\bphi\eff(y,\x)
= \bb v(\bb\trans\x)-\bb E\{v(\bb\trans\X)\}
+\bb y^2+\M\x y+\a(y) 
-E\{\bb Y^2+\M\x Y+\a(Y)\mid\x\},
\ese
where $\a(y)$ satisfies (\ref{eq:a}) and $\M$ is given in (\ref{eq:M}).
Obviously, the variance of the efficient influence function,
i.e. $\var\{\bphi\eff(Y,\X)\}$, is the efficiency bound in
estimating $\bxi$.

\subsection{Derivation of the Efficiency Bound of Quantile Effect Estimation}\label{sec:derive2}

Recall that for an arbitrary parametric submodel in
(\ref{eq:submodel}), we have already derived the corresponding score
functions
$\S_\ba(\x)$, 
$\S_\bb(y,\x)$ and $\S_\bg(y,\x)$.
Now for notational brevity, let 
$\nu\equiv\bb\trans\x$,
$q(\nu)\equiv Q_\tau(Y\mid\nu)$, $\epsilon\equiv\tau-I\{Y<q(\nu)\}$, 
and $\epsilon'_\nu=-\delta\{q(\nu)-Y\}q'(\nu)$.
Write $f(y,\nu)\equiv f_{Y\mid\X}(y,\nu)$.
 Using the quantile definition,
we further have
\bse 
E(\epsilon\mid\nu)&=&0,\\
E( \epsilon Y\mid\nu) 
&=&f\{q(\nu),\nu\}q'(\nu),\\
\frac{E(\epsilon Y^2\mid\nu)}{f\{q(\nu),\nu\}}
&=&2q(\nu)q'(\nu)+q'(\nu)^2[\nu+c'\{q(\nu)\}]+q''(\nu),
\ese 
which can be verified based on
\bse
\tau&=&E[I\{Y\le q(\nu)\}\mid\nu],\\
q'(\nu)
&=&\frac{E\left( [\tau-I\{Y<q(\nu)\}]Y\mid\nu\right) }
{f\{q(\nu),\nu\}},\\
q''(\nu)
&=&\frac{E\left([\tau-I\{Y<q(\nu)\}]Y^2\mid\nu\right)}{f\{q(\nu),\nu\}}
-q(\nu)q'(\nu)-q'(\nu)[q(\nu)+q'(\nu)\nu+q'(\nu)c'\{q(\nu)\}].
\ese
We also have
\bse
\frac{\partial q(\nu)}{\partial\bg}
&=&\frac{E\left\{\epsilon\a(Y)\mid\nu\right\}}{f\{q(\nu),\nu\}},\\
\frac{\partial q'(v)}{\partial\bg}
&=&\frac{E\left\{\epsilon Y\a(Y)\mid\nu\right\}}{f\{q(\nu),\nu\}}-\a\{q(\nu)\}q'(\nu)
-\frac{\partial q(\nu)}{\partial\bg}\left[q(\nu)+q'(\nu)\nu+q'(\nu)c'\{q(\nu)\}\right].
\ese
Note that
$\bta_\tau=\bb E\{q'(\nu)\}$.
We can verify that
\bse
\frac{\partial\bta_\tau}{\partial\ba\trans}
&=&\bb E\{q'(\nu)\a\trans(\X)\}, \\
\frac{\partial\bta_\tau}{\partial\bb\trans}
&=&E\{q'(\nu)\}\I+\bb E[r(Y,\nu)\X\trans\{Y-E(Y\mid\nu)\}],\\
\frac{\partial\bta_\tau}{\partial\bg\trans}
&=&\bb E\left(r(Y,\nu)[\a\trans(Y)-E\{\a\trans(Y)\mid\X\}]\right),
\ese
where
\be\label{eq:r}
r(Y,\nu)\equiv\frac{\epsilon Y+\epsilon'_\nu-\epsilon[q(\nu)+q'(\nu)\nu+q'(\nu) c'\{q(\nu)\}]}{f\{q(\nu),\nu\}}.
\ee
Hence, a possible influence function is
\bse
\bphi(y,\x)
&=&\bb q'(\nu)-\bb E\{q'(\nu)\}+\bb [r(Y,\nu)-E\{r(Y,\nu)\mid\nu\}]\\
&&+E\{q'(\nu)\}\B[\b(y,\x)-E\{\b(Y,\x)\mid\x\}],
\ese
where 
\be\label{eq:Btau}
\B\equiv\left(E[\b(Y,\X)\X\trans\{Y-E(Y\mid\nu)\}]\right)^{-1},
\ee
and $\b(y,\x)$ is such that 
\be\label{eq:btau}
E\{\b(y,\X)\mid
y\}=E[E\{\b(Y,\X)\mid\X\}\mid y].
\ee
We can verify that
$\partial\bta_\tau/\partial\btheta\trans=E(\bphi\S_\btheta\trans)$.
Note that we have derived the tangent space $\cal T$ in Section \ref{sec:xibound}.
Let $\brho(y,\x)\equiv E\{q'(\nu)\}\B\b(y,\x)+\bb r(y,\nu)-\a(y)$,
$v(\nu)\equiv E(Y^2\mid\nu)-\{E(Y\mid\nu)\}^2$, and
\be\label{eq:M1}
\M_1\equiv
E\left[
\brho(Y,\X)\X\trans \{Y-E(Y\mid\nu)\}\right]
\left[E\{\X\X\trans v(\nu)\}\right]^{-1},
\ee
where $\a(y)$ is such that
\be\label{eq:atau}
E[E\{\a(Y) \mid \X\}\mid y]-\a(y)
=-\bb E\{r(y,\nu)|y\} + \M_1 E[\X\{y-E(Y|\X)\}|y].
\ee
Then the efficient influence function is
\be\label{eq:efftau}
\bphi\eff(y,\x)
=\bb q'(\nu)-\bb E\{q'(\nu)\} +\M_1\x y+\a(y) - E\{\M_1\x Y+\a(Y)|\x\},
\ee
where $\a(y)$ satisfies (\ref{eq:atau}), $\M_1$ is given in (\ref{eq:M1}), 
$\b(y,\x)$ satisfies (\ref{eq:btau}) and $\B$ is given in (\ref{eq:Btau}).
Note that once we have $\a(y)$, we can let 
$\b(y,\x)=
\bb r(y,\nu)-\M_1\x y-\a(y)$ and it satisfies
(\ref{eq:btau}). Using this specific $\b(y,\x)$ function, 
we obtain that $\M_1$ and $\B$ need to satisfy
\bse
\B=\left(E[
\{\bb r(Y,\nu)-\M_1\X Y-\a(Y)\}\X\trans\{Y-E(Y\mid\nu)\}]\right)^{-1}
\ese
and
\bse
\M_1 E\{\X\X\trans v(\nu)\} 
=E\{q'(\nu)\}\I+\bb E\{\X\trans q''(\nu)\}-E[\a(Y)\X\trans\{Y-E(Y|\X)\}].
\ese
This leads to
\be\label{eq:Mtau}
\M_1
=(E\{q'(\nu)\}\I+\bb E\{\X\trans q''(\nu)\}-E[\a(Y)\X\trans\{Y-E(Y|\X)\}])[E\{\X\X\trans v(\nu)\}]^{-1},
\ee
where $\a(y)$ satisfies (\ref{eq:atau}).

Thus, in summary, the efficient influence function for estimating
$\bta_\tau$ is given in (\ref{eq:efftau}),
where $\a(y)$ satisfies (\ref{eq:atau}), $r(Y,\nu)$ is given in \eqref{eq:r},
and $\M_1$ is given in (\ref{eq:Mtau}). 

\subsection{Lemmas}

We first prove several lemmas that will be useful. 
\begin{Lem} \label{lem:bnorm}
Let $U\equiv\{\u\in\mR^m:\|\u\|_2=1\}$ and $\u\in U$. 
Under Conditions \ref{con:order}-\ref{con:space},
$\|\B(\cdot)\trans\u\|_2\asymp h^{1/2}$, $\|\B(\cdot)\trans\u\|_1= O(h^{1/2})$
and $\|\B(\cdot)\trans\e_k\|_1\asymp h$.
\end{Lem}

\begin{proof}
We recall Lemma 1 in the supplement of \cite{jmc2020},
which is a direct result from Theorem 5.4.2 on page 145 of \cite{dl1993}.
For each spline $\sum_{k=1}^m \gamma_k B_k(y)$ and $1\leq p \leq\infty$,
there exists a constant $C_r>0$ such that
\bse
C_r \lVert \bg' \rVert_p
\leq \left\lVert \sum_{k=1}^m \gamma_k B_k(\cdot) \right\rVert_p
\leq \lVert \bg' \rVert_p,
\ese
where $\bg'\equiv\{\gamma_k\{(t_k - t_{k-r})/r\}^{1/p}, k=1,\dots,m\}\trans$.
This implies $\|\B(\cdot)\trans\u\|_p\asymp\|\u'\|_p$ where
\bse
\|\u'\|_p
= \left[\sum_{k=1}^m \left\{u_k\left(\frac{t_k - t_{k-r}}{r}\right)^{1/p}\right\}^p\right]^{1/p}
\asymp h^{1/p}\left(\sum_{k=1}^m u_k^p\right)^{1/p}
= h^{1/p}\|\u\|_p
\ese
by Conditions \ref{con:knots} and \ref{con:space}.
Hence we have $\|\B(\cdot)\trans\u\|_2\asymp h^{1/2}$.
In addition, note that $\|\u\|_1=O(m^{1/2})\asymp O(h^{-1/2})$ under Conditions \ref{con:knots} and \ref{con:space},
which leads to
\bse
\|\B(\cdot)\trans\u\|_1 \asymp h\times O(h^{-1/2})=O(h^{1/2}).
\ese
Since $\|\e_k\|_1=1$, we get $\|\B(\cdot)\trans\e_k\|_1\asymp h$.
\end{proof}

\begin{Lem} \label{lem:becov}
Let $U\equiv\{\u\in\mR^m:\|\u\|_2=1\}$ and $\u\in U$. 
Under Conditions \ref{con:bdd}-\ref{con:space},
\bse
E[\{\B(Y)\trans\u\}^2|\x] &\asymp& h, \\
E\{\B(Y)\trans\u|\x\} &=& O(h^{1/2}), \\
E\{\B(Y)\trans\e_k|\x\} &\asymp& h, \\
\cov\{Y,\B(Y)\trans\u|\x\} &=& O(h^{1/2})
\ese
uniformly in $\x$. For fixed $\bb$ and $\bg$,
we have the same result under the density $\fyx^*(y|\x,\bb,\bg)$.
\end{Lem}

\begin{proof}
The result for $\fyx^*(y|\x,\bb,\bg)$ can be obtained similarly so we omit the proof.
We write $E[\{\B(Y)\trans\u\}^2|\x]$ as
\bse
E[\{\B(Y)\trans\u\}^2|\x]
&=& \int_0^1\left\{\B(y)\trans\u\right\}^2\fyx(y|\x)dy \\
&=& \left\lVert\B(\cdot)\trans\u\{\fyx(\cdot|\x)\}^{1/2}\right\rVert_2^2.
\ese
Note that under Condition \ref{con:bdd},
there exist constants $0<c_f\leq C_f<\infty$ such that
\bse
c_f\leq\inf_{y\in[0,1],\x\in\mx}\fyx(y|\x)
\leq \sup_{y\in[0,1],\x\in\mx}\fyx(y|\x)\leq C_f,
\ese
which implies
\bse
c_f\left\lVert\B(\cdot)\trans\u\right\rVert_2^2
&\leq& \inf_{\x\in\mx}\left\lVert\B(\cdot)\trans\u\{\fyx(\cdot|\x)\}^{1/2}\right\rVert_2^2 \\
&\leq& \sup_{\x\in\mx}\left\lVert\B(\cdot)\trans\u\{\fyx(\cdot|\x)\}^{1/2}\right\rVert_2^2 \\
&\leq& C_f\left\lVert\B(\cdot)\trans\u\right\rVert_2^2.
\ese
Then by Lemma \ref{lem:bnorm}, we have $E[\{\B(Y)\trans\u\}^2|\x]\asymp h$.
Further, $E\{\B(Y)\trans\u|\x\}$ satisfies
\bse
E\{\B(Y)\trans\u|\x\}
= \int_0^1 \B(y)\trans\u \fyx(y|\x)dy
\leq C_f\left\lVert\B(\cdot)\trans\u\right\rVert_1.
\ese
By Lemma \ref{lem:bnorm}, the above implies $E\{\B(Y)\trans\u|\x\}=O(h^{1/2})$.
Similarly, we get $E\{\B(Y)\trans\e_k|\x\}\asymp h$.
Now, $\cov\{Y,\B(Y)\trans\u|\x\}$ satisfies
\bse
\cov\{Y,\B(Y)\trans\u|\x\}
&=& \int_0^1\{y-E(Y|\x)\}\B(y)\trans\u\fyx(y|\x)dy \\
&\leq& \|\B(\cdot)\trans\u\|_1\sup_{y\in[0,1],\x\in\mx}|y-E(Y|\x)|\fyx(y|\x) \\
&=& O(h^{1/2}).
\ese
The last equality holds
because $\sup_{y\in[0,1],\x\in\mx}|y-E(Y|\x)|\fyx(y|\x)$ is bounded by Condition \ref{con:bdd}
and $\|\B(\cdot)\trans\u\|_1=O(h^{1/2})$ by Lemma \ref{lem:bnorm}.
\end{proof}

\begin{Lem} \label{lem:bvar}
Let $U\equiv\{\u\in\mR^m:\|\u\|_2=1\}$ and $\u\in U$. 
Under Conditions \ref{con:bdd}-\ref{con:space},
\bse
\sup_{\u\in U}\var\{\B(Y)\trans\u|\x\}\asymp h
\ese
uniformly in $\x$. For fixed $\bb$ and $\bg$,
we have $\sup_{\u\in U}\var^*\{\B(Y)\trans\u|\x,\bb,\bg\}\asymp h$ uniformly in $\x$.
\end{Lem}

\begin{proof}
We omit the proof for $\var^*\{\B(Y)\trans\u|\x,\bb,\bg\}$
because we can obtain the result by a similar way.
The variance $\var\{\B(Y)\trans\u|\x\}$ satisfies
\bse \label{eq:varb}
\var\{\B(Y)\trans\u|\x\}
&=& E[\{\B(Y)\trans\u\}^2|\x] - [E\{\B(Y)\trans\u|\x\}]^2 \\
&=& h - O(h) \\
&=& O(h)
\ese
uniformly in $\x$ by Lemma \ref{lem:becov}.
Similarly, we get $\var\{\B(Y)\trans\e_k|\x\} \asymp h-h^2\asymp h$ by Lemma \ref{lem:becov}.
Then since $\e_k\in U$ and $U$ is compact, we get
\bse
\sup_{\u\in U}\var\{\B(Y)\trans\u|\x\}\asymp h
\ese
uniformly in $\x$ under Condition \ref{con:bdd}.
\end{proof}

\begin{Lem} \label{lem:deboor}
Let $g_1(\cdot)$ be a function such that $\|g_1(\cdot)\|_1<\infty$.
Under Conditions \ref{con:bdd}-\ref{con:deboor},
\bse
\sup_{y\in[0,1],\x\in\mx}|\fyx^*(y,\x,\bb_0,\bg_0)-\fyx(y|\x)|&=&O(h^q), \\
\sup_{\x\in\mx}|E^*\{g_1(Y)|\x,\bb_0,\bg_0\}-E\{g_1(Y)|\x\}| &=& O(h^q)\|g_1(\cdot)\|_1.
\ese
In addition, if $g_2(\cdot)$ satisfies $\|g_2(\cdot)\|_1<\infty$ and $\|g_1(\cdot)g_2(\cdot)\|_1<\infty$,
\bse
&& \sup_{\x\in\mx}|\cov^*\{g_1(Y),g_2(Y)|\x,\bb_0,\bg_0\}-\cov\{g_1(Y),g_2(Y)|\x\}| \\
&=& O(h^q)\{\|g_1(\cdot)g_2(\cdot)\|_1+\|g_1(\cdot)\|_1\|g_2(\cdot)\|_1\}.
\ese
\end{Lem}

\begin{proof}
Conditions \ref{con:bdd}-\ref{con:deboor} implies there exists a constant $0<C_f^*<\infty$ such that
\bse
\sup_{y\in[0,1],\x\in\mx}|\fyx^*(y,\x,\bb_0,\bg_0)-\fyx(y|\x)| \leq C_f^* h^q.
\ese
This implies
\bse
&& \sup_{\x\in\mx}|E^*\{g_1(Y)|\x,\bb_0,\bg_0\}-E\{g_1(Y)|\x\}| \\
&=& \sup_{\x\in\mx}\left|\int_0^1 g_1(y)\{\fyx^*(y,\x,\bb_0,\bg_0)-\fyx(y|\x)\}dy\right| \\
&\leq& \|g_1(\cdot)\|_1\sup_{y\in[0,1],\x\in\mx}|\fyx^*(y,\x,\bb_0,\bg_0)-\fyx(y|\x)| \\
&\leq& C_f^* h^q\|g_1(\cdot)\|_1 \\
&=& O(h^q)\|g_1(\cdot)\|_1.
\ese
Similarly, we have
\bse
\sup_{\x\in\mx}|E^*\{g_1(Y)g_2(Y)|\x,\bb_0,\bg_0\}-E\{g_1(Y)g_2(Y)|\x\}|
 = O(h^q)\|g_1(\cdot)g_2(\cdot)\|_1.
\ese
Now, note that under Condition \ref{con:bdd}, there exists a constant $0<C_f<\infty$ such that
\bse
\sup_{y\in[0,1],\x\in\mx}\fyx(y|\x) \leq C_f.
\ese
Then we get
\bse
&& |E^*\{g_1(Y)|\x,\bb_0,\bg_0\}E^*\{g_2(Y)|\x,\bb_0,\bg_0\}-E\{g_1(Y)|\x\}E\{g_2(Y)|\x\}| \\
&=& |O(h^q)E\{g_1(Y)|\x\}\|g_2(\cdot)\|_1 + O(h^q)\|g_1(\cdot)\|_1 E\{g_2(Y)|\x\} + O(h^{2q})\|g_1(\cdot)\|_1\|g_2(\cdot)\|_1| \\
&\leq& O(h^q)\{2C_f+O(h^q)\}\|g_1(\cdot)\|_1\|g_2(\cdot)\|_1
\ese
uniformly in $\x$ by Condition \ref{con:bdd}. Therefore,
\bse
&& \sup_{\x\in\mx}|\cov^*\{g_1(Y),g_2(Y)|\x,\bb_0,\bg_0\}-\cov\{g_1(Y),g_2(Y)|\x\}| \\
&\leq& \sup_{\x\in\mx}|E^*\{g_1(Y)g_2(Y)|\x,\bb_0,\bg_0\}-E\{g_1(Y)g_2(Y)|\x\}| \\
&&+ \sup_{\x\in\mx}|E^*\{g_1(Y)|\x,\bb_0,\bg_0\}E^*\{g_2(Y)|\x,\bb_0,\bg_0\}-E\{g_1(Y)|\x\}E\{g_2(Y)|\x\}| \\
&=& O(h^q)\{\|g_1(\cdot)g_2(\cdot)\|_1+\|g_1(\cdot)\|_1\|g_2(\cdot)\|_1\}.
\ese
\end{proof}

\begin{Lem} \label{lem:bbbg}
Let $g_1(\cdot)$ be a function which satisfies $\|g_1(\cdot)\|_1<\infty$.
Let $\bb^*, \bg^*$ satisfy $\|\bb^*-\bb_0\|_2=o_p(1)$ and $\|\bg^*-\bg_0\|_2=o_p(1)$.
Under Conditions \ref{con:bdd}-\ref{con:deboor},
\bse
\sup_{y\in[0,1],\x\in\mx}|\fyx^*(y,\x,\bb^*,\bg^*)-\fyx(y|\x)|&=&o_p(1), \\
\sup_{\x\in\mx}|E^*\{g_1(Y)|\x,\bb^*,\bg^*\}-E\{g_1(Y)|\x\}| &=& o_p(1)\|g_1(\cdot)\|_1.
\ese
In addition, if $g_2(\cdot)$ satisfies $\|g_2(\cdot)\|_1<\infty$ and
$\|g_1(\cdot)g_2(\cdot)\|_1<\infty$, then
\bse
&& \sup_{\x\in\mx}|\cov^*\{g_1(Y),g_2(Y)|\x,\bb^*,\bg^*\}-\cov\{g_1(Y),g_2(Y)|\x\}| \\
&=& o_p(1)\{\|g_1(\cdot)g_2(\cdot)\|_1+\|g_1(\cdot)\|_1\|g_2(\cdot)\|_1\}.
\ese
\end{Lem}

\begin{proof}
The result can be obtained similarly to Lemma \ref{lem:deboor}
by noting that under Conditions \ref{con:bdd}-\ref{con:deboor}, we have
\bse
&& \sup_{y\in[0,1],\x\in\mx}|\fyx^*(y,\x,\bb^*,\bg^*)-\fyx(y|\x)| \n\\
&\leq& \sup_{y\in[0,1],\x\in\mx}|\fyx^*(y,\x,\bb^*,\bg^*)-\fyx^*(y,\x,\bb_0,\bg_0)| \n\\
&& +\sup_{y\in[0,1],\x\in\mx}|\fyx^*(y,\x,\bb_0,\bg_0)-\fyx(y|\x)| \n\\
&=& o_p(1).
\ese
\end{proof}

\subsection{Proof of Proposition \ref{pro:bg}}

First, we will show that $\|\wh\bg(\bb_0)-\bg_0\|_2=O_p\{(nh)^{-1/2}\}$.
Notice that the Hessian of the loglikelihood with respect to $\bg$
\bse
\frac{\partial l(\bb_0,\bg)}{\partial\bg\partial\bg\trans}
= -\sumi\cov^*\{\B(Y),\B(Y)|\x_i,\bb_0,\bg\}
\ese
is negative definite.
Hence, it suffices to show the existence of a local maximizer inside a ball
centered at $\bg_0$ with radius of order $(nh)^{-1/2}$.
We will show this by proving for any $\epsilon>0$, there exists a constant $C>0$ such that
\be \label{eq:lbgpr}
\pr\left[l(\bb_0,\bg_0)>\sup_{\|\v\|_2=C}l\{\bb_0,\bg_0+(nh)^{-1/2}\v\}\right]\geq 1-2\epsilon
\ee
for sufficiently large $n$. 

Let $\v\in\mR^m$ be an arbitrary vector with $\|\v\|_2=C$ and
$\bg^*\equiv\alpha\{\bg_0+(nh)^{-1/2}\v\}+(1-\alpha)\bg_0$ for some $\alpha\in(0,1)$.
By the Taylor expansion, we have
\be \label{eq:lbg}
l\{\bb_0,\bg_0+(nh)^{-1/2}\v\} - l(\bb_0,\bg_0)
= (nh)^{-1/2}\frac{\partial l(\bb_0,\bg_0)}{\partial\bg\trans}\v
+ \frac{1}{2}(nh)^{-1}\v\trans\frac{\partial^2 l(\bb_0,\bg^*)}{\partial\bg\partial\bg\trans}\v.
\ee
We first analyze $\partial l(\bb_0,\bg_0)/\partial\bg$. 
Let $\u\in U$ where $U\equiv\{\u\in\mR^m:\|\u\|_2=1\}$, then
\be \label{eq:dldg}
n^{-1}\frac{\partial l(\bb_0,\bg_0)}{\partial\bg\trans}\u
&=& n^{-1}\sumi[\B(y_i)\trans\u-E^*\{\B(Y)\trans\u|\x_i,\bb_0,\bg_0\}] \n\\
&=& n^{-1}\sumi[\B(y_i)\trans\u-E\{\B(Y)\trans\u|\x_i\}] \n\\
&&+ n^{-1}\sumi[E\{\B(Y)\trans\u|\x_i\}-E^*\{\B(Y)\trans\u|\x_i,\bb_0,\bg_0\}].
\ee
The first term of \eqref{eq:dldg} satisfies
\bse
&& E\left\{\left(n^{-1}\sumi\left[\B(Y_i)\trans\u - E\{\B(Y)\trans\u|\X_i\}\right]\right)^2\right\} \n\\
&=& n^{-1}E\left(\left[\B(Y)\trans\u-E\{\B(Y)\trans\u|\X\}\right]^2\right) \n\\
&=& n^{-1}E[\var\{\B(Y)\trans\u|\X\}].
\ese
Furthermore, by Lemma \ref{lem:bvar} we have $\sup_{\u\in U} E[\var\{\B(Y)\trans\u|\X\}] \asymp h$,
which implies
\be \label{eq:dldg1}
\left\|n^{-1}\sumi[\B(y_i)-E\{\B(Y)|\x_i\}]\right\|_2
&=& \sup_{\u\in U} \left|n^{-1}\sumi[\B(y_i)\trans\u-E\{\B(Y)\trans\u|\x_i\}]\right| \n\\
&\asymp_p& n^{-1/2}h^{1/2}.
\ee
Now, the second term of \eqref{eq:dldg} satisfies
\bse
&& \left|n^{-1}\sumi[E\{\B(Y)\trans\u|\x_i\}-E^*\{\B(Y)\trans\u|\x_i,\bb_0,\bg_0\}]\right| \\
&\leq& n^{-1}\sumi \sup_{\x\in\mx}|E\{\B(Y)\trans\u|\x\}-E^*\{\B(Y)\trans\u|\x,\bb_0,\bg_0\}| \\
&=& O(h^q)\|\B(\cdot)\trans\u\|_1 \\
&=& O(h^{q+1/2})
\ese
by Lemma \ref{lem:bnorm} and \ref{lem:deboor}. Hence, we get
\be \label{eq:dldg2}
\left\|n^{-1}\sumi[E\{\B(Y)|\x_i\}-E^*\{\B(Y)|\x_i,\bb_0,\bg_0\}]\right\|_2
&=& O(h^{q+1/2}) \n\\
&=& o(n^{-1/2}h^{1/2})
\ee
under Conditions \ref{con:knots} and \ref{con:space}.
Therefore, by \eqref{eq:dldg}, \eqref{eq:dldg1} and \eqref{eq:dldg2}, we get
\bse
\left\|\frac{\partial l(\bb_0,\bg_0)}{\partial\bg}\right\|_2\asymp_p (nh)^{1/2},
\ese
hence for any $\epsilon>0$, there exists a constant $0<C_1<\infty$ such that
\be \label{eq:dldgpr}
\pr\left\{\left\|\frac{\partial l(\bb_0,\bg_0)}{\partial\bg}\right\|_2
\leq C_1(nh)^{1/2}\right\}\geq 1-\epsilon.
\ee
Next, to analyze $\partial l(\bb_0,\bg^*)/\partial\bg\partial\bg\trans$, we have
\be \label{eq:dldggu}
\u\trans\left\{-n^{-1}\frac{\partial^2 l(\bb_0,\bg^*)}{\partial\bg\partial\bg\trans}\right\}\u
&=& n^{-1}\sumi\var^*\{\B(Y)\trans\u|\x_i,\bb_0,\bg^*\} \n\\
&=& n^{-1}\sumi[\var^*\{\B(Y)\trans\u|\x_i,\bb_0,\bg^*\}-\var\{\B(Y)\trans\u|\x_i\}] \n\\
&&+ \left(n^{-1}\sumi\var\{\B(Y)\trans\u|\x_i\}-E[\var\{\B(Y)\trans\u|\X\}]\right) \n\\
&&+ E[\var\{\B(Y)\trans\u|\X\}].
\ee
Note that $\|\bg^*-\bg_0\|_2=O\{(nh)^{-1/2}\}=o(1)$ by construction.
Then by Lemma \ref{lem:bnorm} and \ref{lem:bbbg}, the first term satisfies
\be \label{eq:dldggu1}
&& \left|n^{-1}\sumi[\var^*\{\B(Y)\trans\u|\x_i,\bb_0,\bg^*\}-\var\{\B(Y)\trans\u|\x_i\}]\right| \n\\
&\leq& n^{-1}\sumi\sup_{\x\in\mx}|\var^*\{\B(Y)\trans\u|\x,\bb_0,\bg^*\}-\var\{\B(Y)\trans\u|\x\}| \n\\
&=& o_p(1)\left[\|\{\B(\cdot)\trans\u\}^2\|_1+\|\B(\cdot)\trans\u\|_1^2\right] \n\\
&=& o_p(1)\left\{\|\B(\cdot)\trans\u\|_2^2+\|\B(\cdot)\trans\u\|_1^2\right\} \n\\
&=& o_p(h).
\ee
For the second term of \eqref{eq:dldggu}, first let $w_i\equiv\var\{\B(Y)\trans\u|\x_i\}$.
Note that by Lemma \ref{lem:bvar},
there exists a constant $C$ so that $0\le w_i\le Ch$ and $0\le E(W_i)\le Ch$ for all $i=1,\dots,n$.
Therefore, there exists a constant $C_w$, where $0<C_w<\infty$,
such that $|w_i-E(W_i)|\leq C_w h$ for all $i=1,\dots,n$.
This also directly implies $\var(W_i)\leq C_w^2 h^2$.
Then, by the Bernstein's inequality, we have
\bse
\pr\left\{n^{-1}\sumi w_i-E(W_i)\geq\epsilon\right\}
&\leq& \exp\left\{-\frac{n^2\epsilon^2/2}{\sumi \var(W_i)+C_w h n\epsilon/3}\right\} \\
&\leq& \exp\left\{-\frac{n^2\epsilon^2/2}{n C_w^2 h^2 +C_w h n\epsilon/3}\right\} \\
&\to& 0
\ese
for $\epsilon=h(\log n)^{-1}$ as $n\to\infty$. Note that $\epsilon=o(h)$. Thus, we get
\be \label{eq:dldggu2}
n^{-1}\sumi w_i-E(W_i)
&=& n^{-1}\sumi\var\{\B(Y)\trans\u|\x_i\}-E[\var\{\B(Y)\trans\u|\X\}] \n\\
&=& o_p(h).
\ee
Now for the last term of \eqref{eq:dldggu}, we have
\be \label{eq:v}
\|\bSig_{22}\|_2 
&=& \sup_{\u\in U}\u\trans\bSig_{22}\u \n\\
&=& \sup_{\u\in U}E[\var\{\B(Y)\trans\u|\X\}] \n\\
&\asymp& h
\ee 
by Lemma \ref{lem:bvar}.
Hence Condition \ref{con:bb} implies that $\bSig_{22}$ has all eigenvalues of order $h$.
Moreover, from \eqref{eq:dldggu}, \eqref{eq:dldggu1}, \eqref{eq:dldggu2} and \eqref{eq:v}, we have
\be \label{eq:dldggv}
\u\trans\left\{-n^{-1}\frac{\partial^2 l(\bb_0,\bg^*)}{\partial\bg\partial\bg\trans}\right\}\u = \u\trans\bSig_{22}\u+o_p(h)
\ee
uniformly for all $\u$.
Thus, all eigenvalues of $\partial^2 l(\bb_0,\bg^*)/\partial\bg\partial\bg\trans$ are negative and of order $nh$
with probability approaching 1, which implies there exists a constant $0<C_2<\infty$ such that
\be \label{eq:dldggpr}
\pr\left\{\v\trans\frac{\partial^2 l(\bb_0,\bg^*)}{\partial\bg\partial\bg\trans}\v\leq-C^2C_2nh\right\}\geq 1-\epsilon.
\ee 
Combining \eqref{eq:lbg}, \eqref{eq:dldgpr}, and \eqref{eq:dldggpr}, we get
\bse
l\{\bb_0,\bg_0+(nh)^{-1/2}\v\} - l(\bb_0,\bg_0)
&\leq& C(nh)^{-1/2}C_1(nh)^{1/2} + \frac{1}{2}(nh)^{-1}(-C^2C_2nh) \\
&=& C\left(C_1 - \frac{C_2}{2}C\right) \\
&<& 0
\ese
with probability at least $1-2\epsilon$ when $C>2C_1/C_2$.
This proves \eqref{eq:lbgpr}, hence we have $\|\wh\bg(\bb_0)-\bg_0\|_2=O_p\{(nh)^{-1/2}\}$. 

Now we analyze the asymptotic behavior of $\wh\bg(\bb_0)$.
Note that $\wh\bg(\bb_0)$ maximizes $l(\bb_0,\bg)$,
then by using the Taylor expansion with $\bg^*$ on the line connecting $\wh\bg(\bb_0)$ and $\bg_0$,
we get
\bse
\0
&=& \frac{\partial l\{\bb_0,\wh\bg(\bb_0)\}}{\partial\bg} \\
&=& \frac{\partial l(\bb_0,\bg_0)}{\partial\bg}
+ \frac{\partial^2 l(\bb_0,\bg^*)}{\partial\bg\partial\bg\trans}\{\wh\bg(\bb_0)-\bg_0\}.
\ese
This leads to
\bse
\wh\bg(\bb_0)-\bg_0
&=& \left\{-n^{-1}\frac{\partial^2 l(\bb_0,\bg^*)}{\partial\bg\partial\bg\trans}\right\}^{-1}
n^{-1}\frac{\partial l(\bb_0,\bg_0)}{\partial\bg} \\
&=& \bSig_{22}^{-1}n^{-1}\sumi[\B(y_i)-E\{\B(Y)|\x_i\}] + \r_1,
\ese
where
\bse
\r_1
&\equiv& \left[\left\{-n^{-1}\frac{\partial^2 l(\bb_0,\bg^*)}{\partial\bg\partial\bg\trans}\right\}^{-1} - \bSig_{22}^{-1}\right]
n^{-1}\sumi[\B(y_i)-E\{\B(Y)|\x_i\}] \\
&&+ \left\{-n^{-1}\frac{\partial^2 l(\bb_0,\bg^*)}{\partial\bg\partial\bg\trans}\right\}^{-1}
\left(n^{-1}\frac{\partial l(\bb_0,\bg_0)}{\partial\bg} - n^{-1}\sumi[\B(y_i)-E\{\B(Y)|\x_i\}]\right).
\ese
Note that
\bse
&& \left\|\left\{-n^{-1}\frac{\partial^2 l(\bb_0,\bg^*)}{\partial\bg\partial\bg\trans}\right\}^{-1} - \bSig_{22}^{-1}\right\|_2 \\
&=& \left\|\left\{-n^{-1}\frac{\partial^2 l(\bb_0,\bg^*)}{\partial\bg\partial\bg\trans}\right\}^{-1}
\left\{\bSig_{22} + n^{-1}\frac{\partial^2 l(\bb_0,\bg^*)}{\partial\bg\partial\bg\trans}\right\}\bSig_{22}^{-1}\right\|_2 \\
&\leq& \left\|\left\{-n^{-1}\frac{\partial^2 l(\bb_0,\bg^*)}{\partial\bg\partial\bg\trans}\right\}^{-1}\right\|_2
\left\|\bSig_{22} + n^{-1}\frac{\partial^2 l(\bb_0,\bg^*)}{\partial\bg\partial\bg\trans}\right\|_2\left\|\bSig_{22}^{-1}\right\|_2 \\
&\asymp_p& h^{-1}\times o_p(h)\times h^{-1} \\
&=& o_p(h^{-1}),
\ese
by \eqref{eq:v}, \eqref{eq:dldggv} and Condition \ref{con:bb}.
Furthermore, by \eqref{eq:dldg} and \eqref{eq:dldg2}, we have
\bse
\left\|n^{-1}\frac{\partial l(\bb_0,\bg_0)}{\partial\bg}
- n^{-1}\sumi[\B(y_i)-E\{\B(Y)|\x_i\}]\right\|_2
= o_p(n^{-1/2}h^{1/2}).
\ese
Hence $\r_1$ satisfies
\bse
\|\r_1\|_2
&\leq& o_p(h^{-1})n^{-1/2}h^{1/2} + h^{-1}o_p(n^{-1/2}h^{1/2}) \n\\
&=& o_p\{(nh)^{-1/2}\}
\ese
by \eqref{eq:dldg1} and \eqref{eq:dldggv}.
Then by Condition \ref{con:deboor}, this leads to
\bse
&&\sup_{y\in[0,1]}|\wh c(y,\bb_0)-c(y)| \\
&\le& \sup_{y\in[0,1]}|\B(y)\trans\{\wh\bg(\bb_0)-\bg_0\}|+\sup_{y\in[0,1]}|\B(y)\trans\bg_0-c(y)|\\
&\le& \sup_{y\in[0,1]}\|\B(y)\|_2
\left(\|\bSig_{22}^{-1}\|_2\left\|n^{-1}\sumi[\B(y_i)-E\{\B(Y)|\x_i\}]\right\|_2 + \|\r_1\|_2\right) +O(h^q)\\
&=& O(1)O(h^{-1})O_p(n^{-1/2}h^{1/2}) + o_p\{(nh)^{-1/2}\}+O(h^q)\\
&=& O_p\{(nh)^{-1/2}+h^q\}.
\ese
\qed

\subsection{Proof of Proposition \ref{pro:bb}} \label{proof:bb}

First we will show for any $\epsilon>0$, there exists a constant $C>0$ such that
\be
\pr\left[l\{\bb_0,\wh\bg(\bb_0)\}>
\sup_{\|\v\|_2=C}l\{\bb_0+n^{-1/2}\v,\wh\bg(\bb_0+n^{-1/2}\v)\}\right] \geq 1-3\epsilon
\label{eq:lmax}
\ee
for a sufficiently large $n$.
This implies there exists a local maximizer of $l\{\bb,\wh\bg(\bb)\}$,
say $\wt\bb$, such that $\|\wt\bb-\bb_0\|_2=O_p(n^{-1/2})$.
We will further show the Hessian of $l\{\bb,\wh\bg(\bb)\}$ is negative definite for any $\bb$,
hence $\wt\bb$ is the global maximizer of $l\{\bb,\wh\bg(\bb)\}$,
i.e. $\wh\bb=\wt\bb$, and 
hence $\|\wh\bb-\bb_0\|_2=O_p(n^{-1/2})$.

By the Taylor expansion,
\be
&& l\{\bb_0+n^{-1/2}\v,\wh\bg(\bb_0+n^{-1/2}\v)\} - l\{\bb_0,\wh\bg(\bb_0)\} \n\\
&=& n^{-1/2}\frac{d l\{\bb_0,\wh\bg(\bb_0)\}}{d\bb\trans}\v
+ \frac{1}{2}n^{-1}\v\trans
\frac{d^2 l\{\bb^*,\wh\bg(\bb^*)\}}{d\bb d\bb\trans}\v,
\label{eq:lt2}
\ee
where $\bb^*\equiv\alpha_1(\bb_0+n^{-1/2}\v)+(1-\alpha_1)\bb_0$ for
some $\alpha_1\in(0,1)$. Noting that 
$\partial l\{\bb,\wh\bg(\bb)\}/\partial\bg=\0$ for any $\bb$, we get
\bse
\frac{d l\{\bb_0,\wh\bg(\bb_0)\}}{d\bb}
&=& \frac{\partial l\{\bb_0,\wh\bg(\bb_0)\}}{\partial\bb}
+ \frac{\partial\wh\bg\trans(\bb_0)}{\partial\bb}
\frac{\partial l\{\bb_0,\wh\bg(\bb_0)\}}{\partial\bg} \\
&=& \frac{\partial l\{\bb_0,\wh\bg(\bb_0)\}}{\partial\bb}\\
&=& \frac{\partial l(\bb_0,\bg_0)}{\partial\bb}
+ \frac{\partial^2 l(\bb_0,\bg^*)}{\partial\bb\partial\bg\trans}\{\wh\bg(\bb_0)-\bg_0\},
\ese
where $\bg^*=\alpha_2\wh\bg(\bb_0)+(1-\alpha_2)\bg_0$ for some
$\alpha_2\in(0,1)$. In addition,
\bse
\frac{d^2 l\{\bb^*,\wh\bg(\bb^*)\}}{d\bb d\bb\trans}
&=& \frac{d}{d\bb\trans}
\left[\frac{\partial l\{\bb,\wh\bg(\bb)\}}{\partial\bb}\right]\biggr|_{\bb=\bb^*} \\
&=& \frac{\partial^2 l\{\bb^*,\wh\bg(\bb^*)\}}{\partial\bb\partial\bb\trans}
+ \frac{\partial^2 l\{\bb^*,\wh\bg(\bb^*)\}}{\partial\bb\partial\bg\trans}
\frac{\partial\wh\bg(\bb^*)}{\partial\bb\trans}.
\ese
Now since $\partial l\{\bb,\wh\bg(\bb)\}/\partial\bg=\0$ for any $\bb$,
\bse
\0
&=& \frac{d}{d\bb\trans}\left[\frac{\partial l\{\bb,\wh\bg(\bb)\}}{\partial\bg}\right] \\
&=& \frac{\partial^2 l\{\bb,\wh\bg(\bb)\}}{\partial\bg\partial\bb\trans}
+ \frac{\partial^2 l\{\bb,\wh\bg(\bb)\}}{\partial\bg\partial\bg\trans}
\frac{\partial\wh\bg(\bb)}{\partial\bb\trans}.
\ese
Note that
\bse
\frac{\partial^2 l(\bb,\bg)}{\partial\bg\partial\bg\trans}
= -\sumi\cov^*\{\B(Y),\B(Y)|\x_i,\bb,\bg\}
\ese
is negative definite. Then,
\be
\frac{\partial \wh\bg(\bb)}{\partial \bb\trans}
= -\left[\frac{\partial^2 l\{\bb,\wh\bg(\bb)\}}{\partial\bg\partial\bg\trans}\right]^{-1}
\frac{\partial^2 l\{\bb,\wh\bg(\bb)\}}{\partial\bg\partial\bb\trans}. \label{eq:dgdb}
\ee
Hence, we can write (\ref{eq:lt2}) as
\be
&& l\{\bb_0+n^{-1/2}\v,\wh\bg(\bb_0+n^{-1/2}\v)\} - l\{\bb_0,\wh\bg(\bb_0)\} \n\\
&=& n^{-1/2}\left[\frac{\partial l(\bb_0,\bg_0)}{\partial\bb}
+ \frac{\partial^2 l(\bb_0,\bg^*)}{\partial\bb\partial\bg\trans}
\{\wh\bg(\bb_0)-\bg_0\}\right]\trans\v \label{eq:lt} \\
&&+ \frac{1}{2}\v\trans n^{-1}\Bigg(
\frac{\partial^2 l\{\bb^*,\wh\bg(\bb^*)\}}{\partial\bb\partial\bb\trans} 
- \frac{\partial^2 l\{\bb^*,\wh\bg(\bb^*)\}}{\partial\bb\partial\bg\trans}
\left[\frac{\partial^2 l\{\bb^*,\wh\bg(\bb^*)\}}{\partial\bg\partial\bg\trans}\right]^{-1}
\frac{\partial^2 l\{\bb^*,\wh\bg(\bb^*)\}}{\partial\bg\partial\bb\trans} \Bigg)\v.\n
\ee
In (\ref{eq:lt}), we first analyze $\partial l(\bb_0,\bg_0)/\partial\bb$. We have
\bse
n^{-1}\frac{\partial l(\bb_0,\bg_0)}{\partial\bb}
&=& n^{-1}\sumi\x_i\{y_i-E^*(Y|\x_i,\bb_0,\bg_0)\} \\
&=& n^{-1}\sumi\x_i\{y_i-E(Y|\x_i)\}
+ n^{-1}\sumi\x_i\{E(Y|\x_i)-E^*(Y|\x_i,\bb_0,\bg_0)\}.
\ese
By Conditions \ref{con:bdd} and \ref{con:bSig}, the first term satisfies
\bse
\left\|n^{-1}\sumi\x_i\{y_i-E(Y|\x_i)\}\right\|_2\asymp_p n^{-1/2}.
\ese
For the second term, by Lemma \ref{lem:deboor} we get
\bse
&& \left\|n^{-1}\sumi\x_i\{E(Y|\x_i)-E^*(Y|\x_i,\bb_0,\bg_0)\}\right\|_2 \\
&\leq& n^{-1}\sumi\left\|\x_i\right\|_2\sup_{\x\in\mx}|E(Y|\x)-E^*(Y|\x,\bb_0,\bg_0)| \\
&=& O_p(h^q),
\ese
where the last equality holds because $E(\|\X\|_2)$ is bounded by Condition \ref{con:bdd}.
Note that $h^q=o_p(n^{-1/2})$ under Conditions \ref{con:knots} and \ref{con:space}.
Therefore, we have
\bse
\left\|\frac{\partial l(\bb_0,\bg_0)}{\partial\bb}\right\|_2 \asymp_p n^{1/2},
\ese
and
\be \label{eq:ldb}
\left\|\frac{\partial l(\bb_0,\bg_0)}{\partial\bb}
-\sumi\x_i\{y_i-E(Y|\x_i)\}\right\|_2 = o_p(n^{1/2}).
\ee
Moreover, for any $\epsilon>0$, there exists a constant $0<C_1<\infty$ such that
\be
\pr\left\{\left\|\frac{\partial l(\bb_0,\bg_0)}{\partial\bb}\right\|_2
\leq C_1n^{1/2}\right\} \geq 1-\epsilon. \label{eq:dldb}
\ee
Now we analyze $\partial^2 l(\bb,\bg)/\partial\bb\partial\bg\trans$.
Let $\u_m\in U_m\equiv\{\u\in\mR^m:\|\u\|_2=1\}$, then for fixed $\bb$ and $\bg$ we have
\bse
\left\|-n^{-1}\frac{\partial^2 l(\bb,\bg)}{\partial\bb\partial\bg\trans}\u_m\right\|_2
&=& \left\|n^{-1}\sumi \x_i\cov^*\{Y,\B(Y)\trans\u_m|\x_i,\bb,\bg\}\right\|_2 \\
&\leq& n^{-1}\sumi \left\|\x_i\right\|_2 \sup_{\x\in\mx}|\cov^*\{Y,\B(Y)\trans\u_m|\x,\bb,\bg\}| \\
&=& O_p(h^{1/2})
\ese
by Condition \ref{con:bdd} and Lemma \ref{lem:becov}.
Then using the fact that $U_m$ is compact, taking supremum with respect to $\u_m$ gives
\be \label{eq:dldbgorder}
\left\|\frac{\partial^2 l(\bb,\bg)}{\partial\bb\partial\bg\trans}\right\|_2
= O_p(nh^{1/2}).
\ee
Further recall that $\|\wh\bg(\bb_0)-\bg_0\|_2=O_p\{(nh)^{-1/2}\}$ by Proposition \ref{pro:bg}.
This implies
\bse
\left\|\frac{\partial^2 l(\bb_0,\bg^*)}{\partial\bb\partial\bg\trans}\{\wh\bg(\bb_0)-\bg_0\}\right\|_2
=O_p(n^{1/2}),
\ese
and equivalently, there exists a constant $0<C_2<\infty$ such that
\be \label{eq:dldbgg}
\pr\left\{
\left\|\frac{\partial^2 l(\bb_0,\bg^*)}{\partial\bb\partial\bg\trans}\{\wh\bg(\bb_0)-\bg_0\}\right\|_2
\leq C_2n^{1/2}\right\} \geq 1-\epsilon.
\ee

We next show $-n^{-1}\partial^2 l(\bb_0,\bg^*)/\partial\bb\partial\bg\trans$ converges to $\bSig_{12}$.
Note that $\|\bg^*-\bg_0\|_2=O_p\{(nh)^{-1/2}\}=o_p(1)$.
Then by Condition \ref{con:bdd}, Lemma \ref{lem:bnorm} and \ref{lem:bbbg}, we get
\be \label{eq:dldbgproof}
&&\left\|\left(-n^{-1}\frac{\partial^2 l(\bb_0,\bg^*)}{\partial\bb\partial\bg\trans}
- \bSig_{12} \right)\u_m\right\|_2 \n\\
&\leq& \left\|n^{-1}\sumi \x_i[\cov^*\{Y,\B(Y)\trans\u_m|\x_i,\bb_0,\bg^*\}
-\cov\{Y,\B(Y)\trans\u_m|\x_i\}]\right\|_2 \n\\
&&+ \left\|n^{-1}\sumi \x_i\cov\{Y,\B(Y)\trans\u_m|\x_i\}
-E[\X\cov\{Y,\B(Y)\trans\u_m|\X\}]\right\|_2 \n\\
&\leq& \left(n^{-1}\sumi \|\x_i\|_2\right)o_p(h^{1/2}) + o_p(n^{-1/2}) \n\\
&=& o_p(h^{1/2}).
\ee
The last equality holds because $n^{-1/2}=o(h^{1/2})$ by Conditions \ref{con:knots} and \ref{con:space}.
This implies
\be \label{eq:dldbg}
\left\|-n^{-1}\frac{\partial^2 l(\bb_0,\bg^*)}{\partial\bb\partial\bg\trans}
- \bSig_{12}\right\|_2 = o_p(h^{1/2}).
\ee

Now we show $-n^{-1}\partial^2 l\{\bb^*,\wh\bg(\bb^*)\}/\partial\bb\partial\bb\trans$ converges to $\bSig_{11}$.
First we will show that
$\|\wh\bg(\bb^*)-\bg_0\|_2=O_p\{(nh)^{-1/2}\}$. 
Note that since
\bse
&& \begin{bmatrix}
\partial^2 l\{\bb,\wh\bg(\bb)\}/\partial\bb\partial\bb\trans
& \partial^2 l\{\bb,\wh\bg(\bb)\}/\partial\bb\partial\bg\trans \\
\partial^2 l\{\bb,\wh\bg(\bb)\}/\partial\bg\partial\bb\trans
& \partial^2 l\{\bb,\wh\bg(\bb)\}/\partial\bg\partial\bg\trans
\end{bmatrix} \\
&=& -\sumi \begin{bmatrix}
\x_i\x_i\trans\var^*\{Y|\x_i,\bb,\wh\bg(\bb)\}
& \x_i\cov^*\{Y,\B(Y)|\x_i,\bb,\wh\bg(\bb)\} \\
[\x_i\cov^*\{Y,\B(Y)|\x_i,\bb,\wh\bg(\bb)\}]\trans
& \cov^*\{\B(Y),\B(Y)|\x_i,\bb,\wh\bg(\bb)\}
\end{bmatrix} \\
&=& -\sumi
\cov^*\left\{ \begin{bmatrix}\x_iY \\\B(Y)\end{bmatrix},
\begin{bmatrix}\x_iY \\ \B(Y)\end{bmatrix} \Bigg|\x_i,\bb,\wh\bg(\bb)\right\}
\ese
is negative definite, $\partial^2 l\{\bb,\wh\bg(\bb)\}/\partial\bb\partial\bb\trans$
and $\partial^2 l\{\bb,\wh\bg(\bb)\}/\partial\bg\partial\bg\trans$ are negative definite,
and the Schur complement of $\partial^2 l\{\bb,\wh\bg(\bb)\}/\partial\bg\partial\bg\trans$,
\bse
\frac{\partial^2 l\{\bb,\wh\bg(\bb)\}}{\partial\bb\partial\bb\trans} 
- \frac{\partial^2 l\{\bb,\wh\bg(\bb)\}}{\partial\bb\partial\bg\trans}
\left[\frac{\partial^2 l\{\bb,\wh\bg(\bb)\}}{\partial\bg\partial\bg\trans}\right]^{-1}
\frac{\partial^2 l\{\bb,\wh\bg(\bb)\}}{\partial\bg\partial\bb\trans},
\ese
is also negative definite. Let $\u_p\in U_p\equiv\{\u\in\mR^p:\|\u\|_2=1\}$, then we have
\bse
0 \leq \u_p\trans \frac{\partial^2 l\{\bb,\wh\bg(\bb)\}}{\partial\bb\partial\bg\trans}
\left[-\frac{\partial^2 l\{\bb,\wh\bg(\bb)\}}{\partial\bg\partial\bg\trans}\right]^{-1}
\frac{\partial^2 l\{\bb,\wh\bg(\bb)\}}{\partial\bg\partial\bb\trans}\u_p
< \u_p\trans\left[-\frac{\partial^2 l\{\bb,\wh\bg(\bb)\}}{\partial\bb\partial\bb\trans}\right]\u_p.
\ese
It is easy to see that
\bse
\u_p\trans\left[-\frac{\partial^2 l\{\bb,\wh\bg(\bb)\}}{\partial\bb\partial\bb\trans}\right]\u_p
=\sumi(\x_i\trans\u_p)^2\var^*\{Y|\x_i,\bb,\wh\bg(\bb)\}
\asymp_p n,
\ese
and
\bse
\left\|\frac{\partial^2l\{\bb,\wh\bg(\bb)\}}{\partial\bg\partial\bg\trans}\right\|_2
&=& \sup_{\u_m\in U_m} \u_m\trans\left[-\frac{\partial^2 l\{\bb,\wh\bg(\bb)\}}{\partial\bg\partial\bg\trans}\right]\u_m \\
&=& \sup_{\u_m\in U_m} \sumi\var^*\{\B(Y)\trans\u_m|\x_i,\bb,\wh\bg(\bb)\} \\
&\asymp_p& nh
\ese
by Lemma \ref{lem:bvar}. Hence we have
\be \label{eq:dgdborder}
\left\|\left[\frac{\partial^2 l\{\bb,\wh\bg(\bb)\}}{\partial\bg\partial\bg\trans}\right]^{-1}
\frac{\partial^2 l\{\bb,\wh\bg(\bb)\}}{\partial\bg\partial\bb\trans}\right\|_2=O_p(h^{-1/2}).
\ee
Now, \eqref{eq:dgdb} and the Taylor expansion of $\wh\bg(\bb^*)$ with
$\bb^{**}$ on the line connecting $\bb_0$ and $\bb^*$ gives
\bse
\wh\bg(\bb^*)-\bg_0
&=& \wh\bg(\bb_0) -\bg_0 
-\left[\frac{\partial^2 l\{\bb^{**},\wh\bg(\bb^{**})\}}{\partial\bg\partial\bg\trans}\right]^{-1}
\frac{\partial^2 l\{\bb^{**},\wh\bg(\bb^{**})\}}{\partial\bg\partial\bb\trans}(\bb^*-\bb_0).
\ese
Recall that $\|\wh\bg(\bb_0) -\bg_0\|_2=O_p\{(nh)^{-1/2}\}$ by Proposition \ref{pro:bg}
and $\|\bb^*-\bb_0\|_2=O(n^{-1/2})$ by construction.
These imply $\|\wh\bg(\bb^*)-\bg_0\|_2=O_p\{(nh)^{-1/2}\}$.
Therefore, we get
\bse
&&\left|\u_p\trans\left[-n^{-1}\frac{\partial^2 l\{\bb^*,\wh\bg(\bb^*)\}}{\partial\bb\partial\bb\trans}
- \bSig_{11} \right]\u_p\right| \\
&\leq& \left|n^{-1}\sumi (\x_i\trans\u_p)^2[\var^*\{Y|\x_i,\bb^*,\wh\bg(\bb^*)\}
-\var(Y|\x_i)]\right| \\
&&+ \left|n^{-1}\sumi (\x_i\trans\u_p)^2 \var(Y|\x_i)
-E\{(\X\trans\u_p)^2 \var(Y|\X)\}\right| \\
&=& o_p(1)
\ese
by Condition \ref{con:bdd} and Lemma \ref{lem:bbbg}, which implies
\be \label{eq:dldbb}
\left\|-n^{-1}\frac{\partial^2 l\{\bb^*,\wh\bg(\bb^*)\}}{\partial\bb\partial\bb\trans}
- \bSig_{11} \right\|_2
= o_p(1).
\ee

Below, we show $-n^{-1}d^2l\{\bb^*,\wh\bg(\bb^*)\}/d\bb d\bb\trans$ converges to $\bSig^*$.
Using similar argument to \eqref{eq:dldbgproof}, we get
\be \label{eq:dldbg2}
\left\|-n^{-1}\frac{\partial^2 l\{\bb^*,\wh\bg(\bb^*)\}}{\partial\bb\partial\bg\trans}
- \bSig_{12}\right\|_2 = o_p(h^{1/2}).
\ee
In addition, by the similar proof
to showing $-n^{-1}\partial^2 l(\bb_0,\bg^*)/\partial\bg\partial\bg\trans$
converges to $\bSig_{22}^{-1}$ in Proposition \ref{pro:bg}, we get
\be  \label{eq:dldgg}
\left\|\left[-n^{-1}\frac{\partial^2 l\{\bb^*,\wh\bg(\bb^*)\}}
{\partial\bg\partial\bg\trans}\right]^{-1}-\bSig_{22}^{-1}\right\|_2=o_p(h^{-1}).
\ee
Now, note that
\bse
\frac{d^2l\{\bb,\wh\bg(\bb)\}}{d\bb d\bb\trans}
= \frac{\partial^2 l\{\bb,\wh\bg(\bb)\}}{\partial\bb\partial\bb\trans}
- \frac{\partial^2 l\{\bb,\wh\bg(\bb)\}}{\partial\bb\partial\bg\trans}
\left[\frac{\partial^2 l\{\bb,\wh\bg(\bb)\}}{\partial\bg\partial\bg\trans}\right]^{-1}
\frac{\partial^2 l\{\bb,\wh\bg(\bb)\}}{\partial\bg\partial\bb\trans}.
\ese
Combining \eqref{eq:dldbb}, \eqref{eq:dldbg2}, and \eqref{eq:dldgg}, we have
\be \label{eq:dldbbA}
\left\|-n^{-1}\frac{d^2l\{\bb^*,\wh\bg(\bb^*)\}}{d\bb d\bb\trans}-\bSig^*\right\|_2= o_p(1),
\ee
where $\bSig^*$ is given in Condition \ref{con:bSig}.
This shows the Hessian at $\bb^*$ is negative definite
and its 2-norm is of order $n$ with probability approaching 1.
Hence, there exists a constant $0<C_3<\infty$ such that
\be \label{eq:hess}
\pr\left\{ 
\v\trans\frac{d^2l\{\bb^*,\wh\bg(\bb^*)\}}{d\bb d\bb\trans} \v \leq -C^2C_3n \right\} \geq 1-\epsilon.
\ee

Combining \eqref{eq:lt}, \eqref{eq:dldb}, \eqref{eq:dldbgg}, and \eqref{eq:hess}, 
we have
\bse
&& l\{\bb_0+n^{-1/2}\v,\wh\bg(\bb_0+n^{-1/2}\v)\} - l\{\bb_0,\wh\bg(\bb_0)\} \\
&\leq& Cn^{-1/2}(C_1n^{1/2} + C_2n^{1/2})
+ \frac{1}{2}n^{-1}(-C^2C_3n) \\
&=& C\left\{(C_1+C_2) - \frac{C_3}{2}C\right\} \\
&<& 0
\ese
with probability at least $1-3\epsilon$ when $C>2(C_1+C_2)/C_3$, which
proves \eqref{eq:lmax}.

We next show that the
Hessian of $l\{\bb,\wh\bg(\bb)\}$ is negative definite for all $\bb$,
hence the local maximizer inside the disk with radius $Cn^{-1/2}$ is
in fact the maximizer $\wh\bb$, i.e. 
$\|\wh\bb-\bb_0\|_2=O_p(n^{-1/2})$. We have
\bse
&& \begin{bmatrix}
\partial^2 l\{\bb,\wh\bg(\bb)\}/\partial\bb\partial\bb\trans
& \partial^2 l\{\bb,\wh\bg(\bb)\}/\partial\bb\partial\bg\trans \\
\partial^2 l\{\bb,\wh\bg(\bb)\}/\partial\bg\partial\bb\trans
& \partial^2 l\{\bb,\wh\bg(\bb)\}/\partial\bg\partial\bg\trans
\end{bmatrix} \\
&=& -\sumi \begin{bmatrix}
\x_i\x_i\trans\var^*\{Y|\x_i,\bb,\wh\bg(\bb)\}
& \x_i\cov^*\{Y,\B(Y)|\x_i,\bb,\wh\bg(\bb)\} \\
[\x_i\cov^*\{Y,\B(Y)|\x_i,\bb,\wh\bg(\bb)\}]\trans
& \cov^*\{\B(Y),\B(Y)|\x_i,\bb,\wh\bg(\bb)\}
\end{bmatrix} \\
&=& -\sumi
\cov^*\left\{ \begin{bmatrix}\x_iY \\\B(Y)\end{bmatrix},
\begin{bmatrix}\x_iY \\ \B(Y)\end{bmatrix} \Bigg|\x_i,\bb,\wh\bg(\bb)\right\}
\ese
is negative definite at all $\bb$.
Then since the Hessian is the Schur complement in the above matrix,
it is also negative definite for all $\bb$.
Hence we conclude that $\|\wh\bb-\bb_0\|_2=O_p(n^{-1/2})$.

Now, we analyze the asymptotic properties of $\wh\bb$.
Since $\wh\bb$ maximizes $l\{\bb,\wh\bg(\bb)\}$, we have
\bse
\0
&=& \frac{d l\{\wh\bb,\wh\bg(\wh\bb)\}}{d\bb} \\
&=& \frac{d l\{\bb_0,\wh\bg(\bb_0)\}}{d\bb}
+ \frac{d^2 l\{\bb^*,\wh\bg(\bb^*)\}}{d\bb d\bb\trans}(\wh\bb-\bb_0),
\ese
by the Taylor expansion with $\bb^*\equiv\alpha_1\wh\bb+(1-\alpha_1)\bb_0$ for some $\alpha_1\in(0,1)$.
This implies
\bse
\wh\bb-\bb_0
= \left[-n^{-1}\frac{d^2 l\{\bb^*,\wh\bg(\bb^*)\}}{d\bb d\bb\trans}\right]^{-1}
n^{-1}\frac{d l\{\bb_0,\wh\bg(\bb_0)\}}{d\bb}.
\ese
Recall that for $\bg^*=\alpha_2\wh\bg(\bb_0)+(1-\alpha_2)\bg_0$ for some $\alpha_2\in(0,1)$,
\bse
\frac{dl\{\bb_0,\wh\bg(\bb_0)\}}{d\bb}
=\frac{\partial l(\bb_0,\bg_0)}{\partial\bb}
+ \frac{\partial^2 l(\bb_0,\bg^*)}{\partial\bb\partial\bg\trans}\{\wh\bg(\bb_0)-\bg_0\}.
\ese
Then by \eqref{eq:ldb}, \eqref{eq:dldbgorder}, \eqref{eq:dldbg} and
Proposition \ref{pro:bg}, we can write the gradient as
\bse
\frac{d l\{\bb_0,\wh\bg(\bb_0)\}}{d\bb} 
&=& \sumi\x_i\{y_i-E(Y|\x_i)\} \n\\
&&- \bSig_{12}\bSig_{22}^{-1}\sumi[\B(y_i)-E\{\B(Y)|\x_i\}] + \r_2',
\ese
where $\|\r_2'\|_2=o_p(n^{1/2})$.
Hence by \eqref{eq:dldbbA},
\bse
\wh\bb-\bb_0
&=& \bSig^{*-1}n^{-1}\sumi\x_i\{y_i-E(Y|\x_i)\} \\
&&- \bSig^{*-1}\bSig_{12}\bSig_{22}^{-1}n^{-1}\sumi[\B(y_i)-E\{\B(Y)|\x_i\}] + \r_2,
\ese
where $\|\r_2\|_2=o_p(n^{-1/2})$.
Thus by the central limit theorem and the Slutsky's theorem, we get
\bse
{\bSig^*}^{1/2}\sqrt{n}(\wh\bb-\bb_0)\to N(\0, \I)
\ese
in distribution when $n\to\infty$.
\qed

\subsection{Lemmas}

We introduce another lemmas for asymptotic properties of the profile estimators.
\begin{Lem} \label{lem:whbg}
Under Conditions \ref{con:bdd}-\ref{con:bSig}, $\|\wh\bg(\wh\bb)-\bg_0\|_2=O_p\{(nh)^{-1/2}\}$ and
\bse
&& \wh\bg(\wh\bb)-\bg_0 \n\\
&=& - \bSig_{22}^{-1}\bSig_{21}\bSig^{*-1}n^{-1}\sumi\x_i\{y_i-E(Y|\x_i)\} \n\\
&&+ (\bSig_{22}^{-1}+\bSig_{22}^{-1}\bSig_{21}\bSig^{*-1}\bSig_{12}\bSig_{22}^{-1})
n^{-1}\sumi[\B(y_i)-E\{\B(Y)|\x_i\}] + \r_3,
\ese
where $\|\r_3\|_2=o_p\{(nh)^{-1/2}\}$.
\end{Lem}
\begin{proof}
Based on \eqref{eq:dgdb} and the Taylor expansion of $\wh\bg(\wh\bb)$,
letting $\bb^{*}$ be on the line connecting $\bb_0$ and $\wh\bb$,
we get
\be \label{eq:whbgwhbb}
\wh\bg(\wh\bb)-\bg_0
&=& \wh\bg(\bb_0) -\bg_0 
-\left[\frac{\partial^2 l\{\bb^*,\wh\bg(\bb^*)\}}{\partial\bg\partial\bg\trans}\right]^{-1}
\frac{\partial^2 l\{\bb^*,\wh\bg(\bb^*)\}}{\partial\bg\partial\bb\trans}(\wh\bb-\bb_0).
\ee
It is easy to verify that $\|\wh\bg(\wh\bb)-\bg_0\|_2=O_p\{(nh)^{-1/2}\}$
because $\|\wh\bg(\bb_0) -\bg_0\|_2=O_p\{(nh)^{-1/2}\}$ by Proposition \ref{pro:bg},
$\|\wh\bb-\bb_0\|_2=O_p(n^{-1/2})$ by Proposition \ref{pro:bb} and \eqref{eq:dgdborder}.
Now recall that $\|\bSig_{22}^{-1}\|_2\asymp h^{-1}$ by \eqref{eq:v} and Condition \ref{con:bb}.
In addition, for any $\u\in\mR^m$ such that $\|\u\|_2=1$, we get
\bse
\|\bSig_{12}\u\|_2
\leq \|E(\X)\|_2\sup_{\x\in\mx}|\cov\{Y,\B(Y)\trans\u|\x\}|
= O(h^{1/2})
\ese
by Condition \ref{con:bdd} and Lemma \ref{lem:becov}, which implies $\|\bSig_{21}\|_2= O(h^{1/2})$.
Then by \eqref{eq:dldbg2} and \eqref{eq:dldgg}, we get
\bse
\R
\equiv -\left[\frac{\partial^2 l\{\bb^*,\wh\bg(\bb^*)\}}{\partial\bg\partial\bg\trans}\right]^{-1}
\frac{\partial^2 l\{\bb^*,\wh\bg(\bb^*)\}}{\partial\bg\partial\bb\trans}
+ \bSig_{22}^{-1}\bSig_{21}
\ese
satisfies $\|\R\|_2= o_p(h^{-1/2})$. Then by Proposition \ref{pro:bb}, we have
\bse
&& -\left[\frac{\partial^2 l\{\bb^*,\wh\bg(\bb^*)\}}{\partial\bg\partial\bg\trans}\right]^{-1}
\frac{\partial^2 l\{\bb^*,\wh\bg(\bb^*)\}}{\partial\bg\partial\bb\trans}(\wh\bb-\bb_0) \\
&=& -\bSig_{22}^{-1}\bSig_{21}\bSig^{*-1}n^{-1}\sumi\x_i\{y_i-E(Y|\x_i)\} \\
&&+ \bSig_{22}^{-1}\bSig_{21}\bSig^{*-1}\bSig_{12}\bSig_{22}^{-1}n^{-1}\sumi[\B(y_i)-E\{\B(Y)|\x_i\}] + \r_3',
\ese
where
\bse
\|\r_3'\|_2
&=& \left\|\R(\wh\bb-\bb_0-\r_2) -\left[\frac{\partial^2 l\{\bb^*,\wh\bg(\bb^*)\}}{\partial\bg\partial\bg\trans}\right]^{-1}
\frac{\partial^2 l\{\bb^*,\wh\bg(\bb^*)\}}{\partial\bg\partial\bb\trans}\r_2\right\|_2 \\
&\leq& \|\R\|_2(\|\wh\bb-\bb_0\|_2+\|\r_2\|_2) +
\left\|\left[\frac{\partial^2 l\{\bb^*,\wh\bg(\bb^*)\}}{\partial\bg\partial\bg\trans}\right]^{-1}
\frac{\partial^2 l\{\bb^*,\wh\bg(\bb^*)\}}{\partial\bg\partial\bb\trans}\right\|_2\|\r_2\|_2 \\
&\leq& o_p(h^{-1/2})O_p(n^{-1/2}) + O_p(h^{-1/2})o_p(n^{-1/2}) \\
&=& o_p\{(nh)^{-1/2}\}
\ese
by Proposition \ref{pro:bb} and \eqref{eq:dgdborder}.
Combining this with Proposition \ref{pro:bg} and \eqref{eq:whbgwhbb}, we get
\bse
&& \wh\bg(\wh\bb)-\bg_0 \n\\
&=& - \bSig_{22}^{-1}\bSig_{21}\bSig^{*-1}n^{-1}\sumi\x_i\{y_i-E(Y|\x_i)\} \n\\
&&+ (\bSig_{22}^{-1}+\bSig_{22}^{-1}\bSig_{21}\bSig^{*-1}\bSig_{12}\bSig_{22}^{-1})
n^{-1}\sumi[\B(y_i)-E\{\B(Y)|\x_i\}] + \r_3,
\ese
where $\|\r_3\|_2=o_p\{(nh)^{-1/2}\}$.
\end{proof}

\begin{Lem} \label{lem:whbbwhbg}
Under Conditions \ref{con:bdd}-\ref{con:bSig},
\bse
\begin{bmatrix}
\wh\bb-\bb_0 \\ \wh\bg(\wh\bb)-\bg_0
\end{bmatrix}
=
\bSig^{-1}
n^{-1}\sumi
\begin{bmatrix}
\x_i\{y_i-E(Y|\x_i)\} \\ \B(y_i)-E\{\B(Y)|\x_i\}
\end{bmatrix}
+
\begin{bmatrix}
\r_2 \\ \r_3
\end{bmatrix},
\ese
where $\|\r_2\|_2=o_p(n^{-1/2})$ and $\|\r_3\|_2=o_p\{(nh)^{-1/2}\}$.
\end{Lem}
\begin{proof}
Note that Proposition \ref{pro:bb} and Lemma \ref{lem:whbg} lead to
\bse
\begin{bmatrix}
\wh\bb-\bb_0 \\ \wh\bg(\wh\bb)-\bg_0
\end{bmatrix}
=
\begin{bmatrix}
\C_{11} & \C_{12} \\ \C_{21} & \C_{22}
\end{bmatrix}
n^{-1}\sumi
\begin{bmatrix}
\x_i\{y_i-E(Y|\x_i)\} \\ \B(y_i)-E\{\B(Y)|\x_i\}
\end{bmatrix}
+
\begin{bmatrix}
\r_2 \\ \r_3
\end{bmatrix},
\ese
where $\|\r_2\|_2=o_p(n^{-1/2})$, $\|\r_3\|_2=o_p\{(nh)^{-1/2}\}$, and
\bse
\C_{11} &\equiv& \bSig^{*-1}, \\
\C_{12} &\equiv& -\bSig^{*-1}\bSig_{12}\bSig_{22}^{-1} = \C_{21}\trans, \\
\C_{22} &\equiv& \bSig_{22}^{-1}+\bSig_{22}^{-1}\bSig_{21}\bSig^{*-1}\bSig_{12}\bSig_{22}^{-1}.
\ese
In addition, it is easy to see that
\bse
\begin{bmatrix}
\C_{11} & \C_{12} \\ \C_{21} & \C_{22}
\end{bmatrix}
=
\bSig^{-1}.
\ese
\end{proof}

\subsection{Proof of Theorem \ref{th:xi}} \label{proof:xi}

We can write $\wh\bxi-\bxi_0$ as
\be \label{eq:bxi}
\wh\bxi-\bxi_0
&=& (\wh\bb-\bb_0)\,n^{-1}\sumi \var^*\{Y|\x_i,\wh\bb,\wh\bg(\wh\bb)\} \n\\
&&+ \bb_0\,n^{-1}\sumi\left[\var^*\{Y|\x_i,\wh\bb,\wh\bg(\wh\bb)\}-\var(Y|\x_i)\right] \n\\
&&+ \bb_0\left[n^{-1}\sumi \var( Y|\x_i )-E\{\var( Y|\X )\}\right].
\ee
Noting that $\|\wh\bb-\bb_0\|_2=O_p(n^{-1/2})$ by Proposition \ref{pro:bb}
and $\|\wh\bg(\wh\bb)-\bg_0\|_2=O_p\{(nh)^{-1/2}\}$ by Lemma \ref{lem:whbg},
it is easy to see that
\be \label{eq:vconv}
n^{-1}\sumi \var^*\{Y|\x_i,\wh\bb,\wh\bg(\wh\bb)\}
&=& n^{-1}\sumi [\var^*\{Y|\x_i,\wh\bb,\wh\bg(\wh\bb)\}-\var(Y|\x_i)] \n\\
&&+ n^{-1}\sumi [\var(Y|\x_i)-E\{\var(Y|\X)\}] + E\{\var(Y|\X)\} \n\\
&=& E\{\var(Y|\X)\} + o_p(1)
\ee
by Lemma \ref{lem:bbbg}. Next, we have
\be \label{eq:v2}
&& n^{-1}\sumi\left[\var^*\{Y|\x_i,\wh\bb,\wh\bg(\wh\bb)\} - \var(Y|\x_i)\right] \\
&=& n^{-1}\sumi\left[\var^*\{Y|\x_i,\wh\bb,\wh\bg(\wh\bb)\} - \var^*(Y|\x_i,\bb_0,\bg_0)
+ \var^*(Y|\x_i,\bb_0,\bg_0) - \var(Y|\x_i)\right] \n\\
&=& n^{-1}\sumi \left[ \frac{\partial \var^*(Y|\x_i,\bb^*,\bg^*)}{\partial\bb\trans}(\wh\bb - \bb_0)
+ \frac{\partial \var^*(Y|\x_i,\bb^*,\bg^*)}{\partial\bg\trans}\{\wh\bg(\wh\bb) - \bg_0\} \right]
+ O(h^q), \n
\ee
where $(\bb^{*\rm T},\bg^{*\rm T})\trans$ is a point on the line
connecting $\{\wh\bb\trans,\wh\bg(\wh\bb)\trans\}\trans$ and $(\bb_0\trans,\bg_0\trans)\trans$.
The last equality holds by Lemma \ref{lem:deboor}.
Then by Lemma \ref{lem:bbbg}, it is easy to check that
\be \label{eq:dvdb}
&& \left\|n^{-1}\sumi\frac{\partial \var^*(Y|\x_i,\bb^*,\bg^*)}{\partial\bb}
- E[\{Y-E(Y|\X)\}^3\X]\right\|_2 \n\\
&=& \left\|n^{-1}\sumi E^*[\{Y-E^*(Y|\x_i,\bb^*,\bg^*)\}^3|\x_i,\bb^*,\bg^*]\x_i
- E[\{Y-E(Y|\X)\}^3\X]\right\|_2 \n\\
&=& o_p(1).
\ee 
Furthermore, we have
\be \label{eq:dvdg}
\frac{\partial \var^*(Y|\x,\bb,\bg)}{\partial\bg}
=E^*(\{Y-E^*(Y|\x,\bb,\bg)\}^2[\B(Y)-E^*\{\B(Y)|\x,\bb,\bg\}]|\x,\bb,\bg).
\ee 
Note that for $\u\in\mR^m$ such that $\|\u\|_2=1$, 
\bse
&& E^*[\{Y-E^*(Y|\x,\bb,\bg)\}^2\B(Y)|\x,\bb,\bg]\trans\u \\
&\leq& \|\B(\cdot)\trans\u\|_1\sup_{y\in[0,1],\x\in\mx}\{y-E^*(Y|\x,\bb,\bg)\}^2\fyx^*(y,\x,\bb,\bg) \\
&=& O(h^{1/2})
\ese
by Condition \ref{con:bdd} and Lemma \ref{lem:bnorm}. Similar argument will show
\bse
E^*[\{Y-E^*(Y|\x,\bb,\bg)\}^2E^*\{\B(Y)|\x,\bb,\bg\}|\x,\bb,\bg]\trans\u = O(h^{1/2}),
\ese
and
\be \label{eq:evb}
E(\{Y-E(Y|\x)\}^2[\B(Y)-E\{\B(Y)|\x\}]|\x)\trans\u = O(h^{1/2}).
\ee
Then \eqref{eq:dvdg}, Condition \ref{con:bdd}, and Lemma \ref{lem:bbbg} lead to
\bse
&& \sup_{\x\in\mx}\left|\frac{\partial \var^*(Y|\x,\bb^*,\bg^*)}{\partial\bg\trans}\u
- E(\{Y-E(Y|\x)\}^2[\B(Y)-E\{\B(Y)|\x\}]|\x)\trans\u\right| \\
&=& o_p(h^{1/2}).
\ese 
Thus, we have
\bse
\left\|n^{-1}\sumi\frac{\partial \var^*(Y|\x_i,\bb^*,\bg^*)}{\partial\bg}
- E(\{Y-E(Y|\X)\}^2[\B(Y)-E\{\B(Y)|\X\}])\right\|_2
= o_p(h^{1/2}).
\ese
Using the fact that $\|\wh\bb-\bb_0\|_2=O_p(n^{-1/2})$ and $\|\wh\bg(\wh\bb)-\bg_0\|_2=O_p\{(nh)^{-1/2}\}$,
we combine the above with \eqref{eq:v2} and \eqref{eq:dvdb} so get
\bse
&& n^{-1}\sumi\left[\var^*\{Y|\x_i,\wh\bb,\wh\bg(\wh\bb)\} - \var(Y|\x_i)\right] \\
&=& E[\{Y-E(Y|\X)\}^3\X]\trans(\wh\bb-\bb_0) \\
&&+ E(\{Y-E(Y|\X)\}^2[\B(Y)-E\{\B(Y)|\X\}])\trans\{\wh\bg(\wh\bb) - \bg_0\} + o_p(n^{-1/2}).
\ese
Then, we can rewrite \eqref{eq:bxi} using \eqref{eq:vconv} as
\be \label{eq:whbxi}
\wh\bxi-\bxi_0
&=& (E\{\var(Y|\X)\}\I + \bb_0E[\{Y-E(Y|\X)\}^3\X\trans])(\wh\bb-\bb_0) \n\\
&&+ \bb_0 E(\{Y-E(Y|\X)\}^2[\B(Y)-E\{\B(Y)|\X\}])\trans\{\wh\bg(\wh\bb) - \bg_0\} \n\\
&&+ \bb_0n^{-1}\sumi\left[\var( Y|\x_i )-E\{\var( Y|\X )\}\right] + \r,
\ee
where $\|\r\|_2=o_p(n^{-1/2})$.
Hence we get $\|\wh\bxi-\bxi_0\|_2=O_p(n^{-1/2})$ by \eqref{eq:evb}.

Now we derive the asymptotic distribution of $\wh\bxi$.
Recall from Theorem \ref{th:xi} that $\A=[\A_1, \A_2]$ and
\bse
\A_1 &=& E\{\var(Y|\X)\}\I + \bb_0E[\{Y-E(Y|\X)\}^3\X\trans], \\
\A_2 &=& \bb_0 E(\{Y-E(Y|\X)\}^2[\B(Y)-E\{\B(Y)|\X\}])\trans.
\ese
Then by Lemma \ref{lem:whbbwhbg}, \eqref{eq:whbxi} equals to
\bse
\wh\bxi-\bxi_0
&=& \A \bSig^{-1}n^{-1}\sumi
\begin{bmatrix}
\x_i\{y_i-E(Y|\x_i)\} \\ \B(y_i)-E\{\B(Y)|\x_i\}
\end{bmatrix} \n\\
&&+ \bb_0n^{-1}\sumi\left[\var( Y|\x_i )-E\{\var( Y|\X )\}\right] \n\\
&&+ \A_1\r_2 + \A_2\r_3 + \r.
\ese
We obviously have
\bse
\cov[\X\{Y-E(Y|\X)\},\var(Y|\X)] &=& \0_p, \\
\cov[\B(Y)-E\{\B(Y)|\X\},\var(Y|\X)] &=& \0_m.
\ese
In addition, recall that $\|\r_2\|_2=o_p(n^{-1/2})$ and $\|\r_3\|_2=o_p\{(nh)^{-1/2}\}$
by Lemma \ref{lem:whbbwhbg},
and $\|\A_2\|_2=O(h^{1/2})$ by \eqref{eq:evb},
then the remainders satisfy
\bse
\|\A_1\r_2 + \A_2\r_3 + \r\|_2
&\leq& \|\A_1\|_2\|\r_2\|_2 + \|\A_2\|_2\|\r_3\|_2 + \|\r\|_2 \n\\
&=& O(1)o_p(n^{-1/2}) + O(h^{1/2})o_p\{(nh)^{-1/2}\} + o_p(n^{-1/2}) \n\\
&=& o_p(n^{-1/2}).
\ese
Hence, letting $\bSig_\bxi\equiv\A\bSig^{-1}\A\trans+\bb_0\bb_0\trans\var\{\var(Y|\X)\}$,
$\bSig_\bxi^{-1/2}\sqrt{n}(\wh\bxi-\bxi_0)$ converges to the normal distribution
with mean $\0$ and variance $\I$.
\qed

\subsection{Proof of Theorem \ref{th:bta}}

For notational brevity, we denote $\nu\equiv\x\trans\bb$, $q^*(\nu,\bg)$ be such that
\bse
\int_0^{q^*(\nu,\bg)}\exp\{t\nu+\B(t)\trans\bg\}dt=\tau\int_0^1\exp\{t\nu+\B(t)\trans\bg\}dt,
\ese
and
\bse
{q^*}'(\nu,\bg)
&\equiv&\frac{\partial q^*(\nu,\bg)}{\partial\nu}\\
&=&\frac{\tau\int_0^1t\exp\{t\nu+\B(t)\trans\bg\}dt-\int_0^yt\exp\{t\nu+\B(t)\trans\bg\}dt}
{\exp\{y\nu+\B(y)\trans\bg\}}\bigg|_{y=q^*(\nu,\bg)}\\
&=&\frac{E^*([\tau-I\{Y<q^*(\nu,\bg)\}]Y|\x,\bb,\bg)}{\fyx^*\{q^*(\nu,\bg),\x,\bb,\bg\}}.
\ese
Let $\bb^*$ and $\bg^*$ be such that $\|\bb^*-\bb_0\|_2=o_p(1)$ and $\|\bg^*-\bg_0\|_2=o_p(1)$ respectively. 
First note that
\be \label{eq:qconv}
\sup_{\x\in\mx}|q^*(\x\trans\bb^*,\bg^*)-q(\x\trans\bb_0)|=o_p(1).
\ee
This is because
\bse
0
&=&\int_0^{q^*(\x\trans\bb^*,\bg^*)}\fyx^*(y,\x,\bb^*,\bg^*)dy
-\int_0^{q(\x\trans\bb_0)}\fyx(y|\x)dy\\
&=&\int_0^{q^*(\x\trans\bb^*,\bg^*)}\{\fyx^*(y,\x,\bb^*,\bg^*)-\fyx(y|\x)\}dy
+\int_{q(\x\trans\bb_0)}^{q^*(\x\trans\bb^*,\bg^*)}\fyx(y|\x)dy,
\ese
which implies, by Condition \ref{con:bdd} and Lemma \ref{lem:bbbg}, that

\bse
c_f|q^*(\x\trans\bb^*,\bg^*)-q(\x\trans\bb_0)|
&\leq&\left|\int_0^{q^*(\x\trans\bb^*,\bg^*)}\{\fyx^*(y,\x,\bb^*,\bg^*)-\fyx(y|\x)\}dy\right|\\
&=&o_p(1)
\ese
uniformly in $\x$, where $c_f=\inf_{y\in[0,1],\x\in\mx}\fyx(y|\x)$.
Condition \ref{con:bdd} and Lemma \ref{lem:bbbg} further lead to
\be \label{eq:q'conv}
\sup_{\x\in\mx}|{q^*}'(\x\trans\bb^*,\bg^*)-q'(\x\trans\bb_0)|=o_p(1),
\ee
since uniformly in $\x, y$,
\be \label{eq:fq}
&&\left|\fyx^*\{q^*(\x\trans\bb^*,\bg^*),\x,\bb^*,\bg^*\}
-\fyx\{q(\x\trans\bb_0)|\x\}\right|\n\\
&\leq&\left|\fyx^*\{q^*(\x\trans\bb^*,\bg^*),\x,\bb^*,\bg^*\}
-\fyx\{q^*(\x\trans\bb^*,\bg^*)|\x\}\right|\n\\
&&+\left|\fyx\{q^*(\x\trans\bb^*,\bg^*)|\x\}-\fyx\{q(\x\trans\bb_0)|\x\}\right|\n\\
&=&o_p(1),
\ee
$|E^*(Y|\x,\bb^*,\bg^*)-E(Y|\x)|=o_p(1)$,
and for an arbitrary $g(\cdot)$ such that $\|g(\cdot)\|_2<\infty$,
\be \label{eq:egyconv}
&&\left|E^*\left[I\{Y<q^*(\x\trans\bb^*,\bg^*)\}g(Y)|\x,\bb^*,\bg^*\right]
-E\left[I\{Y<q(\bb_0\trans\x)\}g(Y)|\x\right]\right|\n\\
&\leq&\left|\int_0^{q^*(\x\trans\bb^*,\bg^*)}g(y)\{\fyx^*(y,\x,\bb^*,\bg^*)-\fyx(y|\x)\}dy\right|
+\left|\int_{q(\x\trans\bb_0)}^{q^*(\x\trans\bb^*,\bg^*)}g(y)\fyx(y|\x)dy\right|\n\\
&\leq&o_p(1)\|g(\cdot)\|_2+\left|\int_{q(\x\trans\bb_0)}^{q^*(\x\trans\bb^*,\bg^*)}\{g(y)\}^2dy
\int_{q(\x\trans\bb_0)}^{q^*(\x\trans\bb^*,\bg^*)}\{\fyx(y|\x)\}^2dy\right|^{1/2}\n\\
&\leq&o_p(1)\|g(\cdot)\|_2+C_f|q^*(\x\trans\bb^*,\bg^*)-q(\x\trans\bb_0)|^{1/2}\|g(\cdot)\|_2\n\\
&=&o_p(1)\|g(\cdot)\|_2.
\ee 
Similarly, we also have
\be \label{eq:q'deboor}
\sup_{\x\in\mx}|{q^*}'(\x\trans\bb_0,\bg_0)-q'(\x\trans\bb_0)|=O(h^q)
\ee
by Condition \ref{con:bdd} and Lemma \ref{lem:deboor}.

In addition, we can show that
\be
&&\frac{\partial {q^*}'(\nu,\bg)}{\partial\bb}\n\\
&=&\frac{\tau\int_0^1\x t^2\exp\{t\nu+\B(t)\trans\bg\}dt}{\exp[q^*(\nu,\bg)\nu+\B\{q^*(\nu,\bg)\}\trans\bg]}\n\\
&&-\frac{\int_0^{q^*(\nu,\bg)}\x t^2\exp\{t\nu+\B(t)\trans\bg\}dt
+q^*(\nu,\bg)\exp[q^*(\nu,\bg)\nu+\B\{q^*(\nu,\bg)\}\trans\bg]\x{q^*}'(\nu,\bg)}
{\exp[q^*(\nu,\bg)\nu+\B\{q^*(\nu,\bg)\}\trans\bg]}\n\\
&&-\frac{\tau\int_0^1t\exp\{t\nu+\B(t)\trans\bg\}dt-\int_0^{q^*(\nu,\bg)}t\exp\{t\nu+\B(t)\trans\bg\}dt}
{\exp[q^*(\nu,\bg)\nu+\B\{q^*(\nu,\bg)\}\trans\bg]}\n\\
&&\times\left[\x{q^*}'(\nu,\bg)\nu+q^*(\nu,\bg)\x+\B'\{q^*(\nu,\bg)\}\trans\bg\x{q^*}'(\nu,\bg)\right]\n\\
&=&\x\Bigg\{\frac{ E^*([\tau-I\{Y\leq q^*(\nu,\bg)\}]Y^2|\x,\bb,\bg)}
{\fyx^*\{q^*(\nu,\bg),\x,\bb,\bg\}}\n\\
&&-2{q^*}'(\nu,\bg)q^*(\nu,\bg)-\{{q^*}'(\nu,\bg)\}^2[\nu+\B'\{q^*(\nu,\bg)\}\trans\bg]\Bigg\},\label{eq:dq'db}\\
&&\frac{\partial q^*(\nu,\bg)}{\partial\bg}\n\\
&=&\frac{\tau\int_0^1\B(t)\exp\{t\nu+\B(t)\trans\bg\}dt-\int_0^y\B(t)\exp\{t\nu+\B(t)\trans\bg\}dt}
{\exp\{y\nu+\B(y)\trans\bg\}}\bigg|_{y=q^*(\nu,\bg)}\n\\
&=&\frac{E^*([\tau-I\{Y\leq q^*(\nu,\bg)\}]\B(Y)|\x,\bb,\bg)}{\fyx^*\{q^*(\nu,\bg),\x,\bb,\bg\}},\label{eq:dqdg}
\ee 
and
\be
&&\frac{\partial {q^*}'(\nu,\bg)}{\partial\bg}\n\\
&=&\frac{\tau\int_0^1t\B(t)\exp\{t\nu+\B(t)\trans\bg\}dt}{\exp\{q^*(\nu,\bg)\nu+\B(y)\trans\bg\}}\n\\
&&-\frac{\int_0^{q^*(\nu,\bg)}t\B(t)\exp\{t\nu+\B(t)\trans\bg\}dt
+q^*(\nu,\bg)\exp[q^*(\nu,\bg)\nu+\B\{q^*(\nu,\bg)\}\trans\bg]\frac{\partial q^*(\nu,\bg)}{\partial\bg}}
{\exp[q^*(\nu,\bg)\nu+\B\{q^*(\nu,\bg)\}\trans\bg]}\n\\
&&-\frac{\tau\int_0^1t\exp\{t\nu+\B(t)\trans\bg\}dt-\int_0^{q^*(\nu,\bg)}t\exp\{t\nu+\B(t)\trans\bg\}dt}
{\exp[q^*(\nu,\bg)\nu+\B\{q^*(\nu,\bg)\}\trans\bg]}\n\\
&&\times\left[\frac{\partial q^*(\nu,\bg)}{\partial\bg}\nu+\B'\{q^*(\nu,\bg)\}\trans\bg
\frac{\partial q^*(\nu,\bg)}{\partial\bg}+\B\{q^*(\nu,\bg)\}\right]\n\\
&=&\frac{E^*([\tau-I\{Y\leq q^*(\nu,\bg)\}]Y\B(Y)|\x,\bb,\bg)}{\fyx^*\{q^*(\nu,\bg),\x,\bb,\bg\}}\label{eq:dq'dg}\\
&&-{q^*}'(\nu,\bg)\B\{q^*(\nu,\bg)\}
-\frac{\partial q^*(\nu,\bg)}{\partial\bg}[q^*(\nu,\bg)+{q^*}'(\nu,\bg)\nu+{q^*}'(\nu,\bg)\B'\{q^*(\nu,\bg)\}\trans\bg].\n
\ee

Now using \eqref{eq:q'deboor}, we write $\wh{\bta_\tau}-\bta_{\tau0}$ as
\be \label{eq:bta}
\wh{\bta_\tau}-\bta_{\tau0}
&=& (\wh\bb-\bb_0)n^{-1}\sumi{q^*}'\{\x_i\trans\wh\bb,\wh\bg(\wh\bb)\}
+\bb_0n^{-1}\sumi\left[{q^*}'\{\x_i\trans\wh\bb,\wh\bg(\wh\bb)\}-{q^*}'(\x_i\trans\bb_0,\bg_0)\right]\n\\
&&+\bb_0n^{-1}\sumi\left\{{q^*}'(\x_i\trans\bb_0,\bg_0)-q'(\x_i\trans\bb_0)\right\}
+\bb_0\left[n^{-1}\sumi q'(\x_i\trans\bb_0)-E\{q'(\X\trans\bb_0)\}\right]\n\\
&=& \left[n^{-1}\sumi{q^*}'\{\x_i\trans\wh\bb,\wh\bg(\wh\bb)\}\I
+\bb_0n^{-1}\sumi\frac{\partial {q^*}'(\x_i\trans\bb^*,\bg^*)}{\partial\bb\trans}\right](\wh\bb-\bb_0)\n\\
&&+\bb_0n^{-1}\sumi\frac{\partial {q^*}'(\x_i\trans\bb^*,\bg^*)}{\partial\bg\trans}\{\wh\bg(\wh\bb)-\bg_0\}\n\\
&&+\bb_0\left[n^{-1}\sumi q'(\x_i\trans\bb_0)-E\{q'(\X\trans\bb_0)\}\right]+O(h^q),
\ee
where $({\bb^*}\trans,{\bg^*}\trans)\trans$ is a point on the line
connecting $\{\wh\bb\trans,\wh\bg(\wh\bb)\trans\}\trans$ and $(\bb_0\trans,\bg_0\trans)\trans$.
To treat the first term in \eqref{eq:bta}, we first obtain
\be\label{eq:bta11}
&&\left|n^{-1}\sumi{q^*}'\{\x_i\trans\wh\bb,\wh\bg(\wh\bb)\}-E\{q'(\X\trans\bb_0)\}\right|\n\\
&\leq&n^{-1}\sumi\left|{q^*}'\{\x_i\trans\wh\bb,\wh\bg(\wh\bb)\}-q'(\x_i\trans\bb_0)\right|
+\left|n^{-1}\sumi q'(\x_i\trans\bb_0)-E\{q'(\X\trans\bb_0)\}\right|\n\\
&=&o_p(1)
\ee
by \eqref{eq:q'conv}.
Now we will show
\be\label{eq:t}
&&\Bigg\|n^{-1}\sumi\frac{\partial {q^*}'(\x_i\trans\bb^*,\bg^*)}{\partial\bb}
-E\Bigg[\X\Bigg\{\frac{ E([\tau-I\{Y\leq q(\X\trans\bb_0)\}]Y^2|\X)}{\fyx\{q(\X\trans\bb_0)|\X\}}\n\\
&&-2q'(\X\trans\bb_0)q(\X\trans\bb_0)-\{q'(\X\trans\bb_0)\}^2[\X\trans\bb_0+c'\{q(\X\trans\bb_0)\}]\}]\|_2\n\\
&=&o_p(1).
\ee
We prove \eqref{eq:t} through proving 
\be\label{eq:dq'dbconv}
&&\Bigg\|\frac{\partial {q^*}'(\x\trans\bb^*,\bg^*)}{\partial\bb}
-\x\Bigg\{\frac{ E([\tau-I\{Y\leq q(\x\trans\bb_0)\}]Y^2|\x)}{\fyx\{q(\x\trans\bb_0)|\x\}}\n\\
&&-2q'(\x\trans\bb_0)q(\x\trans\bb_0)-\{q'(\x\trans\bb_0)\}^2[\x\trans\bb_0+c'\{q(\x\trans\bb_0)\}]\Bigg\}\Bigg\|_2\n\\
&=&o_p(1)
\ee
uniformly in $\x$ under Condition \ref{con:bdd},
where $\partial {q^*}'(\x\trans\bb^*,\bg^*)/\partial\bb$ is given in
\eqref{eq:dq'db}. To prove (\ref{eq:dq'dbconv}), 
first note that
\be\label{eq:dq'db1}
&&\left|\frac{E^*([\tau-I\{Y\leq q^*(\x\trans\bb^*,\bg^*)\}]Y^2|\x,\bb^*,\bg^*)}
{\fyx^*\{ q^*(\x\trans\bb^*,\bg^*),\x,\bb^*,\bg^*\}}
-\frac{ E([\tau-I\{Y\leq q(\x\trans\bb_0)\}]Y^2|\x)}{\fyx\{q(\x\trans\bb_0)|\x\}}\right|\n\\
&=&o_p(1),
\ee
because $\left|E^*(Y^2|\x,\bb^*,\bg^*)-E(Y^2|\x)\right|=o_p(1)$ by Lemma \ref{lem:bbbg},
\bse
|E^*[I\{Y\leq q^*(\x\trans\bb^*,\bg^*)\}Y^2|\x,\bb^*,\bg^*]-E[I\{Y\leq q(\x\trans\bb_0)\}Y^2|\x]|=o_p(1)
\ese
by \eqref{eq:egyconv},  and
\bse
\left|\frac{1}{\fyx^*\{ q^*(\x\trans\bb^*,\bg^*),\x,\bb^*,\bg^*\}}
-\frac{1}{\fyx\{ q(\x\trans\bb_0)|\x\}}\right|
=o_p(1)
\ese
by \eqref{eq:fq}. Furthermore, we have
\be \label{eq:dq'db2}
|q{^*}'(\x\trans\bb^*,\bg^*)q^*(\x\trans\bb^*,\bg^*)-q'(\x\trans\bb_0)q(\x\trans\bb_0)|=o_p(1)
\ee
by \eqref{eq:qconv} and \eqref{eq:q'conv}. We additionally get
\be\label{eq:b'}
&&|\B'\{q^*(\x\trans\bb^*,\bg^*)\}\trans\bg^*-c'\{q(\x\trans\bb_0)\}|\n\\
&\leq&\|\B'\{q^*(\x\trans\bb^*,\bg^*)\}\|_2\|\bg^*-\bg_0\|_2
+|\B'\{q^*(\x\trans\bb^*,\bg^*)\}\trans\bg_0-c'\{q^*(\x\trans\bb^*,\bg^*)\}|\n\\
&&+|c'\{q^*(\x\trans\bb^*,\bg^*)\}-c'\{q(\x\trans\bb_0)\}|\n\\
&=& o_p(1),
\ee
where, to bound the first term, we used
$\|\bg^*-\bg_0\|_2=O_p\{(nh)^{-1/2}\}$ by Lemma \ref{lem:whbg}, for
the second term, we used
$|\B'\{q^*(\x\trans\bb^*,\bg^*)\}\trans\bg_0-c'\{q^*(\x\trans\bb^*,\bg^*)\}|=O(h^{q/2})$
under Conditions \ref{con:deboor}, \ref{con:q} 
and \ref{con:order2}, and for the third term, we noted that $c'(\cdot)$ is continuous under Condition
\ref{con:bdd} and used \eqref{eq:qconv}.
The results in (\ref{eq:b'}) and \eqref{eq:q'conv}, together with the
fact $\|\bb^*-\bb_0\|_2=O_p(n^{-1/2})$
 directly lead to
\be\label{eq:dq'db3}
&&|{q^*}'(\x\trans\bb^*,\bg^*)\}^2[\x\trans\bb^*+\B'\{q^*(\x\trans\bb^*,\bg^*)\}\trans\bg^*]
-\{q'(\x\trans\bb_0)\}^2[\x\trans\bb_0+c'\{q(\x\trans\bb_0)\}]|\n\\
&=&o_p(1).
\ee
Combining the results in \eqref{eq:dq'db1}, \eqref{eq:dq'db2}, and
\eqref{eq:dq'db3}, and taking into
account the form in \eqref{eq:dq'db} lead to \eqref{eq:dq'dbconv}, and
subsequently \eqref{eq:t}.
We now combine \eqref{eq:bta11} and (\ref{eq:t}) to get
\be\label{eq:bta1}
\left\|n^{-1}\sumi{q^*}'\{\x_i\trans\wh\bb,\wh\bg(\wh\bb)\}\I
+\bb_0n^{-1}\sumi\frac{\partial {q^*}'(\x_i\trans\bb^*,\bg^*)}{\partial\bb\trans}-\C_1\right\|_2=o_p(1).
\ee

Next, we handle the second term in \eqref{eq:bta}. 
By \eqref{eq:dqdg} and $\u\in U\equiv\{\u\in\mR^m:\|\u\|_2=1\}$, we
have
\bse
&&\left|\left\{\frac{\partial q^*(\x\trans\bb^*,\bg^*)}{\partial\bg}
-\frac{E([\tau-I\{Y\leq q(\x\trans\bb_0)\}]\B(Y)|\x)}{\fyx\{q(\x\trans\bb_0)|\x\}}\right\}\trans\u\right|\\
&\leq&\Bigg|\frac{1}{\fyx^*\{ q^*(\x\trans\bb^*,\bg^*),\x,\bb^*,\bg^*\}}
\big\{\tau[E^*\{\B(Y)\trans\u|\x,\bb^*,\bg^*\}-E\{\B(Y)\trans\u|\x\}]\\
&&-(E^*[I\{Y\leq q^*(\x\trans\bb^*,\bg^*)\}\B(Y)\trans\u|\x,\bb^*,\bg^*]
-E[I\{Y\leq q(\x\trans\bb_0)\}\B(Y)\trans\u|\x])\big\}\Bigg|\\
&&+\Bigg|\left[\frac{1}{\fyx^*\{ q^*(\x\trans\bb^*,\bg^*),\x,\bb^*,\bg^*\}}
-\frac{1}{\fyx\{ q(\x\trans\bb_0)|\x\}}\right]\\
&&\times E([\tau-I\{Y\leq q(\x\trans\bb_0)\}]\B(Y)\trans\u|\x)\Bigg|\\
&=&o_p(h^{1/2}).
\ese
The last equality holds because
$\left|E^*\{\B(Y)\trans\u|\x,\bb^*,\bg^*\}-E\{\B(Y)\trans\u|\x\}\right|=o_p(h^{1/2})$
by Lemmas \ref{lem:bnorm} and \ref{lem:bbbg}, 
\bse
\left|E^*[I\{Y\leq q^*(\x\trans\bb^*,\bg^*)\}\B(Y)\trans\u|\x,\bb^*,\bg^*]
-E[I\{Y\leq q(\x\trans\bb_0)\}\B(Y)\trans\u|\x]\right|=o_p(h^{1/2})
\ese
by \eqref{eq:egyconv} and Lemma \ref{lem:bnorm}, 
\bse
\left|\frac{1}{\fyx^*\{ q^*(\x\trans\bb^*,\bg^*),\x,\bb^*,\bg^*\}}
-\frac{1}{\fyx\{ q(\x\trans\bb_0)|\x\}}\right|
=o_p(1)
\ese
by \eqref{eq:fq}, and
\be\label{eq:dq'dg1}
|E([\tau-I\{Y\leq q(\x\trans\bb_0)\}]\B(Y)\trans\u|\x)|\leq E\{|\B(Y)\trans\u|\mid\x\}=O(h^{1/2})
\ee
by Lemma \ref{lem:becov}.  Hence we get
\be\label{eq:q'1}
\left\|\frac{\partial q^*(\x\trans\bb^*,\bg^*)}{\partial\bg}
-\frac{E([\tau-I\{Y\leq q(\x\trans\bb_0)\}]\B(Y)|\x)}{\fyx\{q(\x\trans\bb_0)|\x\}}\right\|_2
=o_p(h^{1/2})
\ee
uniformly in $\x$ under Condition \ref{con:bdd}.
Similarly, we have 
\be\label{eq:dq'dg2}
|E([\tau-I\{Y\leq q(\x\trans\bb_0)\}]Y\B(Y)\trans\u|\x)|=O(h^{1/2})
\ee
and
\be\label{eq:q'2}
&&\Bigg\|\frac{E^*([\tau-I\{Y\leq q^*(\x\trans\bb^*,\bg^*)\}]Y\B(Y)|\x,\bb^*,\bg^*)}
{\fyx^*\{q^*(\x\trans\bb^*,\bg^*),\x,\bb^*,\bg^*\}}\n\\
&&-\frac{E([\tau-I\{Y\leq q(\x\trans\bb_0)\}]Y\B(Y)|\x)}{\fyx\{q(\x\trans\bb_0)|\x\}}\Bigg\|_2\n\\
&=&o_p(h^{1/2})
\ee
uniformly in $\x$. In addition, \eqref{eq:qconv}, \eqref{eq:q'conv}
and Condition \ref{con:order2} imply 
\be\label{eq:q'3}
&&\left\|n^{-1}\sumi{q^*}'(\x_i\trans\bb^*,\bg^*)\B\{q^*(\x_i\trans\bb^*,\bg^*)\}
-E\left[q'(\X\trans\bb_0)\B\{q(\X\trans\bb_0)\}\right]\right\|_2\n\\
&\leq&\sup_{\u\in U}\left|n^{-1}\sumi{q^*}'(\x_i\trans\bb^*,\bg^*)\B\{q^*(\x_i\trans\bb^*,\bg^*)\}\trans\u
-E\left[{q^*}'(\X\trans\bb^*,\bg^*)\B\{q^*(\X\trans\bb^*,\bg^*)\}\trans\u\right]\right|\n\\
&&+\sup_{\u\in U}\left|E\left[{q^*}'(\X\trans\bb^*,\bg^*)\B\{q^*(\X\trans\bb^*,\bg^*)\}\trans\u\right]
-E\left[q'(\X\trans\bb_0)\B\{q(\X\trans\bb_0)\}\trans\u\right]\right|\n\\
&=&o_p(h^{1/2}).
\ee
The last equality in \eqref{eq:q'3} is because
\be\label{eq:dq'dg3}
&&\left|E\left[q'(\X\trans\bb_0)\B\{q(\X\trans\bb_0)\}\trans\u\right]\right|\n\\
&=&\left|\int_{\{\nu:\nu=\x\trans\bb_0,\x\in\mx\}}
\B\{q(\nu)\}\trans\u \left\{\int_{\{\x:\x\trans\bb_0=\nu\}}\fx(\x)d\x\right\} q'(\nu)d\nu\n\right|\\
&=&\left|\int_0^1\B(t)\trans\u\left\{\int_{\{\x:q(\x\trans\bb_0)=t\}}\fx(\x)d\x\right\}dt\right|\n\\
&\leq&\|\B(\cdot)\trans\u\|_1\sup_{\x\in\mx}\fx(\x)\n\\
&=&O(h^{1/2}).
\ee
 by Lemma \ref{lem:bnorm},  
and similarly
$|E\left[{q^*}'(\X\trans\bb^*,\bg^*)\B\{q^*(\X\trans\bb^*,\bg^*)\}\trans\u\right]|=O(h^{1/2})$,
which leads to the last step in \eqref{eq:q'3} under Condition
\ref{con:bdd}. 
Furthermore, we have
\be
&&|[q^*(\x\trans\bb^*,\bg^*)+{q^*}'(\x\trans\bb^*,\bg^*)\x\trans\bb^*
+{q^*}'(\x\trans\bb^*,\bg^*)\B'\{q^*(\x\trans\bb^*,\bg^*)\}\trans\bg^*]\n\\
&&-[q(\x\trans\bb_0)+q'(\x\trans\bb_0)\x\trans\bb_0+q'(\x\trans\bb_0)c'\{q(\x\trans\bb_0)\}]|\n\\
&=&o_p(1)\label{eq:added}
\ee
by \eqref{eq:qconv}, \eqref{eq:q'conv}, \eqref{eq:b'}, and
$\|\bb^*-\bb_0\|_2=O_p(n^{-1/2})$.
Now, combining 
\eqref{eq:dq'dg1}, \eqref{eq:dq'dg2} and \eqref{eq:dq'dg3}, we obtain 
\be\label{eq:c2}
\|\C_2\|_2=O(h^{1/2}).
\ee 
Combining \eqref{eq:dq'dg}, \eqref{eq:q'1}, \eqref{eq:q'2},
\eqref{eq:q'3} and \eqref{eq:added}, we further get
\be\label{eq:bta2}
\left\|\bb_0n^{-1}\sumi\frac{\partial {q^*}'(\x_i\trans\bb^*,\bg^*)}{\partial\bg\trans}-\C_2\right\|_2
=o_p(h^{1/2}).
\ee

Inserting 
\eqref{eq:bta1} and \eqref{eq:bta2} in \eqref{eq:bta}, using
 Lemma \ref{lem:whbbwhbg}, we get
\bse
\wh{\bta_\tau}-\bta_{\tau0}
&=&\C_1(\wh\bb-\bb_0)+\C_2\{\wh\bg(\wh\bb)-\bg_0\}\\
&&+\bb_0\left[n^{-1}\sumi q'(\x_i\trans\bb_0)-E\{q'(\X\trans\bb_0)\}\right]+\r\\
&=&\C\bSig^{-1}n^{-1}\sumi
\begin{bmatrix}\x_i\{y_i-E(Y|\x_i)\} \\ \B(y_i)-E\{\B(Y)|\x_i\}\end{bmatrix}\\
&&+\bb_0\left[n^{-1}\sumi q'(\x_i\trans\bb_0)-E\{q'(\X\trans\bb_0)\}\right]
+\C_1\r_2+\C_2\r_3+\r,
\ese
where $\|\r\|_2=o_p(n^{-1/2})$ by Proposition \ref{pro:bb} and Lemma \ref{lem:whbg},
hence
\bse
\|\C_1\r_2+\C_2\r_3+\r\|_2=o_p(n^{-1/2})
\ese
by Lemma \ref{lem:whbbwhbg} and \eqref{eq:c2}. In addition, we obviously have 
\bse
\cov[\X\{Y-E(Y|\X)\},q'(\X\trans\bb_0)] &=& \0_p, \\
\cov[\B(Y)-E\{\B(Y)|\X\},q'(\X\trans\bb_0)] &=& \0_m.
\ese
Thus with $\bSig_{\bta_\tau}=\C\bSig^{-1}\C\trans+\bb_0\bb_0\trans\var\{q'(\X\trans\bb_0)\}$ defined in Theorem \ref{th:bta}, 
$\bSig_{\bta_\tau}^{-1/2}\sqrt{n}(\wh\bta_\tau-\bta_{\tau0})$ converges
to the normal distribution with mean $\0$ and variance $\I$.
\qed

\subsection{Proof of Proposition \ref{pro:effbb}} \label{proof:effbb}

In terms of estimating $\bb$, the score function is 
$\S_\bb=y\x-E(Y\x\mid\x)$ and the nuisance tangent space is 
\bse 
\Lambda&=&[\a(y)-E\{\a(Y)\mid\x\}+\b(\x): \forall \a(y), \b(\x)\in 
{\cal R}^p, E\{\b(\X)\}=\0]. 
\ese 
 Its orthogonal complement is 
\bse 
\Lambda^\perp=(\a(y,\x)-E\{\a(Y,\x)\mid\x\}: 
E\{\a(y,\X)\mid y\}=E[E\{\a(Y,\X)\mid \X\}\mid y]). 
\ese 
Thus, the efficient score is 
$\S\eff=y\x-\a_0(y)-E\{Y\x-\a_0(Y)\mid\x\}$, where 
$\a_0(y)$ satisfies 
\be\label{eq:a0}
\a_0(y)-E[E\{\a_0(Y)\mid\X\}\mid y]=
E(y\X\mid y)-E\{E(Y\X\mid\X)\mid y\}. 
\ee 
Thus, the efficient variance is $\{E(\S\eff^{\otimes2})\}^{-1}$. 

To show that the MLE estimator $\wh\bb$ is actually efficient, we only 
need to show 
$\bSig^*\to E(\S\eff^{\otimes2})$ when $n\to\infty$. 
Let $\a_0(y)=\Lambda\B(y)+O(h^q)$, where $\Lambda$ is a $p\times m$
coefficient matrix. Then (\ref{eq:a0}) implies that 
\bse 
&& \Lambda E\{\B(Y)^{\otimes2}\}-\Lambda E[E\{\B(Y)\mid\X\}^{\otimes2}] \\
&=&
E[E\{Y\X-E(Y\X\mid\X)\B(Y)\trans\mid\X\}]+O(h^q)E\{\B\trans(Y)\}\\
&=&E[\X\cov\{Y,\B(Y)\mid\X\}]+O(h^q)E\{\B\trans(Y)\}, 
\ese 
i.e. 
\bse 
\a_0(y)&=&\Lambda\B(y)=\bSig_{12}\bSig_{22}^{-1}\B(y)+O(h^q) 
E\{\B\trans(Y)\}\bSig_{22}^{-1}\B(y)\\
&=&\bSig_{12}\bSig_{22}^{-1}\B(y)+O(h^{q-1/2}), 
\ese 
where we used $\|E\{\B(Y)\}\|_2=O(h^{1/2}), \|\bSig_{22}^{-1}\|_2\asymp h^{-1}$, 
and $\|\B(\cdot)\|_2=O(1)$. 
Therefore 
\bse 
&&E(\S\eff^{\otimes2}) -\bSig^*\\
&=&E\left\{(
Y\X-E(Y\X\mid\X) 
-
\bSig_{12}\bSig_{22}^{-1}[\B(y)-E\{B(Y)\mid\X])^{\otimes2}\right\}-\bSig^*+O(h^{q-1/2})\\
&=&o(1). 
\ese 
\qed

\subsection{Proof of Theorem \ref{th:effxi}} \label{proof:effxi}

Since 
\bse 
E(\bphi\eff^{\otimes2})&=&E\left(
[\bb v(\bb\trans\X)-\bb E\{v(\bb\trans\X)\}]^{\otimes2}\right)\\
&&+E\left(
[\bb Y^2+\a(Y) +\M\X Y 
-E\{\bb Y^2+\a(Y)+\M\X Y\mid\X\}]^{\otimes2}\right), 
\ese 
and 
\bse 
\bSig_\bxi=\A\bSig^{-1}\A\trans+\bb^{\otimes2}\var\{v(\bb\trans\X)\}, 
\ese 
so we only need to show 
\bse 
E\left(
[\bb Y^2+\a(Y) +\M\X Y 
-E\{\bb Y^2+\a(Y)+\M\X Y\mid\X\}]^{\otimes2}\right)-
\A\bSig^{-1}\A\trans\to\0. 
\ese 
Now we have $\bb y^2+\a(y)=\bLam\B(y)+O(h^q)$, where $\bLam\in{\cal 
  R}^{p\times m}$. Then (\ref{eq:a}) and (\ref{eq:M}) imply 
\bse 
&&-\bLam E[\B(y) -E\{\B(Y) \mid \X\}\mid y]
+\bb E\{ y^2 -E(Y^2 \mid \X)\mid y\}
+O(h^q)\n\\
&=&2\bb E[ yE(Y\mid\X)-
    E\{YE(Y\mid\X) \mid \X\}\mid y]
+\M E[\X\{ y - E( Y \mid \X)\}\mid y]\\
&=&2\bb E[ yE(Y\mid\X)-
    E\{YE(Y\mid\X) \mid \X\}\mid y]\\
&&+\left(E\{v(\bb\trans\X)\}\I 
-E\left[2\bb\X\trans Y v(\bb\trans\X) +\{\bLam\B(Y) -\bb Y^2\}
\{Y-E(Y\mid\X)\}\X\trans 
\right]\right)\\
&&\times \bSig_{11}^{-1}
E[\X\{ y - E( Y \mid \X)\}\mid y]. 
\ese 
Multiplying $\B\trans(y)$ on both sides and taking expectation lead 
to 
\bse 
&&-\bLam \bSig_{22}
+\bb E[\{ Y^2 -E(Y^2 \mid \X)\}\B\trans(Y)]
+O(h^q) E\{\B\trans(Y)\}\\
&=&2\bb E[ YE(Y\mid\X) \B\trans(Y)-
    \{E(Y\mid\X)\}^2\B\trans(Y)]\\
&&+\left(E\{v(\bb\trans\X)\}\I 
-E\left[2\bb\X\trans Y v(\bb\trans\X) +\{\bLam\B(Y)-\bb Y^2\}
\{Y-E(Y\mid\bb\trans\X)\}\X\trans 
\right]\right)\\
&&\times \bSig_{11}^{-1}
E[\X\{ Y - E( Y \mid \X)\}\B\trans(Y)]\\
&=&2\bb E[Y\cov\{Y,\B(Y)\mid\X\}]
+E\{v(\bb\trans\X)\}\bSig_{11}^{-1}\bSig_{12}
-E\left\{2\bb\X\trans Y v(\bb\trans\X) 
\right\}\bSig_{11}^{-1}\bSig_{12}\\
&&-\bLam \bSig_{21}\bSig_{11}^{-1}\bSig_{12}
+\bb E[ Y^2 
\{Y-E(Y\mid\bb\trans\X)\}\X\trans]\bSig_{11}^{-1}\bSig_{12}\\
&=&2\bb E[Y\cov\{Y,\B(Y)\mid\X\}]+\A_1 \bSig_{11}^{-1}\bSig_{12}
-\bLam \bSig_{21}\bSig_{11}^{-1}\bSig_{12}, 
\ese 
hence 
\bse 
&&-\bLam \bSig_{22}+\bLam \bSig_{21}\bSig_{11}^{-1}\bSig_{12}\\
&=&2\bb E[Y\cov\{Y,\B(Y)\mid\X\}]
+\A_1\bSig_{11}^{-1}\bSig_{12}
-\bb E[\{ Y^2 -E(Y^2 \mid \X)\}\B\trans(Y)]\\
&&+O(h^q) E\{\B\trans(Y)\}\\
&=&\A_1\bSig_{11}^{-1}\bSig_{12}-\A_2 
+O(h^q) E\{\B\trans(Y)\}, 
\ese 
and 
\bse 
\bb y^2+\a(y)&=&(\A_1\bSig_{11}^{-1}\bSig_{12}-\A_2) 
(\bSig_{21}\bSig_{11}^{-1}\bSig_{12}-\bSig_{22})^{-1}\B(y)\\
&&+
O(h^q) E\{\B\trans(Y)\}(\bSig_{21}\bSig_{11}^{-1}\bSig_{12}-\bSig_{22})^{-1}\B(y)\\
&=&\U \B(y)+O(h^{q-1/2}), 
\ese 
where, for notational brevity, 
\bse 
\U&\equiv&(\A_1\bSig_{11}^{-1}\bSig_{12}-\A_2) 
(\bSig_{21}\bSig_{11}^{-1}\bSig_{12}-\bSig_{22})^{-1}. 
\ese 
The above holds by $\|(\bSig_{22}-\bSig_{21}\bSig_{11}^{-1}\bSig_{12})^{-1}\|_2=O(h^{-1})$. 
This is because by Conditions \ref{con:bb} and \ref{con:bSig}, 
all eigenvalues of $\bSig$ are of order either 1 or $h$, 
which implies $\bSig^{-1}$ has all eigenvalues either of order 1 or $h^{-1}$. 
Since $(\bSig_{22}-\bSig_{21}\bSig_{11}^{-1}\bSig_{12})^{-1}$ is a block diagonal element of $\bSig^{-1}$, 
its eigenvalues are of order 1 or $h^{-1}$ as well. 
Then 
\bse 
\M&=&
\left(E\{v(\bb\trans\X)\}\I 
-E\{2\bb\X\trans Y v(\bb\trans\X)\}
 -\U  E\left[\B(Y)\{Y-E(Y\mid\bb\trans\X)\}\X\trans 
\right]\right.\\
&&\left. +\bb  E\left[Y^2\{Y-E(Y\mid\bb\trans\X)\}\X\trans 
\right]
\right) \bSig_{11}^{-1}
+O(h^{q-1/2})\\
&=&\left(E\{v(\bb\trans\X)\}\I 
-E\{2\bb\X\trans Y v(\bb\trans\X)\}
 -\U \bSig_{21}+\bb  E\left[Y^2\{Y-E(Y\mid\X)\}\X\trans\right]
 \right)\\
&&\times\bSig_{11}^{-1}+O(h^{q-1/2})\\
&=&(\A_1 -\U \bSig_{21})\bSig_{11}^{-1}+O(h^{q-1/2}). 
\ese 
Hence 
\bse 
\bb y^2+\a(y)-E\{\bb Y^2+\a(Y)\mid\x\}
=\U[
\B(y)-E\{\B(Y)\mid\x\}]+O(h^{q-1/2}), 
\ese 
and 
\bse 
E\{\var(\M\X Y\mid\X)\}
&=&\M \bSig_{11}\M\trans,\\
E[\var\{\bb Y^2+\a(Y)\mid\X\}]
&=&\U\bSig_{22}\U\trans+O(h^{q-1/2}),\\
E[\cov\{\M\X Y,\bb Y^2+\a(Y)\mid\X\}]
&=&\M \bSig_{12}\U\trans +O(h^{q-1/2}). 
\ese 
Thus, noting that $\A=[\A_1, \A_2]$, 
\bse 
&&E\left(
[\bb Y^2+\a(Y) +\M\X Y 
-E\{\bb Y^2+\a(Y)+\M\X Y\mid\X\}]^{\otimes2}\right)\\
&=&\M\bSig_{11}\M\trans+
\U\bSig_{22}\U\trans+
\M \bSig_{12}\U\trans 
+\U\bSig_{21}\M\trans 
+O(h^{q-1/2})\\
&=&(\A_1 -\U \bSig_{21})\bSig_{11}^{-1}(\A_1\trans 
-\bSig_{12}\U\trans) +\U\bSig_{22}\U\trans+
(\A_1 -\U \bSig_{21})\bSig_{11}^{-1}\bSig_{12}\U\trans\\
&&+
\U\bSig_{21}\bSig_{11}^{-1}(\A_1\trans 
-\bSig_{12}\U\trans)+O(h^{q-1/2})\\
&=&\A_1 \bSig_{11}^{-1}\A_1\trans 
+\U(\bSig_{22}
 -
\bSig_{21}\bSig_{11}^{-1}\bSig_{12})\U\trans+O(h^{q-1/2})\\
&=&
\A_1\bSig_{11}^{-1}\A_1\trans 
+
(\A_1\bSig_{11}^{-1}\bSig_{12}-\A_2) 
(\bSig_{22}-\bSig_{21}\bSig_{11}^{-1}\bSig_{12})^{-1}
(\bSig_{21}\bSig_{11}^{-1}\A_1\trans-\A_2\trans) \\
&&+O(h^{q-1/2})\\
&=&\A 
\left(\begin{array}{cc}
\bSig_{11}^{-1}+\bSig_{11}^{-1}\bSig_{12}(\bSig_{22}-\bSig_{21}\bSig_{11}^{-1}\bSig_{12})^{-1}\bSig_{21}\bSig_{11}^{-1}&
                                                                                                          -\bSig_{11}^{-1}\bSig_{12}(\bSig_{22}-\bSig_{21}\bSig_{11}^{-1}\bSig_{12})^{-1}\\
-(\bSig_{22}-\bSig_{21}\bSig_{11}^{-1}\bSig_{12})^{-1}\bSig_{21}\bSig_{11}^{-1}&
(\bSig_{22}-\bSig_{21}\bSig_{11}^{-1}\bSig_{12})^{-1}
\end{array}\right)\A\trans\\
&&+O(h^{q-1/2})\\
&=&\A\bSig^{-1}\A\trans+o(1). 
\ese 
\qed

\subsection{Proof of Theorem \ref{th:effbta}}

By \eqref{eq:efftau}
\bse 
E(\bphi\eff^{\otimes2})
&=&E\left([\bb q'(\nu)-\bb E\{q'(\nu)\}]^{\otimes2}\right)\\
&&+E\left([\M_1\X Y+\a(Y)-E\{\M_1\X Y+\a(Y)\mid\X\}]^{\otimes2}\right), 
\ese 
and 
\bse 
\bSig_{\bta_\tau}=\C\bSig^{-1}\C\trans+\bb^{\otimes2}\var\{q'(\nu)\}
\ese 
in Theorem \ref{th:bta}, 
 we only need to show 
\bse 
E\left([\M_1\X Y+\a(Y)-E\{\M_1\X Y+\a(Y)\mid\X\}]^{\otimes2}\right)-\C\bSig^{-1}\C\trans\to\0. 
\ese 
Now we have $\a(y)=\bLam\B(y)+O(h^q)$ where $\bLam\in{\cal R}^{p\times m}$.
Then \eqref{eq:atau} and \eqref{eq:Mtau} imply 
\bse 
&&-\bLam E[\B(y)-E\{\B(Y)\mid\X\}\mid y]+O(h^q)\n\\
&=&-\bb E\{r(y,\nu)\mid y\}+\M_1 E[\X\{y-E(Y\mid\X)\}\mid y]\\
&=&-\bb E\{r(y,\nu)\mid y\}\\
&&+\left(E\{q'(\nu)\}\I+\bb E\{\X\trans q''(\nu)\}
-E\left[\{\bLam\B(Y)+O(h^q)\}\X\trans\{Y-E(Y\mid\X)\}\right]\right)\\
&&\times\bSig_{11}^{-1}E[\X\{y-E(Y\mid\X)\}\mid y]. 
\ese 
Multiplying $\B\trans(y)$ on both sides above and taking expectation,
incorporating \eqref{eq:r},
we get
\bse 
&&-\bLam\bSig_{22}+O(h^q) E\{\B\trans(Y)\}\\
&=&-\bb E\{r(Y,\nu)\B\trans(Y)\}\\
&&+\left(E\{q'(\nu)\}\I+\bb E\{\X\trans q''(\nu)\}
-E\left[\{\bLam\B(Y)+O(h^q)\}\X\trans\{Y-E(Y\mid\X)\}\right]\right)\\
&&\times\bSig_{11}^{-1}E[\X\{Y-E(Y\mid\X)\}\B\trans(Y)]\\ 
&=&-\bb E\Bigg(\frac{E\{\epsilon Y\B\trans(Y)|\X\}}{f\{q(\nu),\nu\}}-q'(\nu)\B\trans\{q(\nu)\}
-\frac{E\{\epsilon\B\trans(Y)|\X\}}{f\{q(\nu),\nu\}}[q(\nu)+q'(\nu)\nu+q'(\nu)c'\{q(\nu)\}]\Bigg)\\
&&+\C_1\bSig_{11}^{-1}\bSig_{12}-\bLam \bSig_{21}\bSig_{11}^{-1}\bSig_{12}+O(h^q)\bSig_{11}^{-1}\bSig_{12}\\
&=&-\bLam \bSig_{21}\bSig_{11}^{-1}\bSig_{12}+\C_1\bSig_{11}^{-1}\bSig_{12}-\C_2+O(h^q)\bSig_{11}^{-1}\bSig_{12}, 
\ese 
hence 
\bse 
\bLam(\bSig_{22}-\bSig_{21}\bSig_{11}^{-1}\bSig_{12})
=-\C_1\bSig_{11}^{-1}\bSig_{12}+\C_2+O(h^{q+1/2}), 
\ese
since $\|E\{\B\trans(Y)\}\|_2=O(h^{1/2})$ and $\|\bSig_{12}\|_2=O(h^{1/2})$ by Lemma \ref{lem:becov}.
Then 
\bse
\a(y)=\U\B(y)+O(h^{q-1/2}),
\ese 
where for notational brevity, 
\bse 
\U\equiv(-\C_1\bSig_{11}^{-1}\bSig_{12}+\C_2)(\bSig_{22}-\bSig_{21}\bSig_{11}^{-1}\bSig_{12})^{-1}.
\ese
The above holds by $\|(\bSig_{22}-\bSig_{21}\bSig_{11}^{-1}\bSig_{12})^{-1}\|_2=O(h^{-1})$
under Conditions \ref{con:bb} and \ref{con:bSig}. 
Then by \eqref{eq:Mtau},
\bse 
&&\M_1\\
&=&(E\{q'(\nu)\}\I+\bb E\{\X\trans q''(\nu)\}
-\U E\left[\B(Y)\X\trans\{Y-E(Y\mid\X)\}\right])\bSig_{11}^{-1}+O(h^{q-1/2})\\
&=&(\C_1-\U\bSig_{21})\bSig_{11}^{-1}+O(h^{q-1/2}). 
\ese 
Hence 
\bse 
\a(y)-E\{\a(Y)\mid\x\}=\U[\B(y)-E\{\B(Y)\mid\x\}]+O(h^{q-1/2}), 
\ese 
and 
\bse 
E\{\var(\M_1\X Y\mid\X)\}&=&\M_1 \bSig_{11}\M_1\trans,\\
E[\cov\{\M\X Y,\a(Y)\mid\X\}]&=&\M_1\bSig_{12}\U\trans+O(h^{q-1/2}),\\
E[\var\{\a(Y)\mid\X\}]&=&\U\bSig_{22}\U\trans+O(h^{q-1/2}). 
\ese 
Thus noting that $\C=[\C_1,\C_2]$, we obtain
\bse 
&&E\left([\M_1\X Y+\a(Y)-E\{\M_1\X Y+\a(Y)\mid\X\}]^{\otimes2}\right)\\
&=&\M_1\bSig_{11}\M_1\trans+\M_1\bSig_{12}\U\trans+\U\bSig_{21}\M_1\trans+\U\bSig_{22}\U\trans+O(h^{q-1/2})\\
&=&(\C_1-\U\bSig_{21})\bSig_{11}^{-1}(\C_1\trans-\bSig_{12}\U\trans)
+(\C_1-\U\bSig_{21})\bSig_{11}^{-1}\bSig_{12}\U\trans\\
&&+\U\bSig_{21}\bSig_{11}^{-1}(\C_1\trans-\bSig_{12}\U\trans)+\U\bSig_{22}\U\trans+O(h^{q-1/2})\\
&=&\C_1\bSig_{11}^{-1}\C_1\trans +\U(\bSig_{22}-\bSig_{21}\bSig_{11}^{-1}\bSig_{12})\U\trans+O(h^{q-1/2})\\
&=&\C_1\bSig_{11}^{-1}\C_1\trans 
+(-\C_1\bSig_{11}^{-1}\bSig_{12}+\C_2)
(\bSig_{22}-\bSig_{21}\bSig_{11}^{-1}\bSig_{12})^{-1}
(-\bSig_{21}\bSig_{11}^{-1}\C_1\trans+\C_2\trans) \\
&&+O(h^{q-1/2})\\
&=&\C 
\left(\begin{array}{cc}
\bSig_{11}^{-1}+\bSig_{11}^{-1}\bSig_{12}(\bSig_{22}-\bSig_{21}\bSig_{11}^{-1}\bSig_{12})^{-1}\bSig_{21}\bSig_{11}^{-1}
&-\bSig_{11}^{-1}\bSig_{12}(\bSig_{22}-\bSig_{21}\bSig_{11}^{-1}\bSig_{12})^{-1}\\
-(\bSig_{22}-\bSig_{21}\bSig_{11}^{-1}\bSig_{12})^{-1}\bSig_{21}\bSig_{11}^{-1}&
(\bSig_{22}-\bSig_{21}\bSig_{11}^{-1}\bSig_{12})^{-1}
\end{array}\right)\C\trans\\
&=&\C\bSig^{-1}\C\trans+o(1).
\ese 
\qed

\subsection{Lemmas}

We now introduce a lemma for the analysis of the discrete response case.

\begin{Lem}\label{lem:lipschitz}
Let $\btheta^*\equiv(\bb^{*\rm T},\bg^{*\rm T})\trans$ and
$\|\btheta^*-\btheta_0\|_2=o(1)$. 
Under Conditions \ref{con:catbdd}-\ref{con:catlipschitz},
uniformly with respect to $\x$,
\begin{enumerate}[label=(\roman*)]
    \item
    $\|\p(\x,\btheta^*)-\p(\x,\btheta_0)\|_2
    =O(\|\bb^*-\bb_0\|_2 + \|\bg^*-\bg_0\|_2)$,
    \item
    $\|\p_1(\x,\btheta^*)-\p_1(\x,\btheta_0)\|_2
    =O(\|\bb^*-\bb_0\|_2 + \|\bg^*-\bg_0\|_2)$,
    \item
    $\|\p_2(\x,\btheta^*)-\p_2(\x,\btheta_0)\|_2
    =O(\|\bb^*-\bb_0\|_2 + \|\bg^*-\bg_0\|_2)$,
    \item
    $\|\var\{\B(Y)\mid\x,\btheta^*\}-\var\{\B(Y)\mid\x,\btheta_0\}\|_2
    =O(\|\bb^*-\bb_0\|_2 + \|\bg^*-\bg_0\|_2)$.
\end{enumerate}
If $\btheta^*$ satisfies $\|\btheta^*-\btheta_0\|_2=o_p(1)$,
the above results hold in probability.
\end{Lem}
\begin{proof}
For some $\wt\btheta$ on the line connecting $\btheta^*$ and $\btheta_0$,
we have
\bse
\p(\x,\btheta^*)-\p(\x,\btheta_0)
&=&\frac{\partial\p(\x,\wt\btheta)}{\partial\bb\trans}(\bb^*-\bb_0)
+\frac{\partial\p(\x,\wt\btheta)}{\partial\bg\trans}(\bg^*-\bg_0).
\ese
First, by Conditions \ref{con:catbdd} and \ref{con:catlipschitz},
\bse
\left\|\frac{\partial\p(\x,\wt\btheta)}{\partial\bb\trans}(\bb^*-\bb_0)\right\|_2
&=&\left\|\left\{\p_1(\x,\wt\btheta)-E(Y\mid\x,\wt\btheta)\p(\x,\wt\btheta)\right\}
\x\trans(\bb^*-\bb_0)\right\|_2\\
&\leq&\left[\sum_{y=1}^m\left\{y-E(Y\mid\x,\wt\btheta)\right\}^2\pr(Y=y\mid\x,\wt\btheta)^2\right]^{1/2}
O(\|\bb^*-\bb_0\|_2)\\
&\leq&\sum_{y=1}^m\left|y-E(Y\mid\x,\wt\btheta)\right|\pr(Y=y\mid\x,\wt\btheta)O(\|\bb^*-\bb_0\|_2)\\
&\leq&2E(Y\mid\x,\wt\btheta)O(\|\bb^*-\bb_0\|_2)\\
&\leq&2\{E(Y\mid\x,\btheta_0)+O(\|\wt\btheta-\btheta_0\|_2)\}O(\|\bb^*-\bb_0\|_2)\\
&=&O(\|\bb^*-\bb_0\|_2).
\ese
The last inequality holds since $E(Y\mid\x,\btheta_0)$ is bounded by Condition \ref{con:catbdd},
and $\|\wt\btheta-\btheta_0\|_2=o(1)$ by the definition of $\wt\btheta$.
In addition, for a vector $\a$, let $\diag(\a)$ be the diagonal
matrix with entries equal to the elements of $\a$, then
\bse
\left\|\frac{\partial\p(\x,\wt\btheta)}{\partial\bg\trans}(\bg^*-\bg_0)\right\|_2
&=&\left\|\left[\diag\left\{\p(\x,\wt\btheta)\right\}-\p(\x,\wt\btheta)\p\trans(\x,\wt\btheta)\right]
(\bg^*-\bg_0)\right\|_2\\
&\leq&\left[\|\p(\x,\wt\btheta)\|_\infty+\|\p(\x,\wt\btheta)\|_2^2\right]
\|\bg^*-\bg_0\|_2\\
&=&O(\|\bg^*-\bg_0\|_2)
\ese
because $\|\p(\x,\wt\btheta)\|_\infty\leq\|\p(\x,\wt\btheta)\|_2\leq 1$.
Hence, we get
\bse
\|\p(\x,\btheta^*)-\p(\x,\btheta_0)\|_2
&=&O(\|\bb^*-\bb_0\|_2+\|\bg^*-\bg_0\|_2).
\ese
Similarly for some $\wt\btheta$ on the line connecting $\btheta^*$ and $\btheta_0$,
\bse
\left\|\frac{\partial\p_1(\x,\wt\btheta)}{\partial\bb\trans}(\bb^*-\bb_0)\right\|_2
&=&\left\|\left\{\p_2(\x,\wt\btheta)-E(Y\mid\x,\wt\btheta)\p_1(\x,\wt\btheta)\right\}
\x\trans(\bb^*-\bb_0)\right\|_2\\
&\leq&\sum_{y=1}^m\left|y-E(Y\mid\x,\wt\btheta)\right|y\pr(Y=y\mid\x,\wt\btheta)O(\|\bb^*-\bb_0\|_2)\\
&\leq&\left[E(Y^2\mid\x,\btheta_0)+\{E(Y\mid\x,\btheta_0)\}^2+O(\|\wt\btheta-\btheta_0\|_2)\right]
O(\|\bb^*-\bb_0\|_2)\\
&=&O(\|\bb^*-\bb_0\|_2)
\ese
by Conditions \ref{con:catbdd} and \ref{con:catlipschitz}.
In addition,
\bse
\left\|\frac{\partial\p_1(\x,\wt\btheta)}{\partial\bg\trans}(\bg^*-\bg_0)\right\|_2
&=&\left\|\left[\diag\left\{\p_1(\x,\wt\btheta)\right\}-\p_1(\x,\wt\btheta)\p\trans(\x,\wt\btheta)\right]
(\bg^*-\bg_0)\right\|_2\\
&\leq&\left[\left\|\diag\left\{\p_1(\x,\wt\btheta)\right\}\right\|_2
+\left\|\p_1(\x,\wt\btheta)\p\trans(\x,\wt\btheta)\right\|_2\right]
\|\bg^*-\bg_0\|_2\\
&\leq&\left[\|\p_1(\x,\wt\theta)\|_\infty+\|\p_1(\x,\wt\btheta)\|_2\|\p(\x,\wt\btheta)\|_2\right]
\|\bg^*-\bg_0\|_2\\
&\leq&O(\|\bg^*-\bg_0\|_2),
\ese
because $\|\p_1(\x,\wt\theta)\|_\infty\leq\|\p_1(\x,\wt\theta)\|_2\leq\|\p_1(\x,\wt\theta)\|_1
=E(Y\mid\x,\wt\theta)=E(Y\mid\x)+o(1)$
and $E(Y\mid\x)$ is bounded by Condition \ref{con:catbdd}.
Therefore,
\bse
\|\p_1(\x,\btheta^*)-\p_1(\x,\btheta_0)\|_2
&=&\left\|\frac{\partial\p_1(\x,\wt\btheta)}{\partial\bb\trans}(\bb^*-\bb_0)
+\frac{\partial\p_1(\x,\wt\btheta)}{\partial\bg\trans}(\bg^*-\bg_0)\right\|_2\\
&=&O(\|\bb^*-\bb_0\|_2+\|\bg^*-\bg_0\|_2).
\ese
Also, by Conditions \ref{con:catbdd} and \ref{con:catlipschitz},
\bse
\left\|\frac{\partial\p_2(\x,\wt\btheta)}{\partial\bb\trans}(\bb^*-\bb_0)\right\|_2
&=&\left\|\left\{\p_3(\x,\wt\btheta)-E(Y\mid\x,\wt\btheta)\p_2(\x,\wt\btheta)\right\}
\x\trans(\bb^*-\bb_0)\right\|_2\\
&\leq&\sum_{y=1}^m\left|y-E(Y\mid\x,\wt\btheta)\right|y^2\pr(Y=y\mid\x,\wt\btheta)O(\|\bb^*-\bb_0\|_2)\\
&\leq&\left[E(Y^3\mid\x)+E(Y\mid\x)E(Y^2\mid\x)+O(\|\wt\btheta-\btheta_0\|_2)\right]
O(\|\bb^*-\bb_0\|_2)\\
&=&O(\|\bb^*-\bb_0\|_2),
\ese
and
\bse
\left\|\frac{\partial\p_2(\x,\wt\btheta)}{\partial\bg\trans}(\bg^*-\bg_0)\right\|_2
&=&\left\|\left[\diag\left\{\p_2(\x,\wt\btheta)\right\}-\p_2(\x,\wt\btheta)\p\trans(\x,\wt\btheta)\right]
(\bg^*-\bg_0)\right\|_2\\
&\leq&\left[\left\|\diag\left\{\p_2(\x,\wt\btheta)\right\}\right\|_2
+\left\|\p_2(\x,\wt\btheta)\p\trans(\x,\wt\btheta)\right\|_2\right]
\|\bg^*-\bg_0\|_2\\
&\leq&\left[\|\p_2(\x,\wt\theta)\|_\infty+\|\p_2(\x,\wt\btheta)\|_2\|\p(\x,\wt\btheta)\|_2\right]
\|\bg^*-\bg_0\|_2\\
&\leq&O(\|\bg^*-\bg_0\|_2).
\ese
The last inequality holds because
$\|\p_2(\x,\wt\theta)\|_\infty\leq\|\p_2(\x,\wt\theta)\|_2\leq\|\p_2(\x,\wt\theta)\|_1
=E(Y^2\mid\x,\wt\theta)=E(Y^2\mid\x)+o(1)$
and $E(Y^2\mid\x)$ is bounded by Condition \ref{con:catbdd}.
Then we get
\bse
\|\p_2(\x,\btheta^*)-\p_2(\x,\btheta_0)\|_2
&=&\left\|\frac{\partial\p_2(\x,\wt\btheta)}{\partial\bb\trans}(\bb^*-\bb_0)
+\frac{\partial\p_2(\x,\wt\btheta)}{\partial\bg\trans}(\bg^*-\bg_0)\right\|_2\\
&=&O(\|\bb^*-\bb_0\|_2+\|\bg^*-\bg_0\|_2).
\ese
Now,
\bse
&&\left\|\var\{\B(Y)\mid\x,\btheta^*\}-\var\{\B(Y)\mid\x,\btheta_0\}\right\|_2\\
&\leq&\left\|\diag\left\{\p(\x,\btheta^*)-\p(\x,\btheta_0)\right\}\right\|_2
+\left\|\p(\x,\btheta^*)\p\trans(\x,\btheta^*) - \p(\x,\btheta_0)\p\trans(\x,\btheta_0)\right\|_2\\
&=&\|\p(\x,\btheta^*)-\p(\x,\btheta_0)\|_\infty\\
&&+\left\|\p(\x,\btheta^*)\{\p(\x,\btheta^*)-\p(\x,\btheta_0)\}\trans
+\{\p(\x,\btheta^*)-\p(\x,\btheta_0)\}\p\trans(\x,\btheta_0)\right\|_2\\
&\leq&\|\p(\x,\btheta^*)-\p(\x,\btheta_0)\|_2\\
&&+\left|\{\p(\x,\btheta^*)-\p(\x,\btheta_0)\}\trans\p(\x,\btheta^*)\right|
+\left|\p\trans(\x,\btheta_0)\{\p(\x,\btheta^*)-\p(\x,\btheta_0)\}\right|\\
&\leq&3\left\|\p(\x,\btheta^*)-\p(\x,\btheta_0)\right\|_2\\
&=&O(\|\bb^*-\bb_0\|_2+\|\bg^*-\bg_0\|_2).
\ese
Since the support of $\x$ is compact by Condition \ref{con:catbdd},
the about results hold uniformly with respect to $\x$.
\end{proof}

\subsection{Proof of Proposition \ref{pro:catbbbg}} \label{proof:catbbbg}

First, we will show that
$\|\wh\bb-\bb_0\|_2=O_p(n^{-1/2})$ and $\|\wh\bg-\bg_0\|_2=O_p(n^{-1/2}m^{1/2})$.
Note that the Hessian of $l(\bb,\bg)$, i.e. 
\bse
&& \begin{bmatrix}
\partial^2 l(\bb,\bg)/\partial\bb\partial\bb\trans & \partial^2 l(\bb,\bg)/\partial\bb\partial\bg\trans \\
\partial^2 l(\bb,\bg)/\partial\bg\partial\bb\trans & \partial^2 l(\bb,\bg)/\partial\bg\partial\bg\trans
\end{bmatrix} \\
&=& -\sumi \begin{bmatrix}
\x_i\x_i\trans\var(Y\mid\x_i,\btheta) & \x_i\cov\{Y,\B(Y)\mid\x_i,\btheta\} \\
[\x_i\cov\{Y,\B(Y)\mid\x_i,\btheta\}]\trans & \var\{\B(Y)\mid\x_i,\btheta\}
\end{bmatrix} \\
&=& -\sumi
\cov\left\{\begin{bmatrix}\x_iY \\ \B(Y)\end{bmatrix},
\begin{bmatrix}\x_iY \\ \B(Y)\end{bmatrix} \mid \x_i,\btheta\right\}
\ese
is negative definite for any $\btheta$,
which implies that a local maximizer is the global maximizer.
Hence similarly to the proof of Proposition \ref{pro:bg},
it suffices to show for any $\epsilon>0$
there exists constants $C_\bb,C_\bg>0$ such that
\be\label{eq:catball}
\pr\left\{l(\bb_0,\bg_0)>
\sup_{\|\v_\bb\|_2=C_\bb,\|\v_\bg\|_2=C_\bg}
l(\bb_0+n^{-1/2}\v_\bb,\bg_0+n^{-1/2}m^{1/2}\v_\bg)\right\} \geq 1-5\epsilon
\ee
for a sufficiently large $n$.
Now, by the Taylor expansion,
\be\label{eq:cattaylor}
&& l(\bb_0+n^{-1/2}\v_\bb,\bg_0+n^{-1/2}m^{1/2}\v_\bg) - l(\bb_0,\bg_0)\n\\
&=& n^{-1/2}\frac{\partial l(\bb_0,\bg_0)}{\partial\bb\trans}\v_\bb
+ n^{-1/2}m^{1/2}\frac{\partial l(\bb_0,\bg_0)}{\partial\bg\trans}\v_\bg\n\\
&&+ \frac{1}{2}n^{-1}\v_\bb\trans\frac{\partial^2 l(\bb^*,\bg^*)}{\partial\bb\partial\bb\trans}\v_\bb
+ \frac{1}{2}n^{-1}m\v_\bg\trans\frac{\partial^2 l(\bb^*,\bg^*)}{\partial\bg\partial\bg\trans}\v_\bg
+ n^{-1}m^{1/2}\v_\bb\trans\frac{\partial^2 l(\bb^*,\bg^*)}{\partial\bb\partial\bg\trans}\v_\bg
\ee
where $\bb^*=\bb_0+\alpha_1n^{-1/2}\v_\bb$ and $\bg^*=\bg_0+\alpha_2n^{-1/2}m^{1/2}\v_\bg$
for some $\v_\bb,\v_\bg$ such that $\|\v_\bb\|_2=C_\bb,\|\v_\bg\|_2=C_\bg$
and some $\alpha_1\in(0,1),\alpha_2\in(0,1)$.
We first have
\bse
\left\|\frac{\partial l(\bb_0,\bg_0)}{\partial\bb}\right\|_2
= \left\|\sumi\x_i\{y_i-E(Y\mid\x_i)\}\right\|_2
\asymp_p n^{1/2}
\ese
by Conditions \ref{con:catbdd} and \ref{con:catvar}.
Therefore, for any $\epsilon>0$ there exists a constant $0<C_1<\infty$ such that
\be\label{eq:catb}
\pr\left\{\left\|\frac{\partial l(\bb_0,\bg_0)}{\partial\bb}\right\|_2
\leq C_1n^{1/2}\right\} \geq 1-\epsilon.
\ee
Also,
\bse
\left\|\frac{\partial l(\bb_0,\bg_0)}{\partial\bg}\right\|_2
= \left\|\sumi\{\B(y_i)-E\{\B(Y)\mid\x_i\}]\right\|_2
\asymp_p n^{1/2}m^{-1/2},
\ese
because
\bse
\left\|E\left\{\left(\sumi\left[\B(Y_i) - E\{\B(Y)\mid\X_i\}\right]\right)^{\otimes2}\right\}\right\|_2
&=& \left\|n E\left(\left[\B(Y)-E\{\B(Y)\mid\X\}\right]^{\otimes2}\right)\right\|_2 \\
&=& n \left\|E[\var\{\B(Y)\mid\X\}]\right\|_2 \\
&\asymp& nm^{-1}
\ese
by Remark \ref{rem:var22}.
Hence, for any $\epsilon>0$ there exists a constant $0<C_2<\infty$ such that
\be\label{eq:catg}
\pr\left\{\left\|\frac{\partial l(\bb_0,\bg_0)}{\partial\bg}\right\|_2
\leq C_2n^{1/2}m^{-1/2}\right\} \geq 1-\epsilon.
\ee
In addition, noting that
$\|\btheta^*-\btheta_0\|_2=O(n^{-1/2}m^{1/2})=o(1)$ by Condition
\ref{con:catm}, we get
\be\label{eq:R11}
\left\|-n^{-1}\frac{\partial^2 l(\bb^*,\bg^*)}{\partial\bb\partial\bb\trans}-\bSig_{11}\right\|_2
&\leq&\left\|n^{-1}\sumi\x_i\x_i\trans\{\var(Y\mid\x_i,\btheta^*)-\var(Y\mid\x_i)\}\right\|_2 \n\\
&&+\left\|n^{-1}\sumi\x_i\x_i\trans\var(Y\mid\x_i)-E\{\X\X\trans\var(Y\mid\X)\}\right\|_2 \n\\
&=&o_p(1)
\ee
by Conditions \ref{con:catbdd} and \ref{con:catlipschitz}.
Since all eigenvalues of $\bSig_{11}$ are of constant order
by Conditions \ref{con:catbdd} and \ref{con:catvar},
the above implies there exists a constant $0<C_3<\infty$ such that
\be\label{eq:catbb}
\pr\left\{\v_\bb\trans\frac{\partial^2 l(\bb^*,\bg^*)}{\partial\bb\partial\bb\trans}\v_\bb
\leq -C_3C_\bb^2n\right\} \geq 1-\epsilon.
\ee
Further, since $\|\bb^*-\bb_0\|_2=O(n^{-1/2})$ and $\|\bg^*-\bg_0\|_2=O(n^{-1/2}m^{1/2})$,
\be\label{eq:R22}
\left\|-n^{-1}\frac{\partial^2 l(\bb^*,\bg^*)}{\partial\bg\partial\bg\trans}-\bSig_{22}\right\|_2
&\leq&\left\|n^{-1}\sumi[\var\{\B(Y)\mid\x_i,\btheta^*\}-\var\{\B(Y)\mid\x_i\}]\right\|_2 \n\\
&&+\left\|n^{-1}\sumi\var\{\B(Y)\mid\x_i\}-E[\var\{\B(Y)\mid\X\}]\right\|_2 \n\\
&=&O(n^{-1/2}m^{1/2})+O_p(n^{-1/2}m^{-1})\n\\
&=&O_p(n^{-1/2}m^{1/2})\n\\
&=&o_p(m^{-1})
\ee
by Lemma \ref{lem:lipschitz}, Remark \ref{rem:var22},
and $n^{-1/2}m^{1/2}=o(m^{-1})$ by Condition \ref{con:catm}.
Together with Remark \ref{rem:var22}, this implies
\be\label{eq:catgg}
\pr\left\{\v_\bg\trans\frac{\partial^2 l(\bb^*,\bg^*)}{\partial\bg\partial\bg\trans}\v_\bg
\leq -C_4 C_\bg^2 nm^{-1}\right\} \geq 1-\epsilon
\ee
for some constant $0<C_4<\infty$.
Now, we have $\|\bSig_{12}\|_2=O(m^{-1/2})$,
because noting that $\bSig_{11}$ and $\bSig_{22}-\bSig_{21}\bSig_{11}^{-1}\bSig_{12}$
are positive definite by Condition \ref{con:catvar},
for any $\u\in\mR^{m}$ such that $\|\u\|_2=1$ we have
\be\label{eq:bSig12}
c\|\bSig_{12}\u\|_2^2
\leq\|\bSig_{11}^{-1/2}\bSig_{12}\u\|_2^2
=\u\trans\bSig_{21}\bSig_{11}^{-1}\bSig_{12}\u
<\u\trans\bSig_{22}\u
\leq\|\bSig_{22}\|_2
\asymp m^{-1}
\ee
by Remark \ref{rem:var22} for some constant $c>0$. 
Also, since $\|\bb^*-\bb_0\|_2=O(n^{-1/2})$ and $\|\bg^*-\bg_0\|_2=O(n^{-1/2}m^{1/2})$,
\be\label{eq:R12}
\left\|-n^{-1}\frac{\partial^2 l(\bb^*,\bg^*)}{\partial\bb\partial\bg\trans}-\bSig_{12}\right\|_2
&\leq&\left\|n^{-1}\sumi\x_i[\cov\{Y,\B(Y)\mid\x_i,\btheta^*\}-\cov\{Y,\B(Y)\mid\x_i\}]\right\|_2 \n\\
&&+\left\|n^{-1}\sumi\x_i\cov\{Y,\B(Y)\mid\x_i\}-E[\X\cov\{Y,\B(Y)\mid\X\}]\right\|_2 \n\\
&\leq&\left\|n^{-1}\sumi\x_i\{\p_1\trans(\x_i,\btheta^*)-\p_1\trans(\x_i,\btheta_0)\}\right\|_2 \n\\
&&+\left\|n^{-1}\sumi\x_i\{E(Y\mid\x_i,\btheta^*)\p\trans(\x_i,\btheta^*)
-E(Y\mid\x_i,\btheta_0)\p\trans(\x_i,\btheta_0)\}\right\|_2 \n\\
&&+\left\|n^{-1}\sumi\x_i\cov\{Y,\B(Y)\mid\x_i\}-E[\X\cov\{Y,\B(Y)\mid\X\}]\right\|_2\n\\
&=&O(n^{-1/2}m^{1/2})+O_p(n^{-1/2}m^{-1/2})\n\\
&=&o_p(m^{-1/2})
\ee
by Lemma \ref{lem:lipschitz}, Condition \ref{con:catbdd},
and $n^{-1/2}m^{1/2}=o(m^{-1/2})$ by Condition \ref{con:catm}.
Thus we have
\be\label{eq:catbg}
\pr\left\{\v_\bb\trans\frac{\partial^2 l(\bb^*,\bg^*)}{\partial\bb\partial\bg\trans}\v_\bg
\leq C_5 C_\bb C_\bg nm^{-1/2}\right\} \geq 1-\epsilon
\ee
for some constant $0<C_5<\infty$. 
Combining \eqref{eq:cattaylor}, \eqref{eq:catb}, \eqref{eq:catg},
\eqref{eq:catbb}, \eqref{eq:catgg}, and \eqref{eq:catbg},
with probability at least $1-5\epsilon$,
\bse
&& l(\bb_0+n^{-1/2}\v_\bb,\bg_0+n^{-1/2}m^{1/2}\v_\bg) - l(\bb_0,\bg_0) \\
&\leq& C_1C_\bb+C_2C_\bg-\frac{C_3}{2}C_\bb^2-\frac{C_4}{2}C_\bg^2 + C_5C_\bb C_\bg \\
&=& C_\bb\left(C_1-\frac{C_3}{2}C_\bb+C_5C_\bg\right)
+C_\bg\left(C_2-\frac{C_4}{2}C_\bg\right) \\
&<&0
\ese
when $C_\bb>2(C_1+C_5C_\bg)/C_3$ and $C_\bg>2C_2/C_4$, and this proves \eqref{eq:catball}.
Therefore, $\|\wh\bb-\bb_0\|_2=O_p(n^{-1/2})$ and $\|\wh\bg-\bg_0\|_2=O_p(n^{-1/2}m^{1/2})$.

Now, we analyze the asymptotic behavior of $(\wh\bb\trans,\wh\bg\trans)\trans$.
Since $(\wh\bb\trans,\wh\bg\trans)\trans$ is the maximizer of $l(\bb,\bg)$,
letting $(\bb^{*\rm T},\bg^{*\rm T})\trans$ be on the line connecting
$(\wh\bb\trans,\wh\bg\trans)\trans$ and $(\bb_0\trans,\bg_0\trans)\trans$,
\bse
\0
&=&n^{-1}\frac{\partial l(\bb_0,\bg_0)}{\partial\btheta}
-\left\{-n^{-1}\frac{\partial^2 l(\bb^*,\bg^*)}{\partial\btheta\partial\btheta\trans}\right\}
\begin{bmatrix}\wh\bb-\bb_0 \\ \wh\bg-\bg_0\end{bmatrix}\\
&=&n^{-1}\sumi\begin{bmatrix}\x_i\{y_i-E(Y\mid\x_i)\} \\ \B(y_i)-E\{\B(Y)\mid\x_i\}\end{bmatrix}
-\left(\bSig+\begin{bmatrix}\R_{11}&\R_{12}\\\R_{12}\trans&\R_{22} \end{bmatrix}\right)
\begin{bmatrix}\wh\bb-\bb_0 \\ \wh\bg-\bg_0\end{bmatrix},
\ese
where
\bse
\R_{11}&\equiv&
-n^{-1}\frac{\partial^2 l(\bb^*,\bg^*)}{\partial\bb\partial\bb\trans}-\bSig_{11},\\
\R_{12}&\equiv&
-n^{-1}\frac{\partial^2 l(\bb^*,\bg^*)}{\partial\bb\partial\bg\trans}-\bSig_{12},\\
\R_{22}&\equiv&
-n^{-1}\frac{\partial^2 l(\bb^*,\bg^*)}{\partial\bg\partial\bg\trans}-\bSig_{22}.
\ese
Since $\bSig$ is invertible by Condition \ref{con:catvar}, we have
\bse
\begin{bmatrix}\wh\bb-\bb_0 \\ \wh\bg-\bg_0\end{bmatrix}
=
\bSig^{-1}n^{-1}\sumi\begin{bmatrix}\x_i\{y_i-E(Y\mid\x_i)\} \\ \B(y_i)-E\{\B(Y)\mid\x_i\}\end{bmatrix}
-\bSig^{-1}\begin{bmatrix}\R_{11}&\R_{12}\\\R_{12}\trans&\R_{22} \end{bmatrix}
\begin{bmatrix}\wh\bb-\bb_0 \\ \wh\bg-\bg_0\end{bmatrix}.
\ese
Now, note that
\be\label{eq:bsiginv}
\bSig^{-1}
=\begin{bmatrix}\bSig_\bb
&-\bSig_\bb\bSig_{12}\bSig_{22}^{-1}\\
-\bSig_{22}^{-1}\bSig_{21}\bSig_\bb
&\bSig_{22}^{-1}+\bSig_{22}^{-1}\bSig_{21}\bSig_\bb\bSig_{12}\bSig_{22}^{-1} \end{bmatrix}
=\begin{bmatrix}O(1)&O(m^{1/2})\\O(m^{1/2})&O(m)\end{bmatrix}
\ee
in terms of the 2-norms of the block matrices.
This is because $\bSig_\bb^{-1}=\bSig_{11}-\bSig_{12}\bSig_{22}^{-1}\bSig_{21}$
is positive definite by Condition \ref{con:catvar},
$\|\bSig_{12}\|_2=O(m^{-1/2})$ by \eqref{eq:bSig12},
and $\|\bSig_{22}^{-1}\|_2\asymp m$ by Remark \ref{rem:var22}.
In addition, using the fact that
$\|\bb^*-\bb_0\|_2=O_p(n^{-1/2})$ and $\|\bg^*-\bg_0\|_2=O_p(n^{-1/2}m^{1/2})$, 
similar arguments to \eqref{eq:R11}, \eqref{eq:R22}, and
\eqref{eq:R12} leads to 
$\|\R_{11}\|_2=o_p(1)$,
$\|\R_{22}\|_2=o_p(m^{-1})$ and  $\|\R_{12}\|_2=o_p(m^{-1/2})$.
Hence, in terms of their 2-norms,
\bse
&&-\bSig^{-1}\begin{bmatrix}\R_{11}&\R_{12}\\\R_{12}\trans&\R_{22} \end{bmatrix}
\begin{bmatrix}\wh\bb-\bb_0 \\ \wh\bg-\bg_0\end{bmatrix}\\
&=&\begin{bmatrix}O(1)&O(m^{1/2})\\O(m^{1/2})&O(m)\end{bmatrix}
\begin{bmatrix}o_p(1)&o_p(m^{-1/2})\\o_p(m^{-1/2})&o_p(m^{-1})\end{bmatrix}
\begin{bmatrix}O_p(n^{-1/2}) \\ O_p(n^{-1/2}m^{1/2})\end{bmatrix}\\
&=&\begin{bmatrix}o_p(n^{-1/2})\\o_p(n^{-1/2}m^{1/2})\end{bmatrix}.
\ese
Therefore we get
\bse
\begin{bmatrix}\wh\bb-\bb_0 \\ \wh\bg-\bg_0\end{bmatrix}
=
\bSig^{-1}n^{-1}\sumi\begin{bmatrix}\x_i\{y_i-E(Y\mid\x_i)\} \\ \B(y_i)-E\{\B(Y)\mid\x_i\}\end{bmatrix}
+\begin{bmatrix}\r_1 \\ \r_2\end{bmatrix},
\ese
where $\|\r_1\|_2=o_p(n^{-1/2})$ and $\|\r_2\|_2=o_p(n^{-1/2}m^{1/2})$,
which proves the second result of Proposition \ref{pro:catbbbg}.
Also, we can express $\wh\bb-\bb_0$ as
\bse
\wh\bb-\bb_0
&=&\begin{bmatrix}\I_p&\0_m\end{bmatrix}
\bSig^{-1}n^{-1}\sumi\begin{bmatrix}\x_i\{y_i-E(Y\mid\x_i)\} \\ \B(y_i)-E\{\B(Y)\mid\x_i\}\end{bmatrix}
+\r_1.
\ese
Then since
\bse
E\left(\begin{bmatrix}\X\{Y-E(Y\mid\X)\} \\
\B(Y)-E\{\B(Y)\mid\X\}\end{bmatrix}^{\otimes2}\right)=\bSig,
\ese
using \eqref{eq:bsiginv} we can conclude 
$\bSig_\bb^{-1/2}\sqrt{n}(\wh\bb-\bb_0)\to N(\0_p,\I_p)$ in distribution as $n\to\infty$.
\qed

\subsection{Proof of Theorem \ref{th:catxi}} \label{proof:catxi}

We can express $\wh\bxi-\bxi_0$ as
\be\label{eq:catbxi}
\wh\bxi-\bxi_0
&=&(\wh\bb-\bb_0)n^{-1}\sumi\var(Y\mid\x_i,\wh\btheta)
+\bb_0n^{-1}\sumi\{\var(Y\mid\x_i,\wh\btheta)-\var(Y\mid\x_i)\}\n\\
&&+\bb_0\left[n^{-1}\sumi\var(Y\mid\x_i)-E\{\var(Y\mid\X)\}\right]\n\\
&=&\left\{n^{-1}\sumi\var(Y\mid\x_i,\wh\btheta)\I
+\bb_0n^{-1}\sumi\frac{\partial\var(Y\mid\x_i,\btheta^*)}{\partial\bb\trans}\right\}(\wh\bb-\bb_0)\n\\
&&+\bb_0n^{-1}\sumi\frac{\partial\var(Y\mid\x_i,\btheta^*)}{\partial\bg\trans}(\wh\bg-\bg_0)\n\\
&&+\bb_0\left[n^{-1}\sumi\var(Y\mid\x_i)-E\{\var(Y\mid\X)\}\right],
\ee
where $\btheta^*$ is on the line connecting $\wh\btheta$ and $\btheta_0$.
Since $\|\wh\btheta-\btheta_0\|_2=O_p(n^{-1/2}m^{1/2})$ by Proposition \ref{pro:catbbbg}
and $\var(Y\mid\x_i,\btheta)$ is Lipschitz continuous at $\btheta_0$ uniformly with respect to $\x$
by Conditions \ref{con:catbdd} and \ref{con:catlipschitz},
it is easy to see that
\be\label{eq:bxi1}
n^{-1}\sumi\var(Y\mid\x_i,\wh\btheta)
&=&n^{-1}\sumi\var(Y\mid\x_i)+O_p(n^{-1/2}m^{1/2})\n\\
&=&E\{\var(Y\mid\X)\}+O_p(n^{-1/2}m^{1/2}).
\ee
Similarly, Conditions \ref{con:catbdd} and \ref{con:catlipschitz} lead to
\be\label{eq:bxi2}
&&\left\|n^{-1}\sumi\frac{\partial\var(Y\mid\x_i,\btheta^*)}{\partial\bb}
-E\left[\X\{Y-E(Y\mid\X)\}^3\right]\right\|_2\n\\
&=&\left\|n^{-1}\sumi\x_i E[\{Y-E(Y\mid\x_i,\btheta^*)\}^3\mid\x_i,\btheta^*]
-E\left[\X\{Y-E(Y\mid\X)\}^3\right]\right\|_2\n\\
&=&O_p(n^{-1/2}m^{1/2}).
\ee 
In addition, in terms of the 2-norm, we have
\bse
&&n^{-1}\sumi\frac{\partial\var(Y\mid\x_i,\btheta^*)}{\partial\bg}\\
&=&n^{-1}\sumi \left[E\{\B(Y)Y^2\mid\x_i,\btheta^*\}
-2E\{\B(Y)Y\mid\x_i,\btheta^*\}E(Y\mid\x_i,\btheta^*)\right.\\
&&\left.+2E\{\B(Y)\mid\x_i,\btheta^*\}E(Y\mid\x_i,\btheta^*)^2
-E\{\B(Y)\mid\x_i,\btheta^*\}E(Y^2\mid\x_i,\btheta^*)\right]\\
&=&n^{-1}\sumi\left\{\p_2(\x_i,\btheta^*)-2\p_1(\x_i,\btheta^*)E(Y\mid\x_i)
+2\p(\x_i,\btheta^*)E(Y\mid\x_i)^2-\p(\x_i,\btheta^*)E(Y^2\mid\x_i)\right\}\\
&&+O_p(n^{-1/2}m^{1/2})\\
&=&n^{-1}\sumi\left\{\p_2(\x_i,\btheta_0)-2\p_1(\x_i,\btheta_0)E(Y\mid\x_i)
+2\p(\x_i,\btheta_0)E(Y\mid\x_i)^2-\p(\x_i,\btheta_0)E(Y^2\mid\x_i)\right\}\\
&&+O_p(n^{-1/2}m^{1/2})\\
&=&n^{-1}\sumi
E([\B(Y)-E\{\B(Y)\mid\x_i\}]\{Y-E(Y\mid\x_i)\}^2\mid\x_i)+O_p(n^{-1/2}m^{1/2})\\
&=&
E([\B(Y)-E\{\B(Y)\mid\X\}]\{Y-E(Y\mid\X)\}^2)\\
&&+
n^{-1}\sumi
E([\B(Y)-E\{\B(Y)\mid\x_i\}]\{Y-E(Y\mid\x_i)\}^2\mid\x_i)
-E([\B(Y)-E\{\B(Y)\mid\X\}]\{Y-E(Y\mid\X)\}^2)\\
&&+O_p(n^{-1/2}m^{1/2}).
\ese
The second equality holds by Conditions \ref{con:catbdd} and \ref{con:catlipschitz},
and the third equality holds by Lemma \ref{lem:lipschitz}.
Now let $\W_1\equiv\B(Y)-E\{\B(Y)\mid\X\}$ and $W_2\equiv\{Y-E(Y\mid\X)\}^2$,
then we can show $\|E(\W_1W_2)\|_2=O(m^{-1/2})$ since for any $\u\in\mR^{m}$ such that $\|\u\|_2=1$,
\bse
\{E(\W_1W_2)\trans\u\}^2
&=&\u\trans E\{\cov(\W_1,W_2\mid\X)\}E\{\cov(W_2,\W_1\mid\X)\}\u\n\\
&\leq&\u\trans E\{\cov(\W_1,W_2\mid\X)\cov(W_2,\W_1\mid\X)\}\u\n\\
&\leq&\u\trans E\{\var(\W_1\mid\X)\var(W_2\mid\X)\}\u\n\\
&\leq&C\u\trans E\{\var(\W_1\mid\X)\}\u\n\\
&=&O(m^{-1})
\ese
for some constant $C>0$.
The fourth argument holds since $\var(W_2\mid\x)$ is uniformly bounded
by Condition \ref{con:catbdd}, 
and the last argument holds because
$\|E\{\var(\W_1\mid\X)\}\|_2=\|\bSig_{22}\|_2\asymp m^{-1}$ by Remark
\ref{rem:var22}.
This leads to
\be\label{eq:bxi3}
\|\A_2\|_2=\|\bb_0E(\W_1W_2)\trans\|_2=\|\bb_0\|_2\|E(\W_1W_2)\|_2=O(m^{-1/2}).
\ee
Similarly,
\bse
E\left[\left\{m^{1/2}E(\W_1\trans\u W_2\mid\X)\right\}^2\right]
&=&m E\left[\left\{\cov(\W_1\trans\u,W_2\mid\X)\right\}^2\right]\\
&\leq&m E\left\{\var(\W_1\trans\u\mid\X)\var(W_2\mid\X)\right\}\\
&\leq&Cm \u\trans E\left\{\var(\W_1\mid\X)\right\}\u\\
&=&O(1),
\ese
i.e. the second moment of $m^{1/2}E(\W_1\trans\u W_2\mid\X)$ is finite, then
\bse
&&\left\|n^{-1}\sumi E(\W_1W_2\mid\x_i)-E(\W_1W_2)\right\|_2\\
&=&m^{-1/2}\sup_{\u\in\R^m:\|\u\|_2=1}
\left|n^{-1}\sumi m^{1/2}E(\W_1\trans\u W_2\mid\x_i)-m^{1/2}E(\W_1\trans\u W_2)\right|\\
&=&O_p(m^{-1/2}n^{-1/2})\\
&=&o_p(m^{-1/2}).
\ese
Thus using $n^{-1/2}m^{1/2}=o(m^{-1/2})$ under Condition \ref{con:catm}, we get
\be\label{eq:bxi4}
&&\left\|n^{-1}\sumi\frac{\partial\var(Y\mid\x_i,\btheta^*)}{\partial\bg}
-E([\B(Y)-E\{\B(Y)\mid\X\}]\{Y-E(Y\mid\X)\}^2)\right\|_2\n\\
&=&\left\|
n^{-1}\sumi
E(\W_1W_2\mid\x_i)
-E(\W_1W_2)\right\|_2+O_p(n^{-1/2}m^{1/2})\n\\
&=&o_p(m^{-1/2}).
\ee
Combining the results of Proposition \ref{pro:catbbbg},
\eqref{eq:catbxi}, \eqref{eq:bxi1}, \eqref{eq:bxi2}, \eqref{eq:bxi3},
and \eqref{eq:bxi4},
in terms of the 2-norm we get
\bse
\wh\bxi-\bxi_0
&=&\A_1(\wh\bb-\bb_0)+\A_2(\wh\bg-\bg_0)+\bb_0\left[n^{-1}\sumi\var(Y\mid\x_i)-E\{\var(Y\mid\X)\}\right]\\
&&+O_p(n^{-1}m^{1/2})+o_p(n^{-1/2})\\
&=&\A\bSig^{-1} n^{-1}\sumi
\begin{bmatrix} \x_i\{y_i-E(Y\mid\x_i)\} \\ \B(y_i)-E\{\B(Y)\mid\x_i\} \end{bmatrix}
+\bb_0\left[n^{-1}\sumi\var(Y\mid\x_i)-E\{\var(Y\mid\X)\}\right]\\
&&+\A_1\r_1+\A_2\r_2+o_p(n^{-1/2})\\
&=&\A\bSig^{-1} n^{-1}\sumi
\begin{bmatrix} \x_i\{y_i-E(Y\mid\x_i)\} \\ \B(y_i)-E\{\B(Y)\mid\x_i\} \end{bmatrix}
+\bb_0\left[n^{-1}\sumi\var(Y\mid\x_i)-E\{\var(Y\mid\X)\}\right]\\
&&+o_p(n^{-1/2})
\ese
since $m=o(n^{1/3})$ by Condition \ref{con:catm}.
Since
\bse
E\left(\begin{bmatrix}\X\{Y-E(Y\mid\X)\} \\
\B(Y)-E\{\B(Y)\mid\X\}\end{bmatrix}^{\otimes2}\right)&=&\bSig,\\
\cov[\X\{Y-E(Y\mid\X)\},\var(Y\mid\X)]&=&\0,\\
\cov[\B(Y)-E\{\B(Y)\mid\X\},\var(Y\mid\X)]&=&\0,
\ese
we can conclude
$\bSig_\bxi^{-1/2}\sqrt{n}(\wh\bxi-\bxi_0)\to N(\0,\I)$ in distribution as $n\to\infty$.
\qed

\subsection{Proof of Proposition \ref{pro:cateffbb}} \label{proof:cateffbb}

The efficient score for $\bb$ given in Appendix \ref{proof:effbb} is
$\S\eff=y\x-\a_0(y)-E\{Y\x-\a_0(Y)\mid\x\}$, where 
$\a_0(y)$ satisfies 
\be\label{eq:cata0}
\a_0(y)-E[E\{\a_0(Y)\mid\X\}\mid y]=
E(y\X\mid y)-E\{E(Y\X\mid\X)\mid y\}. 
\ee 
Thus, the efficient variance is $\{E(\S\eff^{\otimes2})\}^{-1}$. 

To show that the MLE estimator $\wh\bb$ in Proposition \ref{pro:catbbbg} is efficient, we
need to show 
$\bSig_\bb^{-1}= E(\S\eff^{\otimes2})$.
$\a_0(y)$ must be of the general form $\Lambda\B(y)$,
  i.e. $\a_0(y)=\Lambda\B(y)$,
  where $\Lambda$ is a $p\times m$ 
coefficient matrix. Then (\ref{eq:cata0}) implies that 
\bse 
\Lambda E\{\B(Y)^{\otimes2}\}-\Lambda E[E\{\B(Y)\mid\X\}^{\otimes2}]
&=&
E[E\{Y\X-E(Y\X\mid\X)\B(Y)\trans\mid\X\}]\\
&=&E[\X\cov\{Y,\B(Y)\mid\X\}], 
\ese 
i.e.,
\bse 
\a_0(y)=\Lambda\B(y)=\bSig_{12}\bSig_{22}^{-1}\B(y).
\ese 
Therefore,
\bse 
E(\S\eff^{\otimes2})
=E\left\{(
Y\X-E(Y\X\mid\X) 
-
\bSig_{12}\bSig_{22}^{-1}[\B(Y)-E\{\B(Y)\mid\X\}])^{\otimes2}\right\}
=\bSig_\bb^{-1}. 
\ese 
\qed

\subsection{Proof of Theorem \ref{th:cateffxi}} \label{proof:cateffxi}

Since the efficient score for $\bxi$ given in Section \ref{sec:xibound} leads to
\bse 
E(\bphi\eff^{\otimes2})&=&E\left(
[\bb v(\bb\trans\X)-\bb E\{v(\bb\trans\X)\}]^{\otimes2}\right)\\
&&+E\left(
[\bb Y^2+\a(Y) +\M\X Y 
-E\{\bb Y^2+\a(Y)+\M\X Y\mid\X\}]^{\otimes2}\right), 
\ese 
and the asymptotic variance of $\wh\bxi$ in Theorem \ref{th:catxi} is
\bse 
\bSig_\bxi=\A\bSig^{-1}\A\trans+\bb^{\otimes2}\var\{v(\bb\trans\X)\}, 
\ese 
we only need to show 
\bse 
E\left(
[\bb Y^2+\a(Y) +\M\X Y 
-E\{\bb Y^2+\a(Y)+\M\X Y\mid\X\}]^{\otimes2}\right)=
\A\bSig^{-1}\A\trans. 
\ese 
Now we can always write $\bb y^2+\a(y)=\bLam\B(y)$, where $\bLam\in
\mR^{p\times m}$. Then (\ref{eq:a}) and (\ref{eq:M}) imply 
\bse 
&&-\bLam E[\B(y) -E\{\B(Y) \mid \X\}\mid y]
+\bb E\{ y^2 -E(Y^2 \mid \X)\mid y\}\n\\
&=&2\bb E[ yE(Y\mid\X)-
    E\{YE(Y\mid\X) \mid \X\}\mid y]
+\M E[\X\{ y - E( Y \mid \X)\}\mid y]\\
&=&2\bb E[ yE(Y\mid\X)-
    E\{YE(Y\mid\X) \mid \X\}\mid y]\\
&&+\left(E\{v(\bb\trans\X)\}\I 
-E\left[2\bb\X\trans Y v(\bb\trans\X) +\{\bLam\B(Y) -\bb Y^2\}
\{Y-E(Y\mid\X)\}\X\trans 
\right]\right)\\
&&\times \bSig_{11}^{-1}
E[\X\{ y - E( Y \mid \X)\}\mid y]. 
\ese 
Multiplying $\B\trans(y)$ on both sides and taking expectation lead
to 
\bse 
&&-\bLam \bSig_{22}
+\bb E[\{ Y^2 -E(Y^2 \mid \X)\}\B\trans(Y)]\\
&=&2\bb E[ YE(Y\mid\X) \B\trans(Y)-
    \{E(Y\mid\X)\}^2\B\trans(Y)]\\
&&+\left(E\{v(\bb\trans\X)\}\I 
-E\left[2\bb\X\trans Y v(\bb\trans\X) +\{\bLam\B(Y)-\bb Y^2\}
\{Y-E(Y\mid\bb\trans\X)\}\X\trans 
\right]\right)\\
&&\times \bSig_{11}^{-1}
E[\X\{ Y - E( Y \mid \X)\}\B\trans(Y)]\\
&=&2\bb E[Y\cov\{Y,\B(Y)\mid\X\}]
+E\{v(\bb\trans\X)\}\bSig_{11}^{-1}\bSig_{12}
-E\left\{2\bb\X\trans Y v(\bb\trans\X) 
\right\}\bSig_{11}^{-1}\bSig_{12}\\
&&-\bLam \bSig_{21}\bSig_{11}^{-1}\bSig_{12}
+\bb E[ Y^2 
\{Y-E(Y\mid\bb\trans\X)\}\X\trans]\bSig_{11}^{-1}\bSig_{12}\\
&=&2\bb E[Y\cov\{Y,\B(Y)\mid\X\}]+\A_1 \bSig_{11}^{-1}\bSig_{12}
-\bLam \bSig_{21}\bSig_{11}^{-1}\bSig_{12}, 
\ese 
where $\A_1, \A_2$ are defined in Theorem \ref{th:catxi},
hence 
\bse 
&&-\bLam \bSig_{22}+\bLam \bSig_{21}\bSig_{11}^{-1}\bSig_{12}\\
&=&2\bb E[Y\cov\{Y,\B(Y)\mid\X\}]
+\A_1\bSig_{11}^{-1}\bSig_{12}
-\bb E[\{ Y^2 -E(Y^2 \mid \X)\}\B\trans(Y)]\\
&=&\A_1\bSig_{11}^{-1}\bSig_{12}-\A_2, 
\ese 
and 
\bse 
\bb y^2+\a(y)=(\A_1\bSig_{11}^{-1}\bSig_{12}-\A_2) 
(\bSig_{21}\bSig_{11}^{-1}\bSig_{12}-\bSig_{22})^{-1}\B(y)
=\U \B(y), 
\ese 
where, for notational brevity, 
\bse 
\U&\equiv&(\A_1\bSig_{11}^{-1}\bSig_{12}-\A_2) 
(\bSig_{21}\bSig_{11}^{-1}\bSig_{12}-\bSig_{22})^{-1}. 
\ese 
Then 
\bse 
\M&=&
\left(E\{v(\bb\trans\X)\}\I 
-E\{2\bb\X\trans Y v(\bb\trans\X)\}
 -\U  E\left[\B(Y)\{Y-E(Y\mid\bb\trans\X)\}\X\trans 
\right]\right.\\
&&\left. +\bb  E\left[Y^2\{Y-E(Y\mid\bb\trans\X)\}\X\trans 
\right]
\right) \bSig_{11}^{-1}\\
&=&\left(E\{v(\bb\trans\X)\}\I 
-E\{2\bb\X\trans Y v(\bb\trans\X)\}
 -\U \bSig_{21}+\bb  E\left[Y^2\{Y-E(Y\mid\X)\}\X\trans\right]
 \right)\bSig_{11}^{-1}\\
&=&(\A_1 -\U \bSig_{21})\bSig_{11}^{-1}. 
\ese 
Hence 
\bse 
\bb y^2+\a(y)-E\{\bb Y^2+\a(Y)\mid\x\}
=\U[
\B(y)-E\{\B(Y)\mid\x\}], 
\ese 
and 
\bse 
E\{\var(\M\X Y\mid\X)\}
&=&\M \bSig_{11}\M\trans,\\
E[\var\{\bb Y^2+\a(Y)\mid\X\}]
&=&\U\bSig_{22}\U\trans,\\
E[\cov\{\M\X Y,\bb Y^2+\a(Y)\mid\X\}]
&=&\M \bSig_{12}\U\trans. 
\ese 
Thus, noting that $\A=[\A_1, \A_2]$, 
\bse 
&&E\left(
[\bb Y^2+\a(Y) +\M\X Y 
-E\{\bb Y^2+\a(Y)+\M\X Y\mid\X\}]^{\otimes2}\right)\\
&=&\M\bSig_{11}\M\trans+
\U\bSig_{22}\U\trans+
\M \bSig_{12}\U\trans 
+\U\bSig_{21}\M\trans \\
&=&(\A_1 -\U \bSig_{21})\bSig_{11}^{-1}(\A_1\trans 
-\bSig_{12}\U\trans) +\U\bSig_{22}\U\trans+
(\A_1 -\U \bSig_{21})\bSig_{11}^{-1}\bSig_{12}\U\trans\\
&&+
\U\bSig_{21}\bSig_{11}^{-1}(\A_1\trans 
-\bSig_{12}\U\trans)\\
&=&\A_1 \bSig_{11}^{-1}\A_1\trans 
+\U(\bSig_{22}
 -
\bSig_{21}\bSig_{11}^{-1}\bSig_{12})\U\trans\\
&=&
\A_1\bSig_{11}^{-1}\A_1\trans 
+
(\A_1\bSig_{11}^{-1}\bSig_{12}-\A_2) 
(\bSig_{22}-\bSig_{21}\bSig_{11}^{-1}\bSig_{12})^{-1}
(\bSig_{21}\bSig_{11}^{-1}\A_1\trans-\A_2\trans) \\
&\\
&=&\A 
\left(\begin{array}{cc}
\bSig_{11}^{-1}+\bSig_{11}^{-1}\bSig_{12}(\bSig_{22}-\bSig_{21}\bSig_{11}^{-1}\bSig_{12})^{-1}\bSig_{21}\bSig_{11}^{-1}&
                                                                                                          -\bSig_{11}^{-1}\bSig_{12}(\bSig_{22}-\bSig_{21}\bSig_{11}^{-1}\bSig_{12})^{-1}\\
-(\bSig_{22}-\bSig_{21}\bSig_{11}^{-1}\bSig_{12})^{-1}\bSig_{21}\bSig_{11}^{-1}&
(\bSig_{22}-\bSig_{21}\bSig_{11}^{-1}\bSig_{12})^{-1}
\end{array}\right)\A\trans\\
&&\\
&=&\A\bSig^{-1}\A\trans. 
\ese 
\qed

\subsection{Additional Tables for Simulation Experiments}

\bgroup
\def\arraystretch{0.95}
\begin{table}[hp]
    \centering
    \caption{$\bta_\tau$ estimation results under the truncated normal distribution.}
    \label{tab:normaleta}
    \vspace{0.5em}
    \begin{tabular}{c|cl|cccc}
    $\tau$ & & Method & $|$bias$|$ & $\sigma_{\text{sim}}$ & $\wh\sigma_{\text{est}}$ & C.I. \\ \hline
    \multirow{9}{*}{0.05} & $\eta_1$ & aMLE & 0.022 & 0.028 & 0.029 & 0.954 \\
                          &          & pMLE & 0.026 & 0.029 & -     & -     \\
                          &          & MLE  & 0.035 & 0.035 & 0.035 & 0.893 \\
                          & $\eta_2$ & aMLE & 0.029 & 0.036 & 0.034 & 0.937 \\
                          &          & pMLE & 0.036 & 0.035 & -     & -     \\
                          &          & MLE  & 0.053 & 0.041 & 0.036 & 0.681 \\
                          & $\eta_3$ & aMLE & 0.033 & 0.041 & 0.042 & 0.957 \\
                          &          & pMLE & 0.047 & 0.040 & -     & -     \\
                          &          & MLE  & 0.075 & 0.044 & 0.035 & 0.463 \\ \hline
    \multirow{9}{*}{0.25} & $\eta_1$ & aMLE & 0.023 & 0.029 & 0.029 & 0.950 \\
                          &          & pMLE & 0.023 & 0.029 & -     & -     \\
                          &          & MLE  & 0.028 & 0.035 & 0.035 & 0.954 \\
                          & $\eta_2$ & aMLE & 0.028 & 0.036 & 0.034 & 0.941 \\
                          &          & pMLE & 0.028 & 0.035 & -     & -     \\
                          &          & MLE  & 0.035 & 0.041 & 0.036 & 0.898 \\
                          & $\eta_3$ & aMLE & 0.032 & 0.040 & 0.040 & 0.953 \\
                          &          & pMLE & 0.032 & 0.039 & -     & -     \\
                          &          & MLE  & 0.040 & 0.044 & 0.035 & 0.839 \\ \hline
    \multirow{9}{*}{0.50} & $\eta_1$ & aMLE & 0.023 & 0.029 & 0.029 & 0.945 \\
                          &          & pMLE & 0.023 & 0.030 & -     & -     \\
                          &          & MLE  & 0.029 & 0.035 & 0.035 & 0.940 \\
                          & $\eta_2$ & aMLE & 0.028 & 0.036 & 0.034 & 0.945 \\
                          &          & pMLE & 0.029 & 0.035 & -     & -     \\
                          &          & MLE  & 0.039 & 0.041 & 0.036 & 0.841 \\
                          & $\eta_3$ & aMLE & 0.032 & 0.040 & 0.040 & 0.957 \\
                          &          & pMLE & 0.033 & 0.039 & -     & -     \\
                          &          & MLE  & 0.049 & 0.044 & 0.035 & 0.742 \\ \hline
    \multirow{9}{*}{0.75} & $\eta_1$ & aMLE & 0.023 & 0.029 & 0.029 & 0.941 \\
                          &          & pMLE & 0.023 & 0.029 & -     & -     \\
                          &          & MLE  & 0.028 & 0.035 & 0.035 & 0.953 \\
                          & $\eta_2$ & aMLE & 0.028 & 0.036 & 0.034 & 0.940 \\
                          &          & pMLE & 0.028 & 0.034 & -     & -     \\
                          &          & MLE  & 0.035 & 0.041 & 0.036 & 0.898 \\
                          & $\eta_3$ & aMLE & 0.033 & 0.040 & 0.040 & 0.959 \\
                          &          & pMLE & 0.032 & 0.039 & -     & -     \\
                          &          & MLE  & 0.039 & 0.044 & 0.035 & 0.846 \\ \hline
    \multirow{9}{*}{0.95} & $\eta_1$ & aMLE & 0.023 & 0.028 & 0.029 & 0.940 \\
                          &          & pMLE & 0.025 & 0.029 & -     & -     \\
                          &          & MLE  & 0.036 & 0.035 & 0.035 & 0.893 \\
                          & $\eta_2$ & aMLE & 0.028 & 0.035 & 0.034 & 0.945 \\
                          &          & pMLE & 0.032 & 0.034 & -     & -     \\
                          &          & MLE  & 0.053 & 0.041 & 0.036 & 0.678 \\
                          & $\eta_3$ & aMLE & 0.033 & 0.041 & 0.042 & 0.964 \\
                          &          & pMLE & 0.041 & 0.038 & -     & -     \\
                          &          & MLE  & 0.075 & 0.044 & 0.035 & 0.458
    \end{tabular}
\end{table}
\egroup

\bgroup
\def\arraystretch{0.95}
\begin{table}[hp]
    \centering
    \caption{$\bta_\tau$ estimation results under the normal distribution.}
    \label{tab:normal2eta}
    \vspace{0.5em}
    \begin{tabular}{c|cl|cccc}
    $\tau$ & & Method & $|$bias$|$ & $\sigma_{\text{sim}}$ & $\wh\sigma_{\text{est}}$ & C.I. \\ \hline
    \multirow{9}{*}{0.05} & $\eta_1$ & aMLE & 0.026 & 0.033 & 0.032 & 0.947 \\
                          &          & pMLE & 0.030 & 0.039 & -     & -     \\
                          &          & MLE  & 0.026 & 0.032 & 0.032 & 0.954 \\
                          & $\eta_2$ & aMLE & 0.026 & 0.032 & 0.032 & 0.948 \\
                          &          & pMLE & 0.031 & 0.039 & -     & -     \\
                          &          & MLE  & 0.025 & 0.032 & 0.032 & 0.957 \\
                          & $\eta_3$ & aMLE & 0.026 & 0.033 & 0.033 & 0.934 \\
                          &          & pMLE & 0.032 & 0.039 & -     & -     \\
                          &          & MLE  & 0.025 & 0.032 & 0.032 & 0.941 \\ \hline
    \multirow{9}{*}{0.25} & $\eta_1$ & aMLE & 0.026 & 0.032 & 0.032 & 0.950 \\
                          &          & pMLE & 0.030 & 0.038 & -     & -     \\
                          &          & MLE  & 0.026 & 0.032 & 0.032 & 0.954 \\
                          & $\eta_2$ & aMLE & 0.025 & 0.032 & 0.032 & 0.950 \\
                          &          & pMLE & 0.030 & 0.039 & -     & -     \\
                          &          & MLE  & 0.025 & 0.032 & 0.032 & 0.957 \\
                          & $\eta_3$ & aMLE & 0.025 & 0.032 & 0.032 & 0.940 \\
                          &          & pMLE & 0.031 & 0.038 & -     & -     \\
                          &          & MLE  & 0.025 & 0.032 & 0.032 & 0.941 \\ \hline
    \multirow{9}{*}{0.50} & $\eta_1$ & aMLE & 0.026 & 0.032 & 0.032 & 0.953 \\
                          &          & pMLE & 0.029 & 0.038 & -     & -     \\
                          &          & MLE  & 0.026 & 0.032 & 0.032 & 0.954 \\
                          & $\eta_2$ & aMLE & 0.025 & 0.032 & 0.032 & 0.955 \\
                          &          & pMLE & 0.030 & 0.039 & -     & -     \\
                          &          & MLE  & 0.025 & 0.032 & 0.032 & 0.957 \\
                          & $\eta_3$ & aMLE & 0.025 & 0.032 & 0.032 & 0.935 \\
                          &          & pMLE & 0.030 & 0.038 & -     & -     \\
                          &          & MLE  & 0.025 & 0.032 & 0.032 & 0.941 \\ \hline
    \multirow{9}{*}{0.75} & $\eta_1$ & aMLE & 0.026 & 0.032 & 0.032 & 0.953 \\
                          &          & pMLE & 0.029 & 0.038 & -     & -     \\
                          &          & MLE  & 0.026 & 0.032 & 0.032 & 0.954 \\
                          & $\eta_2$ & aMLE & 0.025 & 0.032 & 0.032 & 0.951 \\
                          &          & pMLE & 0.030 & 0.038 & -     & -     \\
                          &          & MLE  & 0.025 & 0.032 & 0.032 & 0.957 \\
                          & $\eta_3$ & aMLE & 0.026 & 0.032 & 0.032 & 0.938 \\
                          &          & pMLE & 0.030 & 0.038 & -     & -     \\
                          &          & MLE  & 0.025 & 0.032 & 0.032 & 0.941 \\ \hline
    \multirow{9}{*}{0.95} & $\eta_1$ & aMLE & 0.026 & 0.032 & 0.032 & 0.950 \\
                          &          & pMLE & 0.029 & 0.038 & -     & -     \\
                          &          & MLE  & 0.026 & 0.032 & 0.032 & 0.954 \\
                          & $\eta_2$ & aMLE & 0.026 & 0.032 & 0.032 & 0.951 \\
                          &          & pMLE & 0.031 & 0.039 & -     & -     \\
                          &          & MLE  & 0.025 & 0.032 & 0.032 & 0.957 \\
                          & $\eta_3$ & aMLE & 0.027 & 0.034 & 0.033 & 0.937 \\
                          &          & pMLE & 0.033 & 0.039 & -     & -     \\
                          &          & MLE  & 0.025 & 0.032 & 0.032 & 0.941
    \end{tabular}
\end{table}
\egroup

\begin{table}[hp]
    \centering
    \caption{$\bta_\tau$ estimation results under the gamma distribution.}
    \label{tab:gamma2eta}
    \vspace{0.5em}
    \begin{tabular}{c|cl|cccc}
    $\tau$ & & Method & $|$bias$|$ & $\sigma_{\rm sim}$ & $\wh\sigma_{\rm est}$ & C.I. \\ \hline
    \multirow{6}{*}{0.05} & $\eta_1$ & aMLE & 0.031 & 0.040 & 0.040 & 0.949 \\
                          &          & pMLE & 0.561 & 0.690 & -     & -     \\
                          &          & MLE  & 0.021 & 0.026 & 0.025 & 0.950 \\
                          & $\eta_2$ & aMLE & 0.038 & 0.049 & 0.049 & 0.939 \\
                          &          & pMLE & 0.840 & 0.588 & -     & -     \\
                          &          & MLE  & 0.022 & 0.028 & 0.028 & 0.950 \\ \hline
    \multirow{6}{*}{0.25} & $\eta_1$ & aMLE & 0.051 & 0.066 & 0.064 & 0.942 \\
                          &          & pMLE & 1.070 & 0.953 & -     & -     \\
                          &          & MLE  & 0.035 & 0.044 & 0.043 & 0.949 \\
                          & $\eta_2$ & aMLE & 0.058 & 0.074 & 0.073 & 0.943 \\
                          &          & pMLE & 1.666 & 0.753 & -     & -     \\
                          &          & MLE  & 0.037 & 0.046 & 0.046 & 0.953 \\ \hline
    \multirow{6}{*}{0.50} & $\eta_1$ & aMLE & 0.070 & 0.089 & 0.088 & 0.944 \\
                          &          & pMLE & 1.259 & 1.080 & -     & -     \\
                          &          & MLE  & 0.048 & 0.060 & 0.059 & 0.948 \\
                          & $\eta_2$ & aMLE & 0.077 & 0.096 & 0.096 & 0.950 \\
                          &          & pMLE & 1.945 & 0.844 & -     & -     \\
                          &          & MLE  & 0.051 & 0.064 & 0.064 & 0.950 \\ \hline
    \multirow{6}{*}{0.75} & $\eta_1$ & aMLE & 0.093 & 0.117 & 0.116 & 0.944 \\
                          &          & pMLE & 1.246 & 1.410 & -     & -     \\
                          &          & MLE  & 0.064 & 0.081 & 0.080 & 0.948 \\
                          & $\eta_2$ & aMLE & 0.103 & 0.127 & 0.126 & 0.944 \\
                          &          & pMLE & 1.831 & 0.907 & -     & -     \\
                          &          & MLE  & 0.069 & 0.086 & 0.086 & 0.951 \\ \hline
    \multirow{6}{*}{0.95} & $\eta_1$ & aMLE & 0.146 & 0.186 & 0.185 & 0.948 \\
                          &          & pMLE & 0.828 & 1.529 & -     & -     \\
                          &          & MLE  & 0.094 & 0.118 & 0.116 & 0.949 \\
                          & $\eta_2$ & aMLE & 0.178 & 0.227 & 0.232 & 0.957 \\
                          &          & pMLE & 0.991 & 1.025 & -     & -     \\
                          &          & MLE  & 0.102 & 0.127 & 0.126 & 0.948
    \end{tabular}
\end{table}

\end{document}